\documentclass[%
 amsmath,amssymb,
longbibliography, % for titles in references
aps,
prx,
10pt
]{revtex4-2}

\usepackage{tikz}
\usepackage{graphicx}% Include figure files
\usepackage{dcolumn}% Align table columns on decimal point
\usepackage{bm}% bold math
\usepackage{lipsum}
\usepackage{amsthm}
\usepackage{physics}
\usepackage[colorlinks=true,linkcolor=black,anchorcolor=black,citecolor=black,filecolor=black,menucolor=black,runcolor=black,urlcolor=black]{hyperref}
\usepackage[inline]{enumitem}
\usepackage[english]{babel}

\usepackage{makecell}
\usepackage{mathtools}
\usepackage{pst-eucl}
\usepackage{qcircuit}

\usepackage{algorithm}
\usepackage{algpseudocode}
\usepackage{subcaption}
\usepackage{colortbl}

\newtheorem{theorem}{Theorem}
\newtheorem*{theorem*}{Theorem}
\newtheorem{definition}[theorem]{Definition}
\newtheorem*{definition*}{Definition}

\newtheorem*{proposition*}{Proposition}
\newtheorem{corollary}[theorem]{Corollary}
\newtheorem*{corollary*}{Corollary}
\newtheorem{lemma}[theorem]{Lemma}
\newtheorem*{lemma*}{Lemma}

\newtheorem*{claim*}{Claim}

\newcommand{\polylog}{\mathrm{polylog}}

\newcommand{\R}{\mathbb{R}}

\newcommand{\N}{\mathbb{N}}
\newcommand{\C}{\mathbb{C}}
\newcommand{\I}{\mathbb{I}}
\newcommand{\Ot}{\widetilde{O}}
\DeclareMathOperator*{\argmax}{arg\,max}
\DeclareMathOperator*{\argmin}{arg\,min}

\newcommand{\ceil}[1]{\lceil #1 \rceil}

\newcommand{\defeq}{\vcentcolon=}
\newcommand{\eqdef}{=\vcentcolon}

\newcommand{\quoted}[1]{``#1''}
\newcommand{\inner}[2]{\langle #1 , \; #2 \rangle}
\newcommand{\bracket}[2]{\langle #1 \; | \;  #2 \rangle}

\definecolor{darkgreen}{rgb}{0.0, 0.2, 0.13}
\definecolor{darklava}{rgb}{0.28, 0.24, 0.2}
\definecolor{darkolivegreen}{rgb}{0.33, 0.42, 0.18}
\definecolor{darkpastelred}{rgb}{0.76, 0.23, 0.13}
\definecolor{deepcarrotorange}{rgb}{0.91, 0.41, 0.17}
\definecolor{deepcarmine}{rgb}{0.66, 0.13, 0.24}
\definecolor{green(html/cssgreen)}{rgb}{0.0, 0.5, 0.0}

\definecolor{Arm}{rgb}{0.3,0.2,0.7}

\makeatletter
\def\blfootnote{\gdef\@thefnmark{}\@footnotetext}
\makeatother

\begin{document}

\preprint{APS/123-QED}

\title{Quantum Sparse Recovery and Quantum Orthogonal Matching Pursuit}

\author{Armando Bellante}\email{armando.bellante@mpq.mpg.de}
\affiliation{Max-Planck-Institut f\"ur Quantenoptik, Hans-Kopfermann-Str. 1, 85748 Garching, Germany}
\affiliation{Munich Center for Quantum Science and Technology (MCQST), Schellingstr. 4, 80799 M\"unchen, Germany}
\affiliation{Politecnico di Milano, DEIB, Via Ponzio 34/5 – Building 20, Milan 20133, Italy.}

\author{Stefano Vanerio}
\affiliation{Politecnico di Milano, DEIB, Via Ponzio 34/5 – Building 20, Milan 20133, Italy.}

\author{Stefano Zanero}
\affiliation{Politecnico di Milano, DEIB, Via Ponzio 34/5 – Building 20, Milan 20133, Italy.}

\date{\today}
\begin{abstract}
We study quantum sparse recovery in non-orthogonal, overcomplete dictionaries: given coherent quantum access to a state and a dictionary of vectors, the goal is to reconstruct the state up to $\ell_2$ error using as few vectors as possible. 
We first show that the general recovery problem is NP-hard, ruling out efficient exact algorithms in full generality. 
To overcome this, we introduce Quantum Orthogonal Matching Pursuit (QOMP), the first quantum analogue of the classical OMP greedy algorithm. 
QOMP combines quantum subroutines for inner product estimation, maximum finding, and block-encoded projections with an error-resetting design that avoids iteration-to-iteration error accumulation. 
Under standard mutual incoherence and well-conditioned sparsity assumptions, QOMP provably recovers the exact support of a $K$-sparse state in polynomial time. 
As an application, we give the first framework for sparse quantum tomography with non-orthogonal dictionaries in $\ell_2$ norm, achieving query complexity $\widetilde{O}(\sqrt{N}/\epsilon)$ in favorable regimes and reducing tomography to estimating only $K$ coefficients instead of $N$ amplitudes. 
In particular, for pure-state tomography with $m=O(N)$ dictionary vectors and sparsity $K=\widetilde{O}(1)$ on a well-conditioned subdictionary, this circumvents the $\widetilde{\Omega}(N/\epsilon)$ lower bound that holds in the dense, orthonormal-dictionary setting, without contradiction, by leveraging sparsity together with non-orthogonality.
Beyond tomography, we analyze QOMP in the QRAM model, where it yields polynomial speedups over classical OMP implementations, and provide a quantum algorithm to estimate the mutual incoherence of a dictionary of $m$ vectors in $O(m/\epsilon)$ queries, improving over both deterministic and quantum-inspired classical methods.
\end{abstract}

\maketitle

\tableofcontents

%%%%%%%%%%%%%%%%%%%%%%%%%%%%%%%%%%%%%%%%%%%%%%%%%%%%%%%%%%%%%%%%%%%%%%%%%%%%%%%%%%%%%%%%%%%%%
%%%%%%%%%%%%%%%%%%%%%%%%%%%%%%%%%%%%%%%%%%%%%%%%%%%%%%%%%%%%%%%%%%%%%%%%%%%%%%%%%%%%%%%%%%%%%
%%%%%%%%%%%%%%%%%%%%%%%%%%%%%%%%%%%%%%%%%%%%%%%%%%%%%%%%%%%%%%%%%%%%%%%%%%%%%%%%%%%%%%%%%%%%%
%%%%%%%%%%%%%%%%%%%%%%%%%%%%%%%%%%%%%%%%%%%%%%%%%%%%%%%%%%%%%%%%%%%%%%%%%%%%%%%%%%%%%%%%%%%%%
%%%%%%%%%%%%%%%%%%%%%%%%%%%%%%%%%%%%%%%%%%%%%%%%%%%%%%%%%%%%%%%%%%%%%%%%%%%%%%%%%%%%%%%%%%%%%
%%%%%%%%%%%%%%%%%%%%%%%%%%%%%%%%%%%%%%%%%%%%%%%%%%%%%%%%%%%%%%%%%%%%%%%%%%%%%%%%%%%%%%%%%%%%%
%%%%%%%%%%%%%%%%%%%%%%%%%%%%%%%%%%%%%%%%%%%%%%%%%%%%%%%%%%%%%%%%%%%%%%%%%%%%%%%%%%%%%%%%%%%%%
%%%%%%%%%%%%%%%%%%%%%%%%%%%%%%%%%%%%%%%%%%%%%%%%%%%%%%%%%%%%%%%%%%%%%%%%%%%%%%%%%%%%%%%%%%%%%

\section{Introduction}
\label{sec: introduction}

Quantum tomography, the task of learning a classical description of an unknown quantum state, is one of the most important problems and fundamental primitives of quantum information.
It underlies diverse areas of quantum science, from verification of quantum devices~\cite{guhne2009entanglement,eisert2020quantum}, to the design of quantum algorithms~\cite{kerenidis2019qmeans,kerenidis2020interiorpoint, qadra}, and learning-theoretic studies on quantum systems and dynamics~\cite{aaronson2007learnability,aaronson2018shadow}.
Yet tomography is notoriously costly: for dense pure states in $N$ dimensions and target $\ell_2$-error $\epsilon$, the optimal algorithms require $\widetilde{\Theta}(N/\epsilon^2)$ copies of the state~\cite{kerenidis2020quantumIP,van2023quantum}, or equivalently $\widetilde{\Theta}(N/\epsilon)$ queries to a state-preparation unitary and its inverse~\cite{van2023quantum}.
These bounds are tight, ruling out further polynomial savings for arbitrary pure states.  
One natural question is therefore:
\begin{center}
    \emph{Can additional structural promises allow us to go beyond the $\widetilde{\Theta}(N/\epsilon)$ pure state tomography barrier?}    
\end{center}

In classical signal processing, the most successful such promise is probably \emph{sparsity}.
While the Shannon–Nyquist theorem dictates that reconstructing the frequency spectrum of a signal requires sampling at twice the highest frequency~\cite{shannon1949communication, nyquist1928certain}, compressed sensing shows that signals sparse in a known dictionary can be reconstructed from far fewer measurements~\cite{candes2006robust,donoho2006compressed}.
This principle has transformed modern signal processing, leading to sparsity-based methods for magnetic resonance imaging (MRI)~\cite{lustig2007sparse}, compression formats such as JPEG~\cite{wallace1991jpeg}, denoising~\cite{elad2006image}, and anomaly detection~\cite{adler2015sparse,luo2017revisit}.
Importantly, sparsity is often realized in overcomplete, non-orthogonal dictionaries~\cite{holger2008redundant}, where the number of dictionary elements $m$ exceeds the ambient dimension $N$.
As the dictionary size grows, more signals admit very sparse descriptions; in the limit of a dictionary that spans every direction, any vector becomes $1$-sparse.
The same redundancy that enables concise representations also makes identifying the sparsest representation harder, since the search space expands and many near alternatives arise.
From a learning perspective, such dictionaries enable concise and expressive representations; from an algorithmic perspective, they pose combinatorial challenges that are NP-hard even classically~\cite{natarajan1995sparse}.
On the quantum side, structural promises have already led to major efficiency gains in tomography: low-rank structure enables compressed-sensing methods for reconstructing density matrices~\cite{gross2010quantum,kalev2015quantum}, while stabilizer or Pauli structure admits specialized algorithms for learning and certification~\cite{aaronson2018shadow,montanaro2017learning}.
By contrast, sparsity in arbitrary non-orthogonal dictionaries has remained unexplored.
Bridging this gap is the goal of the present work.
Motivated by the role of sparsity in classical signal processing, we ask:
\begin{center}
    \emph{Can sparsity in arbitrary, possibly overcomplete dictionaries be harnessed to reduce the cost of quantum tomography?}
\end{center}

\medskip
\emph{This work.}~
We introduce and study \emph{quantum sparse recovery}, the problem of reconstructing a pure state that admits a sparse representation in a dictionary, given access to state-preparation unitaries for both the state and the dictionary \emph{atoms} (elements).
This access model is strictly stronger than copy access and matches the oracle assumptions in recent tight bounds of pure-state tomography~\cite{van2023quantum}.
Our contributions are as follows: (i) we introduce and formalize the problem of quantum sparse recovery with non-orthogonal dictionaries (Definitions~\ref{def:qpo},~\ref{def:qpoeps}); (ii) we prove that quantum sparse recovery is NP-hard in general (Theorem~\ref{theorem: Quantum approximate recovery is NP-Hard}); (iii) we design and analyze the Quantum Orthogonal Matching Pursuit (QOMP) algorithm, the first stable greedy quantum method for sparse recovery in non-orthogonal, overcomplete dictionaries; (Theorem~\ref{theorem: QOMP iteration cost}) (iv) we show that QOMP achieves provable guarantees for support identification and tomography under dictionary incoherence; (Corollary~\ref{coro: qomp incoherence condition omp}) (v) we find regimes that avoid the lower bound $\Omega(N/\epsilon)$, enabling tomography with $O(\sqrt{N}/\epsilon)$ queries to the state preparation unitaries. (Theorem~\ref{theorem: sparse recovery with QOMP}, Corollary~\ref{corollary: exact sparse recovery with QOMP})

\medskip
\emph{Outline.}~
Section~\ref{sec:sparse-recovery} introduces sparse recovery and its quantum analogue, providing an overview of the results.
Section~\ref{sec: quantum sparse recovery is np hard} proves NP-hardness via a reduction from \textsc{Exact Cover by $3$-Sets (X3C)}.
Section~\ref{sec:qomp} presents QOMP and its iteration cost, Section~\ref{sec:guarantees} establishes support-recovery guarantees, and Section~\ref{sec: learning sparse quantum states} applies them to tomography.
We conclude in Section~\ref{sec:conclusions} with implications and open directions.

%%%%%%%%%%%%%%%%%%%%%%%%%%%%%%%%%%%%%%%%%%%%%%%%%%%%%%%%%%%%%%%%%%%%%%%%%%%%%%%%%%%%%%%%%%%%%
%%%%%%%%%%%%%%%%%%%%%%%%%%%%%%%%%%%%%%%%%%%%%%%%%%%%%%%%%%%%%%%%%%%%%%%%%%%%%%%%%%%%%%%%%%%%%
%%%%%%%%%%%%%%%%%%%%%%%%%%%%%%%%%%%%%%%%%%%%%%%%%%%%%%%%%%%%%%%%%%%%%%%%%%%%%%%%%%%%%%%%%%%%%
%%%%%%%%%%%%%%%%%%%%%%%%%%%%%%%%%%%%%%%%%%%%%%%%%%%%%%%%%%%%%%%%%%%%%%%%%%%%%%%%%%%%%%%%%%%%%
%%%%%%%%%%%%%%%%%%%%%%%%%%%%%%%%%%%%%%%%%%%%%%%%%%%%%%%%%%%%%%%%%%%%%%%%%%%%%%%%%%%%%%%%%%%%%
%%%%%%%%%%%%%%%%%%%%%%%%%%%%%%%%%%%%%%%%%%%%%%%%%%%%%%%%%%%%%%%%%%%%%%%%%%%%%%%%%%%%%%%%%%%%%
%%%%%%%%%%%%%%%%%%%%%%%%%%%%%%%%%%%%%%%%%%%%%%%%%%%%%%%%%%%%%%%%%%%%%%%%%%%%%%%%%%%%%%%%%%%%%
%%%%%%%%%%%%%%%%%%%%%%%%%%%%%%%%%%%%%%%%%%%%%%%%%%%%%%%%%%%%%%%%%%%%%%%%%%%%%%%%%%%%%%%%%%%%%

\section{Notation}
\label{sec: notation}
We use $n$ and $N$ interchangeably for the \emph{ambient dimension} (the dimension of the space where the target vector, or state, lives). 
When discussing quantum tomography, we typically write $N$ to emphasize the Hilbert space dimension rather than the number of qubits (e.g., for $q$ qubits, $N=2^q$).
%%%%%%%%%%%%%%%
For an integer $n \in \mathbb{N}$, we use $[n]$ to denote the set $\{0,1, \dots, n-1\} \subset \N$. 
We use the soft-$O$ notation $\widetilde{O}(\cdot)$ to suppress all polylogarithmic factors; for example, $O(n \,\mathrm{polylog}(n,\epsilon^{-1},\delta^{-1})) = \widetilde{O}(n)$. 
%%%%%%%%%%%%%%%
Whenever we say that a randomized algorithm succeeds with high probability, we mean with some fixed constant probability strictly greater than $1/2$ (e.g., at least $2/3$); standard amplification arguments (see Section~\ref{section: amplification of success probabilities}) can increase this probability arbitrarily close to~$1$.  
For vectors $\vec{a},\vec{b}$, we denote the Euclidean inner product by $\inner{\vec{a}}{\vec{b}}$ and their cosine similarity by $\bracket{\vec{a}}{\vec{b}} := \inner{\tfrac{\vec{a}}{\|\vec{a}\|}}{\tfrac{\vec{b}}{\|\vec{b}\|}}$, so that $\inner{\vec{a}}{\vec{b}} = \|\vec{a}\|\|\vec{b}\|\bracket{\vec{a}}{\vec{b}}$. Unless otherwise specified, $\|\vec{a}\| = \|\vec{a}\|_2$ denotes the Euclidean norm.  
We also use the pseudonorm $\|\vec{x}\|_0$, which counts the number of nonzero entries of $\vec{x}$.  
%%%%%%%%%%%%%%%
Let $D$ be a matrix with $m$ columns, and let $\Lambda \subseteq [m]$. We define $D_\Lambda$ as the matrix obtained from $D$ by zeroing out all columns whose indices are not in $\Lambda$. 
Its complement is denoted $D_{\overline{\Lambda}}$, so that $D = D_\Lambda + D_{\overline{\Lambda}}$. For a general matrix $A$, we write its singular value decomposition as $A = U \Sigma V^\dagger$, where $U$ and $V$ are isometries and $\Sigma$ is diagonal with strictly positive real entries (the singular values). 
The number of entries of $\Sigma$ is the rank of $A$, and we write $\sigma_{\min}(A)$ and $\sigma_{\max}(A)$ for its smallest and largest singular values, respectively. 
In general $U$ and $V$ are not unitary, but isometries with a number of columns equal to the rank of $A$. 
We use the operator norm $\|A\| := \sigma_{\max}(A)$ and the Frobenius norm $\|A\|_F := \sqrt{\sum_{i \in [n]} \sum_{j \in [m]} |A_{ij}|^2} = \sqrt{\sum_{k \in \mathrm{rank}(A)} \sigma_k^2(A)}$, where $\sigma_k(A)$ denotes the singular values of $A$.  
%%%%%%%%%%%%%%%
For a classical bit string $x \in \{0,1\}^n$, we write $\ket{x}$ for the corresponding computational basis state; for example, if $x = 010010$, then $\ket{x} = \ket{010010}$. For a real vector $\vec{x}$, we write $\ket{\vec{x}}$ to denote the amplitude encoding of the normalized vector $\vec{x}/\|\vec{x}\|$ in the computational basis; i.e., $\ket{\vec{x}} = \frac{1}{\|\vec{x}\|} \sum_{i \in [n]} x_i \ket{i}$.

%%%%%%%%%%%%%%%%%%%%%%%%%%%%%%%%%%%%%%%%%%%%%%%%%%%%%%%%%%%%%%%%%%%%%%%%%%%%%%%%%%%%%%%%%%%%%
%%%%%%%%%%%%%%%%%%%%%%%%%%%%%%%%%%%%%%%%%%%%%%%%%%%%%%%%%%%%%%%%%%%%%%%%%%%%%%%%%%%%%%%%%%%%%
%%%%%%%%%%%%%%%%%%%%%%%%%%%%%%%%%%%%%%%%%%%%%%%%%%%%%%%%%%%%%%%%%%%%%%%%%%%%%%%%%%%%%%%%%%%%%
%%%%%%%%%%%%%%%%%%%%%%%%%%%%%%%%%%%%%%%%%%%%%%%%%%%%%%%%%%%%%%%%%%%%%%%%%%%%%%%%%%%%%%%%%%%%%
%%%%%%%%%%%%%%%%%%%%%%%%%%%%%%%%%%%%%%%%%%%%%%%%%%%%%%%%%%%%%%%%%%%%%%%%%%%%%%%%%%%%%%%%%%%%%
%%%%%%%%%%%%%%%%%%%%%%%%%%%%%%%%%%%%%%%%%%%%%%%%%%%%%%%%%%%%%%%%%%%%%%%%%%%%%%%%%%%%%%%%%%%%%
%%%%%%%%%%%%%%%%%%%%%%%%%%%%%%%%%%%%%%%%%%%%%%%%%%%%%%%%%%%%%%%%%%%%%%%%%%%%%%%%%%%%%%%%%%%%%
%%%%%%%%%%%%%%%%%%%%%%%%%%%%%%%%%%%%%%%%%%%%%%%%%%%%%%%%%%%%%%%%%%%%%%%%%%%%%%%%%%%%%%%%%%%%%

\section{Sparse Recovery}
\label{sec:sparse-recovery}
Sparse recovery is the task of representing a dense high-dimensional signal as a linear combination of as few vectors as possible.
The basic ingredients are a \emph{dictionary} $D = \{\vec{d}_1,\dots,\vec{d}_m\}$ of unit vectors, called \emph{atoms}, and a target \emph{signal} $\vec{s} \in \mathbb{C}^n$.
A sparse representation consists of a coefficient vector $\vec{x}$ supported on only $K \ll n$ atoms such that $D\vec{x} = \vec{s}$ (exact recovery) or $D\vec{x} \approx \vec{s}$ (approximate recovery).
The number of nonzero coefficients, $\|\vec{x}\|_0$, quantifies the sparsity.

\begin{figure}[t]
    \centering
    \includegraphics[width=0.3\linewidth]{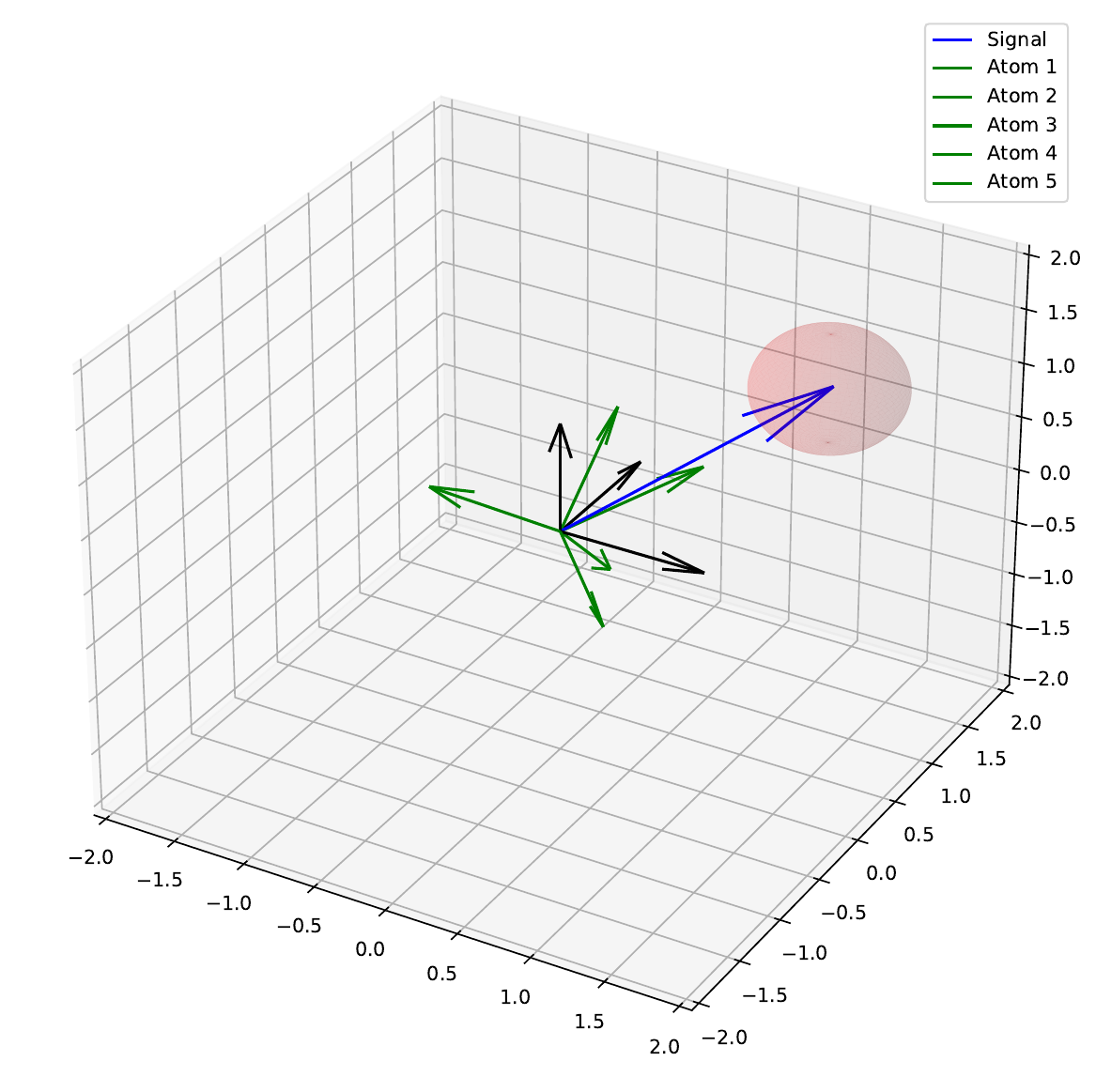}
    \caption{A sparse approximation problem in $\mathbb{R}^3$.
            The target signal (blue) is required to be reconstructed from a few atoms (green).
            Exact recovery corresponds to lying exactly on the span of the selected atoms, while approximate recovery allows $\epsilon$-error within the red ball.}
    \label{fig: sparse cube}
\end{figure}

Figure~\ref{fig: sparse cube} provides a geometric illustration in $\mathbb{R}^3$.
The green vectors are atoms from the dictionary, the blue vector is the target signal, and the red ball indicates an approximation threshold.
Exact recovery corresponds to reconstructing the blue signal from as few green atoms as possible; approximate recovery relaxes the requirement to any vector within the red ball.

Sparse recovery arises in many domains of data science and signal processing.
JPEG compression~\cite{wallace1991jpeg}, for instance, exploits that natural images are sparse in the discrete cosine transform basis; compressed sensing exploits sparsity in the Fourier domain to reduce the number of measurements needed for signal reconstruction~\cite{donoho2006compressed, candes2006robust, candes2006stable}.
In practice, sparsity is often realized not in orthogonal bases but in \emph{overcomplete, non-orthogonal, incoherent dictionaries}, where the number of atoms $m$ exceeds the ambient dimension $n$~\cite{holger2008redundant}.
This redundancy enables more flexible and compact representations but makes the sparsest support recovery problem combinatorial: one must identify the correct subset of atoms among exponentially many candidates.

Formally, the two central problems are the following.
\begin{definition}[Exact recovery, $\mathcal{P}_0$]
\label{def: p0}
    Given $\vec{s} \in \C^n$ and $D \in \C^{n\times m}$, find
    \begin{align}
        \argmin_{\vec{x} \in \C^m} \|\vec{x}\|_0 \quad\text{s.t.}\quad \vec{s} = D\vec{x}.
    \end{align}
\end{definition}

\begin{definition}[Approximate recovery, $\mathcal{P}_0^\epsilon$]
\label{def: p0eps}
    Given $\vec{s} \in \C^n$, $D \in \C^{n\times m}$, and error tolerance $\epsilon > 0$, find
    \begin{align}
        \argmin_{\vec{x} \in \C^m} \|\vec{x}\|_0 \quad\text{s.t.}\quad \|\vec{s} - D\vec{x}\|_2 \leq \epsilon.
    \end{align}
\end{definition}
Both problems are NP-hard in general~\cite{natarajan1995sparse}, with the hardness lying in identifying the optimal set of atoms spanning an exact or approximate representation of the target vector.
Nevertheless, sparse recovery is central because many real-world signals admit sparse or approximately sparse representations in natural dictionaries.
This tension between expressivity and computational tractability has motivated decades of classical algorithmic development.

Two broad strategies dominate the classical literature.
One is convex relaxation: $\ell_0$ minimization can be replaced by $\ell_1$ minimization (basis pursuit or LASSO), which under incoherence or restricted isometry conditions recovers the correct support efficiently~\cite{candes2006robust,donoho2006compressed}.
The other is greedy pursuit: algorithms such as Matching Pursuit~\cite{mallat1993matching} and Orthogonal Matching Pursuit (OMP)~\cite{pati1993orthogonal} iteratively select atoms with large correlations to the residual, refining the approximation step by step.
Despite being heuristic, these greedy algorithms come with provable polynomial time optimality guarantees under incoherence assumptions~\cite{tropp2004greed} and are widely used in applications where speed and interpretability are paramount.
Numerous refinements of OMP exist - including Regularized OMP~\cite{needell2009uniform}, CoSaMP~\cite{needell2010cosamp}, and StOMP~\cite{donoho2012sparse} - which improve robustness, stability, or scalability under extra assumptions.

On the quantum side, the landscape is far less mature.
There has been significant progress on quantum algorithms for regularized linear systems, such as ridge regression~\cite{chakraborty2022quantum} and LASSO~\cite{chen2021quantum, doriguello2025quantum}, which can sometimes act as convex surrogates for $\ell_0$ minimization.
However, these algorithms do not explicitly address sparse recovery.
Closer in spirit are greedy approaches: Quantum Matching Pursuit (QMP)~\cite{bellante2022mp} introduced a quantum analogue of the classical MP algorithm.
Yet, QMP relies on QRAM access to both the signal and its residual, effectively assuming that the target vector is available in classical memory.
This limitation makes QMP inapplicable to inherently quantum tasks such as tomography, where one only has oracle access to the state-preparation unitary.

In this work we develop a quantum analogue of Orthogonal Matching Pursuit, which we call QOMP.
Unlike QMP, QOMP does not require storing the signal or residual classically: it operates directly on quantum states, leveraging approximate quantum subroutines for inner product estimation, maximum finding, and projection.
The guiding question is whether such routines can enable quantum sparse recovery while retaining the recovery guarantees that have made OMP a cornerstone of compressed sensing and sparse approximation.

This naturally leads us to formalize the \emph{quantum sparse recovery problem}.
While the classical version assumes that the signal vector $\vec{s}$ is explicitly given, in the quantum setting the input may only be available through a state-preparation unitary $U_s$.
In such cases, one cannot simply run classical OMP on stored copies of $\vec{s}$: the algorithm must directly manipulate quantum states.
We therefore introduce the problems $\mathcal{QP}_0$ and $\mathcal{QP}_0^\epsilon$, quantum analogues of $\mathcal{P}_0$ and $\mathcal{P}_0^\epsilon$, and motivate them through their application to quantum tomography.
%%%%%%%%%%%%%%%%%%%%%%%%%%%%%%%%%%%%%%%%%%%%%%%%%%%%%%%%%%%%%%%%%%%%%%%%%%%%%%%%%%%%%%%%%%%%%
%%%%%%%%%%%%%%%%%%%%%%%%%%%%%%%%%%%%%%%%%%%%%%%%%%%%%%%%%%%%%%%%%%%%%%%%%%%%%%%%%%%%%%%%%%%%%
\subsection{Quantum sparse recovery}
We now introduce and formalize the problem of \emph{Quantum Sparse Recovery}, which is the central object of study in this work.
In the classical setting, sparse recovery assumes direct access to the signal $\vec{s}$ as a vector.
In the quantum setting, however, the natural and most powerful access model is through \emph{state-preparation unitaries}.
This is the model that underlies the strongest formulations of quantum tomography~\cite{van2023quantum} and much of quantum algorithm design~\cite{brassard2002quantum, tang2025amplitude, kothari2023mean, huggins2022nearly}.
Specifically, we assume access to:
\begin{itemize}
    \item A unitary $U_s$ such that $U_s: \ket{0} \to \ket{\vec{s}}$, preparing the target state, together with its inverse and controlled versions;
    \item A set of dictionary unitaries $\{U_j\}_{j \in [m]}$ that prepare atoms $\ket{d_j}$, or equivalently a single oracle $U_D : \ket{j}\ket{0} \to \ket{j}|\vec{d}_j\rangle$, with inverses and controlled versions.
\end{itemize}
This model is strictly stronger than having independent copies of $\ket{\vec{s}}$, since access to $U_s$ allows one to generate arbitrarily many copies, and it is flexible enough to capture realistic scenarios where both the state and the dictionary come from known preparation procedures.
Using these unitaries, we formalize the quantum counterparts of problems $\mathcal{P}_0$ and $\mathcal{P}_0^\epsilon$.

\begin{definition}[Quantum exact recovery, $\mathcal{QP}_0$]
\label{def:qpo}
Given access to a quantum state $\ket{\vec{s}} \in \C^N$ and a dictionary $D \in \C^{N \times m}$ via unitaries $U_s$ and $\{U_j\}_{j \in [m]}$ (or $U_D$), together with their inverses and controlled versions, find the smallest set $\Lambda \subseteq [m]$ such that, for some coefficients $\{x_j\}_{j \in \Lambda} \subset \C$,
\begin{align}
    \ket{\vec{s}} = \sum_{j \in \Lambda} x_j |\vec{d}_j\rangle.
\end{align}
\end{definition}

\begin{definition}[Quantum approximate recovery, $\mathcal{QP}_0^\epsilon$]
\label{def:qpoeps}
Given an error tolerance $\epsilon > 0$ and access to a quantum state $\ket{\vec{s}} \in \C^N$ and a dictionary $D \in \C^{N \times m}$ via unitaries $U_s$ and $\{U_j\}_{j \in [m]}$ (or $U_D$), together with their inverses and controlled versions, find the smallest set $\Lambda \subseteq [m]$ such that, for some coefficients $\{x_j\}_{j \in \Lambda} \subset \C$, 
\begin{align}
    \Big\| \ket{\vec{s}} - \sum_{j \in \Lambda} x_j |\vec{d}_j\rangle \Big\|_2 \leq \epsilon.
\end{align}
\end{definition}

Intuitively, the task is to identify the smallest set of dictionary atoms whose span contains (or nearly contains) the target state.
Once this support $\Lambda$ is identified, the problem of tomography reduces to estimating only the coefficients ${x_j : j \in \Lambda}$, rather than reconstructing all $N$ amplitudes in the computational basis.
This two-stage decomposition - first support recovery, then coefficient estimation - is what makes quantum sparse recovery a natural bridge between compressed sensing and efficient quantum tomography.

%%%%%%%%%%%%%%%%%%%%%%%%%%%%%%%%%%%%%%%%%%%%%%%%%%%%%%%%%%%%%%%%%%%%%%%%%%%%%%%%%%%%%%%%%%%%%
%%%%%%%%%%%%%%%%%%%%%%%%%%%%%%%%%%%%%%%%%%%%%%%%%%%%%%%%%%%%%%%%%%%%%%%%%%%%%%%%%%%%%%%%%%%%%
\subsection{Applications to pure state tomography}
Quantum tomography asks for a classical description of an unknown state $\ket{\vec{s}}$, given either copies of the state or oracle access to $U_s$, its inverse, and controlled versions.
In the absence of structure this task is intrinsically costly: reconstructing a pure state in $N$ dimensions requires $\widetilde{\Theta}(N/\epsilon^2)$ copies, or $\widetilde{\Theta}(N/\epsilon)$ queries to $U_s$ and $U_s^\dagger$ to achieve $\ell_2$-norm accuracy $\epsilon$~\cite{van2023quantum}.
These optimal bounds delineate the fundamental limits of tomography for arbitrary pure states.

Sparsity provides a way to break through this barrier.
If $\ket{\vec{s}}$ admits a $K$-sparse representation in an incoherent dictionary $D$, then tomography decomposes into two simpler stages:
\begin{enumerate}
\item \textbf{Support recovery:} Identify the small set $\Lambda \subseteq [m]$ of dictionary atoms whose span contains (or $\epsilon$-approximates) $\ket{\vec{s}}$.
\item \textbf{Coefficient estimation:} Once $\Lambda$ is known, estimate only the $K$ coefficients of $\ket{\vec{s}}$ in the subdictionary $D_\Lambda$, rather than all $N$ amplitudes in the computational basis.
\end{enumerate}

This perspective reframes tomography from an intrinsically high-dimensional reconstruction problem into a structured learning task.
The potential savings are dramatic: the cost of coefficient estimation scales only with $K$ and with the conditioning of the subdictionary $D_\Lambda$, rather than with the full ambient dimension $N$.
The key algorithmic challenge is therefore whether one can recover the sparse support itself with fewer than $\widetilde{O}(N/\epsilon)$ queries to $U_s$.
In this work, we answer this question in the affirmative for certain regimes.

More broadly, our framework is the first to address \emph{sparse tomography in non-orthogonal, overcomplete dictionaries}.
It complements previous structural promises that enabled efficient tomography, such as low rank~\cite{gross2010quantum,kalev2015quantum}, and stabilizer or Pauli structure~\cite{montanaro2017learning, leone2024learning, aaronson2018shadow}.
Here, sparsity plays the role that frequency locality plays in classical compressed sensing: it enables concise descriptions and efficient recovery in settings where naïve tomography would be infeasible.

Beyond its theoretical significance, sparse tomography has potential direct applications, among which:
\begin{itemize}
\item \emph{Approximate state preparation:} an $\epsilon$-close copy of $\ket{\vec{s}}$ can be prepared from a handful of dictionary atoms, potentially using simpler unitaries than those that generated $\ket{\vec{s}}$;
\item \emph{Compact communication:} two parties who agree on a dictionary can transmit only the sparse coefficient vector, akin to JPEG image compression algorithm in the discrete cosine transform basis;
\item \emph{Feature extraction:} sparse coefficients could serve as low-dimensional, interpretable features for downstream quantum or classical learning tasks.
\end{itemize}

In summary, quantum sparse recovery provides a principled route to efficient tomography by leveraging sparsity.
The remainder of this work is devoted to its algorithmic and complexity-theoretic foundations.
Before turning to our methods, we summarize our main results.

%%%%%%%%%%%%%%%%%%%%%%%%%%%%%%%%%%%%%%%%%%%%%%%%%%%%%%%%%%%%%%%%%%%%%%%%%%%%%%%%%%%%%%%%%%%%%
%%%%%%%%%%%%%%%%%%%%%%%%%%%%%%%%%%%%%%%%%%%%%%%%%%%%%%%%%%%%%%%%%%%%%%%%%%%%%%%%%%%%%%%%%%%%%

\subsection{Summary of the results}
Our contributions can be grouped into three main themes: a hardness result that delineates the limits of quantum sparse recovery, the design and analysis of the Quantum Orthogonal Matching Pursuit (QOMP) algorithm, and provable guarantees connecting sparse recovery to efficient tomography.

\emph{\textbf{Hardness.}}~
We begin by showing that \emph{quantum sparse recovery is intractable in full generality}. 
In Theorem~\ref{theorem: Quantum approximate recovery is NP-Hard} we prove that both the exact and approximate formulations, $\mathcal{QP}_0$ and $\mathcal{QP}_0^\epsilon$, are NP-hard for any $\epsilon < \sqrt{3/N}$. 
In particular, unless $\mathrm{NP} \subseteq \mathrm{BQP}$, no quantum algorithm can solve $\mathcal{QP}_0^\epsilon$ using $\mathrm{poly}(N)$ queries to $U_s$ and $U_D$. 
This motivates heuristics and algorithms that exploit additional structure to provide guarantees in identifiable regimes.

\emph{\textbf{The QOMP algorithm.}}~
Motivated by Orthogonal Matching Pursuit (OMP), we introduce Quantum Orthogonal Matching Pursuit (QOMP), the first greedy quantum algorithm for quantum sparse recovery in non-orthogonal, overcomplete dictionaries. 
QOMP mirrors the iterative structure of OMP: at each round it selects the atom with the largest correlation to the residual, updates the current span, projects the residual outside the span, and repeats until the residual norm is small or a certain sparsity threshold is exceeded. 
The challenge is to implement these steps directly on quantum states, where storing and updating the residual classically is not possible.
Our design uses quantum subroutines for inner product estimation, maximum finding, and projection, together with an \emph{error-resetting} strategy that prevents errors from compounding across iterations.

The resulting iteration complexity is captured by Theorem~\ref{theorem: QOMP iteration cost}, considering state preparation unitaries and other $1$- and $2$-qubit gates.
At the $k$-th iteration, QOMP selects the best atom and evaluates the exit condition using per-iteration query complexity
\begin{align} 
 \Ot \left(\left(\frac{\sqrt{m}}{\epsilon_i} + \frac{1}{\epsilon_f}\right)\frac{1}{\gamma}\right) \quad \text{to the target state,}\quad \text{ and } \quad  \quad \Ot \left(\left(\frac{\sqrt{m}}{\epsilon_i} + \frac{1}{\epsilon_f}\right)\frac{\sqrt{k}}{\gamma}\right) \quad \text{to the dictionary}
\end{align}
plus only polynomially many $1$- and $2$-qubit gates.
Here $\epsilon_i$ and $\epsilon_f$ are the accuracies of inner product and norm estimation, $k$ is the iteration counter, and $\gamma$ lower bounds the smallest singular value of the current subdictionary. 
Conceptually, this is the first greedy quantum algorithm that faithfully preserves the spirit of OMP while remaining stable under iteration.

\emph{\textbf{Sparse recovery and tomography.}}~
Our main recovery guarantee shows that QOMP achieves support identification under standard incoherence, \emph{and we quantify the query cost in the state-preparation oracle model}.
Theorem~\ref{theorem: sparse recovery with QOMP} states that if the target can be exactly represented with a $K$-sparse vector in a dictionary of mutual incoherence $\mu=\max_{i \neq j}|\bracket{\vec{d}_i}{\vec{d}_j}|$, and
\begin{align}
    K < \frac{1-\eta}{2-\eta}\left(1 + \frac{1}{\mu}\right),
\end{align}
then running QOMP for at most $K$ iterations (or until the residual norm is $\le \epsilon/2$) with $\epsilon_i\le \eta\gamma\epsilon/\sqrt{K}$ and $\epsilon_f=\epsilon/2$ returns a support $\Lambda\subseteq\Lambda_{\mathrm{opt}}$ of size $\le K$ whose span contains an $\epsilon$-approximation to $\ket{\vec{s}}$, with high probability.
The total number of queries is
\begin{align}
    \Ot\left(\frac{K^{3/2}}{\gamma \eta}\frac{\sqrt{m}}{\epsilon}\right) \quad \text{to $U_s, U_s^\dagger$} \quad \text{ and }\quad \Ot\left(\frac{K^{2}}{\gamma \eta}\frac{\sqrt{m}}{\epsilon}\right) \quad \text{to $U_D, U_D^\dagger$,}
\end{align}
plus polynomially many other resources.
Under the natural \emph{identifiability} condition that no smaller support yields an $\epsilon$-approximation, Corollary~\ref{corollary: exact sparse recovery with QOMP} shows that QOMP recovers the \emph{full} optimal support $\Lambda_{\mathrm{opt}}$, thereby \textbf{solving $\mathcal{QP}_0$ in polynomial time} in this regime.

These guarantees immediately translate into efficient tomography. 
Once the support $\Lambda$ has been recovered, tomography reduces to estimating only the coefficients of $\ket{\vec{s}}$ in the low-dimensional subdictionary $D_\Lambda$. 
Corollary~\ref{corollary: sparse coefficient tomography} shows that, with probability $\ge 1-\delta$, we can output a $\widetilde{O}(K)$-sparse classical vector $\vec{y}$ such that $\|\ket{\vec{s}} - \frac{D_\Lambda \vec{y}}{\|D_\Lambda \vec{y}\|}\| \leq \epsilon$ using
\begin{align}
    \Ot\left(\frac{K^2}{\gamma^2} \frac{1}{\epsilon} \mathrm{polylog}(1/\delta)\right)
\end{align}
queries to $U_s,U_D$ (and inverses/controlled).
In the sparse regime of main interest, where $K=\widetilde{O}(1)$, coefficient estimation is strictly lower-order, so the cost is dominated by support recovery.

When $m=O(N)$ and the optimal support is well conditioned (e.g., $\gamma\in\widetilde{\Omega}(\mathrm{polylog}(N)^{-1})$) with $K=\widetilde{O}(1)$, the support recovery costs $\widetilde{O}(\sqrt{N}/\epsilon)$ queries to $U_s$ (and a comparable number to $U_D$), while coefficients add only $\widetilde{O}(1/\epsilon)$.
This improves over the tight $\Theta(N/\epsilon)$ bound for \emph{general} pure-state tomography with state preparation unitaries~\cite{van2023quantum}, demonstrating that sparsity in incoherent dictionaries permits genuine polynomial savings.

\medskip
\emph{Additional results.}~Our primary results require \emph{no} QRAM; all guarantees are proved in an oracular model, using additional $1$- and $2$-qubit gates and classical computation.
For completeness, we also analyze QOMP under QRAM access (Corollary~\ref{corollary: iteration cost QRAM} and Table~\ref{table: iteration cost quantum}), showing per-iteration polynomial speedups against several classical OMP implementations.
Finally, we provide a quantum routine to estimate the mutual incoherence of a dictionary, achieving additive error $\epsilon$ in time $\widetilde{O}(T_D m/\epsilon)$, where $T_D$ is the dictionary-state preparation cost (Theorem~\ref{theorem: Estimating the mutual incoherence}).
This improves quadratically over the best known classical approximation methods.

\medskip
\emph{In summary}, our results delineate both the barriers and the opportunities of quantum sparse recovery: the problem is NP-hard in general, but in incoherent, well-conditioned regimes QOMP provides a polynomial-time, query-efficient path to both support identification and tomography.
%%%%%%%%%%%%%%%%%%%%%%%%%%%%%%%%%%%%%%%%%%%%%%%%%%%%%%%%%%%%%%%%%%%%%%%%%%%%%%%%%%%%%%%%%%%%%
%%%%%%%%%%%%%%%%%%%%%%%%%%%%%%%%%%%%%%%%%%%%%%%%%%%%%%%%%%%%%%%%%%%%%%%%%%%%%%%%%%%%%%%%%%%%%
%%%%%%%%%%%%%%%%%%%%%%%%%%%%%%%%%%%%%%%%%%%%%%%%%%%%%%%%%%%%%%%%%%%%%%%%%%%%%%%%%%%%%%%%%%%%%
%%%%%%%%%%%%%%%%%%%%%%%%%%%%%%%%%%%%%%%%%%%%%%%%%%%%%%%%%%%%%%%%%%%%%%%%%%%%%%%%%%%%%%%%%%%%%
%%%%%%%%%%%%%%%%%%%%%%%%%%%%%%%%%%%%%%%%%%%%%%%%%%%%%%%%%%%%%%%%%%%%%%%%%%%%%%%%%%%%%%%%%%%%%
%%%%%%%%%%%%%%%%%%%%%%%%%%%%%%%%%%%%%%%%%%%%%%%%%%%%%%%%%%%%%%%%%%%%%%%%%%%%%%%%%%%%%%%%%%%%%
%%%%%%%%%%%%%%%%%%%%%%%%%%%%%%%%%%%%%%%%%%%%%%%%%%%%%%%%%%%%%%%%%%%%%%%%%%%%%%%%%%%%%%%%%%%%%
%%%%%%%%%%%%%%%%%%%%%%%%%%%%%%%%%%%%%%%%%%%%%%%%%%%%%%%%%%%%%%%%%%%%%%%%%%%%%%%%%%%%%%%%%%%%%

\section{Quantum sparse recovery is NP-Hard}
\label{sec: quantum sparse recovery is np hard}
Hardness results serve as guideposts: they delineate the boundary between what is algorithmically feasible and what must rely on structure or heuristics.
In the classical setting, \citet{natarajan1995sparse} showed that sparse recovery is NP-hard via a reduction from \textsc{Exact Cover by 3-Sets (X3C)}.
While one might expect this result to lift directly to the quantum case, we need to take care of one subtlety: classical vectors can be rescaled arbitrarily, quantum states are constrained to have unit norm.
This normalization constraint changes how approximation errors must be handled and invalidates a naive port of the classical reduction.
Classically, one can absorb absolute approximation errors by scaling the target signal, but this freedom disappears for quantum states, where all errors are inherently relative to unit norm.
As we show below, careful control of this point is essential in order to preserve hardness.

We adapt Natarajan’s reduction to the quantum setting, constructing dictionary states that encode the sets in an \textsc{X3C} instance, and a target state that is a uniform superposition over the ground set.
The normalization constraint forces us to bound the allowable error $\epsilon$ explicitly.
We prove that $\mathcal{QP}_0$ and $\mathcal{QP}_0^\epsilon$ remain NP-hard for any $\epsilon < \sqrt{3/N}$, thereby showing that quantum sparse recovery is intractable even with powerful access via state-preparation unitaries.

\begin{theorem}[Quantum Approximate Sparse Recovery is NP-Hard]
\label{theorem: Quantum approximate recovery is NP-Hard}
        Both problems $\mathcal{QP}_0$ and $\mathcal{QP}_0^\epsilon$ are NP-hard for any $\epsilon < \sqrt{3/N}$.
        In particular, no quantum algorithm can solve $\mathcal{QP}_0^\epsilon$ for $\epsilon < \sqrt{3/N}$ using $\mathrm{poly}(N)$ queries to $U_s$ and $U_D$, and $\mathrm{poly}(N)$ additional quantum gates, unless $\mathrm{NP} \subseteq \mathrm{BQP}$.
\end{theorem}
\begin{proof}
    We reduce from the NP-complete problem $\textsc{Exact Cover by 3-Sets (X3C)}$.

    \vspace{2mm}
    \noindent\textbf{Exact Cover by 3-Sets.}
    Given a ground set $B = \{b_1, b_2, \dots, b_N\}$, with $N$ divisible by $3$, and a collection $C = \{c_1, c_2, \dots, c_M\}$ of subsets of $B$, each of size exactly $3$, the task is to decide whether there exists a sub-collection $C' \subseteq C$ such that the sets in $C'$ are pairwise disjoint and collectively cover $B$. That is, every element of $B$ belongs to exactly one set in $C'$ (i.e., $C'$ is an exact cover of $B$).

    \vspace{2mm}
    \noindent\textbf{Reduction Construction.}
    Given an $\textsc{X3C}$ instance ($B$, $C$), we construct an instance of $\mathcal{QP}_0^\epsilon$ as follows:
    \begin{itemize}
        \item Define the target quantum state as the uniform superposition: $\ket{\vec{s}} = \frac{1}{\sqrt{N}} \sum_{i=1}^{N} \ket{i}.$
        \item For each set $c_i \in C$, define a dictionary state: $|\vec{d}_i\rangle = \frac{1}{\sqrt{3}} \sum_{j: b_j \in c_i} \ket{j}.$
    \end{itemize}
    Each $|\vec{d}_i\rangle$ is a unit vector with support on exactly three indices corresponding to the elements in $c_i$.
    The solution vector $\vec{x}$ picks the collections to include in the exact cover.

    The unitaries $U_s$ and $U_D$ that provide access to $\ket{\vec{s}}$ and the dictionary $D = \{ |\vec{d}_i\rangle \}_{i=1}^{M}$ can be implemented with $O(M\polylog(N,M))$ quantum gates and classical preprocessing:
    \begin{itemize}
    \item $U_s: \ket{0} \rightarrow \ket{\vec{s}}$ can be realized at $O(\polylog (N))$ cost.
    \item $U_D: \ket{i}\ket{0} \rightarrow \ket{i} |\vec{d}_i\rangle$ can be implemented by controlling $M$ unitaries $U_i: \ket{0} \rightarrow |\vec{d}_i\rangle$, each requiring $O(\mathrm{polylog}(M))$ cost~\cite{gleinig2021efficient}.
\end{itemize}

    \vspace{2mm}
    \noindent\textbf{Reduction Correctness.}
    We show that the given $\mathrm{X3C}$ instance has an exact cover if and only if the constructed $\mathcal{QP}_0^\epsilon$ instance admits a solution with $\|\vec{x}\|_0 \leq N/3$ and approximation error $< \sqrt{3/N}$.

    \emph{(1) $\textsc{X3C}$ solution $ \implies \mathcal{QP}_0^\epsilon$ solution with $\leq N/3$ entries and $\epsilon < \sqrt{3/N}$.}

    If an exact cover $C' \subseteq C$ exists, define $\vec{x} \in \C^M$ as
    $x_i = \begin{cases}
    \frac{\sqrt{3}}{\sqrt{N}} & \text{if } c_i \in C', \\
    0 & \text{otherwise}
    \end{cases}$.
    Since $C'$ is an exact cover, every element of $B$ appears exactly once among the $|\vec{d}_i\rangle$ with $x_i \neq 0$ and the linear combination becomes $\ket{\vec{s}} = \sum_{i \in [M]} x_i|\vec{d}_i\rangle$.
    Thus, $\vec{x}$ is a valid solution with $\|\vec{x}\|_0 = N/3$ and achieves zero approximation error ($\epsilon = 0 < \sqrt{3/N}$): $\mathcal{QP}^{\epsilon}_0$ can only admit sparser solutions.

    \emph{(2) $\mathcal{QP}_0^\epsilon$ solution with $\leq N/3$ entries  and $\epsilon < \sqrt{3/N}$ $\implies \textsc{X3C}$ solution.}

    Suppose $\mathcal{QP}_0^\epsilon$ admits a solution $\vec{x} \in \C^M$ with $\|\vec{x}\|_0 \leq N/3$ and approximation error $\|\ket{\vec{s}} - \sum_{i=1}^{M} x_i |\vec{d}_i\rangle \|_2 < \sqrt{3/N}$.

    Let $\Lambda = \mathrm{supp}(\vec{x})$ with $|\Lambda| \leq N/3$.
    Each dictionary element $|\vec{d}_i\rangle$ has support on exactly 3 indices.
    Thus, the combined support of $\{ |\vec{d}_i\rangle : i \in \Lambda \}$ covers at most $3 \cdot |\Lambda| \leq N$ indices.

    Since $\ket{\vec{s}}$ has support on all $N$ indices, with amplitude $1/\sqrt{N}$ on each of them, having $|\Lambda| < N/3$ introduces a total $\ell_2$ error of $\sqrt{3/N}$, violating the constraint $\epsilon < \sqrt{3/N}$.
    Therefore, $\Lambda$ must select exactly $N/3$ dictionary elements, whose supports are disjoint and collectively cover $B$.
    This corresponds to an exact cover in the original $\textsc{X3C}$ instance.

    \vspace{2mm}
    \noindent\textbf{Conclusion.}
    We have shown a polynomial-time reduction from X3C to $\mathcal{QP}_0^\epsilon$ for $\epsilon < \sqrt{3/N}$. Hence, $\mathcal{QP}_0^\epsilon$ is NP-hard. Since the reduction uses only polynomial-size quantum circuits for $U_s$ and $U_D$, no quantum algorithm with polynomially many queries and gates can solve $\mathcal{QP}_0^\epsilon$ in the worst case unless $\mathrm{NP} \subseteq \mathrm{BQP}$.
\end{proof}

Theorem~\ref{theorem: Quantum approximate recovery is NP-Hard} rules out efficient classical-quantum algorithms in full generality and motivates the study of \emph{structured regimes}, where additional promises (such as incoherence) permit efficient algorithms.
It is to such regimes that we now turn, introducing the algorithmic background behind the design of the Quantum Orthogonal Matching Pursuit (QOMP) algorithm.

%%%%%%%%%%%%%%%%%%%%%%%%%%%%%%%%%%%%%%%%%%%%%%%%%%%%%%%%%%%%%%%%%%%%%%%%%%%%%%%%%%%%%%%%%%%%%
%%%%%%%%%%%%%%%%%%%%%%%%%%%%%%%%%%%%%%%%%%%%%%%%%%%%%%%%%%%%%%%%%%%%%%%%%%%%%%%%%%%%%%%%%%%%%
%%%%%%%%%%%%%%%%%%%%%%%%%%%%%%%%%%%%%%%%%%%%%%%%%%%%%%%%%%%%%%%%%%%%%%%%%%%%%%%%%%%%%%%%%%%%%
%%%%%%%%%%%%%%%%%%%%%%%%%%%%%%%%%%%%%%%%%%%%%%%%%%%%%%%%%%%%%%%%%%%%%%%%%%%%%%%%%%%%%%%%%%%%%
%%%%%%%%%%%%%%%%%%%%%%%%%%%%%%%%%%%%%%%%%%%%%%%%%%%%%%%%%%%%%%%%%%%%%%%%%%%%%%%%%%%%%%%%%%%%%
%%%%%%%%%%%%%%%%%%%%%%%%%%%%%%%%%%%%%%%%%%%%%%%%%%%%%%%%%%%%%%%%%%%%%%%%%%%%%%%%%%%%%%%%%%%%%
%%%%%%%%%%%%%%%%%%%%%%%%%%%%%%%%%%%%%%%%%%%%%%%%%%%%%%%%%%%%%%%%%%%%%%%%%%%%%%%%%%%%%%%%%%%%%
%%%%%%%%%%%%%%%%%%%%%%%%%%%%%%%%%%%%%%%%%%%%%%%%%%%%%%%%%%%%%%%%%%%%%%%%%%%%%%%%%%%%%%%%%%%%%

\section{Quantum algorithms background}
We begin by formalizing the data access models that will be used throughout, both in the oracular-circuit setting and under QRAM assumptions, and by specifying how we measure query and gate complexity. 
We then review a set of standard quantum primitives, such as amplitude amplification and estimation, inner product estimation, quantum minimum/maximum finding, and block-encodings with singular value transformation. 
Although some of these tools are by now well established, our setting requires adapting their formulations and combining them in ways that ensure stability and efficiency across the iterative structure of QOMP. 
We include them here both for completeness and to keep the exposition self-contained. 
Together, these ingredients establish the background against which our contributions are developed.

%%%%%%%%%%%%%%%%%%%%%%%%%%%%%%%%%%%%%%%%%%%%%%%%%%%%%%%%%%%%%%%%%%%%%%%%%%%%%%%%%%%%%%%%%%%%%
%%%%%%%%%%%%%%%%%%%%%%%%%%%%%%%%%%%%%%%%%%%%%%%%%%%%%%%%%%%%%%%%%%%%%%%%%%%%%%%%%%%%%%%%%%%%%

\subsection{Data access and computational models}
We analyze algorithms in a hybrid setting where a classical computer controls a quantum device operating in the circuit model.
The classical machine stores variables, designs and schedules quantum circuits, and processes measurement outcomes to decide subsequent circuits.
The quantum computer always begins in the all-$\ket{0}$ state, executes a circuit, and measures in the computational basis.

We measure complexity in two complementary ways:
\begin{itemize}
    \item \emph{Gate complexity}, i.e., the asymptotic number of one- and two-qubit gates used across all circuit executions.
    \item \emph{Query complexity}, i.e., the number of calls to oracles implementing state preparation or dictionary access.
\end{itemize}
In line with common practice, we mostly suppress polylogarithmic factors, focusing on the leading polynomial dependencies. 
In the tomography setting in particular, our main resource of interest is the query complexity to the state-preparation and dictionary oracles, while ensuring that all other gates and classical operations remain polynomial in the problem parameters. 
The classical controller itself is assumed to run in the standard RAM model, where memory accesses and arithmetic operations take constant time.

Since data may originate either from classical descriptions or from physical quantum processes, we consider two input models:
\begin{itemize}
    \item \emph{Oracular-Circuit model.} The input consists of explicit state-preparation circuits provided to the algorithm. The classical controller can compile these into larger quantum circuits and invoke them as black-box oracles.
    \item \emph{QRAM model.} The input is loaded into a \emph{classically writable, quantum readable} random access memory (QRAM), which can be queried in superposition. Here, we also account for the classical preprocessing cost of updating QRAM contents during the algorithm.
\end{itemize}

For clarity, our main results assume \emph{exact} access to the state-preparation oracles for the target state and dictionary vectors. 
This idealization isolates the algorithmic ideas and avoids carrying additional technical overhead. 
In realistic settings, finite-precision descriptions or compilation (e.g., via Solovay–Kitaev) introduce small errors.
Since these can be suppressed with logarithmic overhead in the circuit size, one can expect the analysis to extend to this approximate-access regime with only minor modifications.

\subsubsection{The Oracular-Circuit model}
In the \emph{Oracular-Circuit} model we assume black-box access to the signal and dictionary through explicitly given unitary circuits.
That is, the classical controller knows the circuits, and can program the quantum computer to implement them together with their inverses and controlled versions.
In this model, we express algorithmic complexity by counting the number of $1$- and $2$-qubit gates required to run our algorithms and use symbolic variables to keep track of the costs associated to the oracles.
This abstraction is natural when input data is generated by a quantum process/algorithm rather than stored classically, and it provides a clean framework for analyzing query complexity before considering more specialized settings (such as QRAM).
We refer to the ability to implement such black-box circuits as \emph{quantum access} to the data.

Our first step is to formalize what quantum access means in the simplest case of vectors.
\begin{definition} [Quantum access to a vector]
\label{Def:efficient quantum access vector}
    Let $\vec{s} \in \C^n$. 
    We say we have quantum access to $\vec{s}$ if we can implement a unitary operator (controlled, and controlled inverse) that performs the mapping $U_s: \ket{0} \rightarrow \ket{\vec{s}} \defeq \frac{1}{\norm{\vec{s}}_2} \sum_{i \in [n]} s_i \ket{i}$ and the norm $\|\vec{s}\|$ is known.
\end{definition}

In words: given $\ceil{\log (n)}$ qubits, we can coherently load the normalized entries of $\vec{s} \in \C^n$ into amplitudes.
We allow $\vec{s}$ to be non–unit-norm, provided its norm is available as side information.
This notion extends naturally to matrices.

\begin{definition}[Quantum access to a matrix]
\label{def: efficient quantum access matrix} %Eigenfaces
We say we have quantum access to a matrix $A \in \C^{n\times m}$ if we know the norm $\|A\|_F$ and can perform the following mappings (controlled, and controlled inverse):
\begin{itemize}
    \item $U:\ket{j}\ket{0} \rightarrow \ket{j} \ket{\Vec{a}_{j}}=\ket{j}\frac{1}{\norm{\Vec{a}_{j}}} \sum_{i \in [n]} A_{ij}\ket{i},$  for $j \in [m];$
    \item $V:\ket{0}\rightarrow\frac{1}{\|A\|_F}\sum_{j \in [m]}\|\Vec{a}_{j}\|\ket{j}.$
\end{itemize}
\end{definition}

Together, $U$ and $V$ allow one to prepare an amplitude encoding of the full matrix,
\begin{align}
    \ket{A} = \mathrm{SWAP}~U(V \otimes I)\ket{0}\ket{0} = \frac{1}{\norm{A}_F}\sum_{i \in [n]}\sum_{j \in [m]} A_{ij}\ket{i}\ket{j}.
\end{align} 
This generalizes vector access (obtained by considering a single column).

While this is the general definition of quantum access to a matrix, in this work, the main matrix of interest is the dictionary $D$. 
Its columns are normalized, so the access model simplifies: $V$ becomes just the uniform superposition $\frac{1}{\sqrt{m}}\sum_{j\in[m]} \ket{j}$, which can be implemented in polylogarithmic time.

\begin{definition}[Quantum access to the dictionary]
\label{def: qomp  quantum access dictionary}
    We define quantum access to a dictionary $D \in \C^{n \times m}$ as the ability to implement a unitary $U_D$, its inverse $U_D^\dagger$, and their controlled versions, in time $T_D$. The unitary acts as
    \begin{align}
        U_D\ket{j}\ket{0} = \ket{j}|\vec{d}_j\rangle
    \end{align}
    for all $j \in [m]$, where $|\vec{d}_j\rangle = \sum_{i \in [n]} D_{ij}\ket{i}.$
\end{definition}

Later we will also need to restrict the dictionary to a subset of columns. 
We can do so by changing the unitary $V$ that selects the columns.
For this, we formalize quantum access to sets of indices.

\begin{definition}[Quantum access to a set and its complement]
\label{def: qomp  quantum access sets}
    Let $\Lambda \subseteq [m]$ be a set. 
    We define quantum access to $\Lambda$ and its complement $\overline{\Lambda} = [m] \setminus \Lambda$ as the ability to implement unitaries $U_{\Lambda}, U_{\overline{\Lambda}}$, their inverses $U_{\Lambda}^\dagger, U_{\overline{\Lambda}}^\dagger$, and their controlled versions, in times $O(T_{\Lambda})$ and $O(T_{\overline{\Lambda}})$, respectively. The unitaries act as
    \begin{align}
        U_{\Lambda}\ket{0} = \frac{1}{\sqrt{\abs{\Lambda}}} \sum_{i \in \Lambda} \ket{i} \quad\text{ and }\quad U_{\overline{\Lambda}} \ket{0} = \frac{1}{\sqrt{m-\abs{\Lambda}}} \sum_{i \in [m] \setminus \Lambda} \ket{i}.
    \end{align}
    We use $T_U$ to denote the time needed by a classical algorithm to update the circuits upon insertion or deletion of one element in $\Lambda$.
    
    We call the access efficient if $T_{\Lambda}, T_{\overline{\Lambda}} \in O(\min(\abs{\Lambda}, \abs{\overline{\Lambda}})\mathrm{polylog}(m))$ and if $T_U \in O(\mathrm{polylog}(m))$. 
\end{definition}

While Definition~\ref{def: qomp  quantum access sets} introduces the access model abstractly, one may ask about its implementability.
In principle, it is always possible to construct unitaries $U_\Lambda$ and $U_{\overline{\Lambda}}$ using $\Ot(m)$ gates together and classical preprocessing.
Moreover, more careful constructions can achieve $O(\min(|\Lambda|,|\overline{\Lambda}|)\polylog (m))$ gate complexity, with classical updates supported in $O(\mathrm{poly}(\min(|\Lambda|,|\overline{\Lambda}|)))$ time per insertion or deletion.
Since these bounds are not the focus of this work, we keep the corresponding costs symbolic throughout the analysis.

These access primitives form the basis of the Oracular-Circuit model.
In particular, they will allow us to efficiently implement block encodings of $D$ and its subdictionaries $D_\Lambda$, which are the key ingredients enabling QOMP.

\subsubsection{The QRAM model}
\label{sec: qram}
A quantum random access memory (QRAM) is a device that, analogously to classical RAM, allows efficient storage and retrieval of bitstrings, but with the additional capability of being queried in superposition. Formally, given $N$ cells each storing a bitstring $x_i$ of length $p$, QRAM implements the unitary
\begin{align}
\label{eq: QRAM}
    U_{\mathrm{QRAM}}: \ket{i}\ket{0} \to \ket{i}\ket{x_i}, \qquad i \in [N]
\end{align}
where each $x_j$ is encoded in $p$ qubits as a corresponding computational basis state (e.g., $10010 \to \ket{10010}$)
We adopt the standard convention that QRAM is \emph{classically writable and quantum readable}, and that queries are unitary, with inverses and controlled versions available.
Following standard practice, we regard a QRAM call as taking $O(\polylog(N))=\Ot(1)$ time, in analogy to constant-time classical RAM access.

\begin{figure*}[th]
\centering
\begin{subfigure}{0.48\linewidth}
\begin{tikzpicture}
    \node(0){$\|A\|_F^2$}
        child{node{$\dots$}
            child{node{$\|\vec{a}_{0}\|^2 + \|\vec{a}_{1}\|^2$}
                child{node{$\vec{a}_{0}$}}
                child{node{$\vec{a}_{1}$}}
            }
            child{node{$\dots$}}
        }
        child{node{$\dots$}};
\end{tikzpicture}
\caption{Tree storing the column norms of $A$.}
\label{subfig:norm tree}
\end{subfigure}
\begin{subfigure}{0.48\linewidth}
\begin{tikzpicture}
    \node(0){$\|\vec{a}_j\|^2$}
        child{node{$\dots$}
            child{node{$|A_{0j}|^2+|A_{1j}|^2$}
                child{node{$A_{0j}$}}
                child{node{$A_{1j}$}}
            }
            child{node{$\dots$}}
        }
        child{node{$\dots$}};
\end{tikzpicture}
\caption{Tree storing the entries of the $j^\text{th}$ column.}
\label{subfig:column entries tree}
\end{subfigure}
\caption{Binary tree structures enabling efficient quantum access to a matrix $A \in \C^{n \times m}$. 
Each node stores the sum of squares of its children. 
A global tree (left) encodes the column norms, while one tree per column (right) encodes entry magnitudes. 
These structures allow efficient implementation of the $U$ and $V$ unitaries from Def.~\ref{def: efficient quantum access matrix} via QRAM.}
\label{fig:column tree}
\end{figure*}
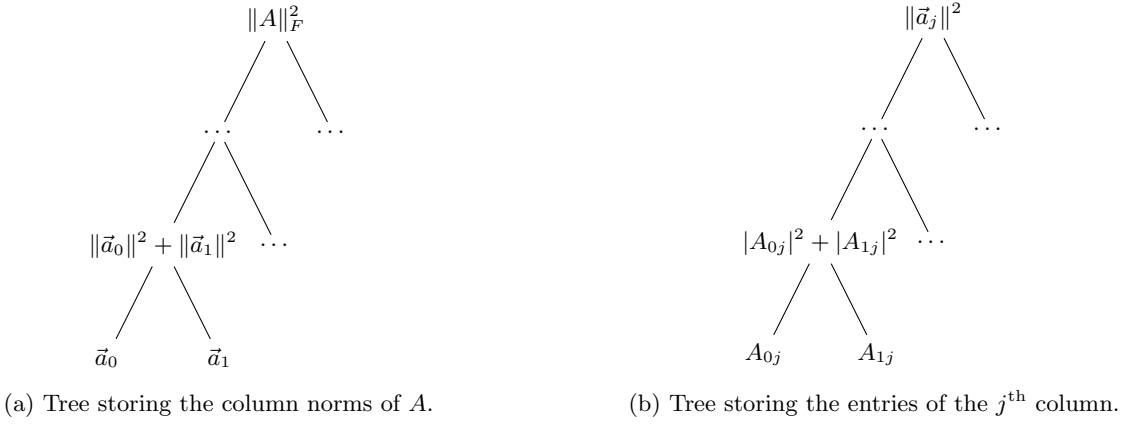

This assumption is debated. 
In principle, a multiplexer circuit can realize this mapping with $O(\log(N) \mathrm{poly}(p))$ qubits and depth $O(N\mathrm{poly}(p))$, while more advanced architectures use $O(N\mathrm{poly}(p))$ qubits and logarithmic depth~\cite{van2023quantum}. 
Hardware-oriented proposals such as bucket-brigade QRAM~\cite{giovannetti2008architectures} achieve polylogarithmic depth by parallelizing memory access, and recent results~\cite{hann2021resilience} show such architectures can be resilient even with error correction. 
At the same time, significant skepticism remains about scalability and integration with fault-tolerant hardware, suggesting potentially large overheads~\cite{jaques2023qram}. 
In this work, as is common in quantum algorithms and learning theory, we \emph{assume} the availability of such devices and focus on the algorithmic consequences. 
In this work we adopt the conventional $\Ot(1)$ abstraction, while stressing that our algorithms can also run in the \emph{Oracular-Circuit model} with costs proportional to the chosen data-loading schemes.

In practice, QRAM-based access to vectors and matrices is realized through hierarchical binary trees that store prefix sums of squared amplitudes, sometimes called \emph{KP-trees} after Kerenidis and Prakash~\cite{kerenidisP17recommendation, kerenidis2020gradient}.
Figure~\ref{fig:column tree} illustrates the trees: one global tree storing column norms, and one tree per column storing entry norms. 
Coherent QRAM access to these tree entries suffices to implement the unitaries $U$ and $V$ from Def.~\ref{def: efficient quantum access matrix}.

\begin{theorem}[Implementing quantum operators using an efficient data structure {\cite{kerenidisP17recommendation}}]
\label{thm:efficient data structure}
    Let $A \in \C^{n\times m}$.
    There exists a data structure to store the matrix $A$ with the following properties:
    \begin{enumerate}
        \item The size of the data structure is $O(nnz(A)\log^2(nm))$.
        \item The time to update/store a new entry $(i,j,A_{ij})$ is $O(\log(nm))$\footnote{The original proof, which can be found in the appendix of the referenced paper, considers time $O(\log^2(nm))$ because it considers that the entries are encoded in $\log(nm)$ bits. Similarly to \citet[Theorem 4]{chakraborty2019power}, we do not consider this overhead, as one might want to tune the number of bits to the required precision.
        Note that we generally omit logarithmic overheads due to the precision of binary encodings and hardware limitations.}.
        \item Provided coherent quantum access to this structure (Eq. (\ref{eq: QRAM})) there exists quantum algorithms that implement $U$ and $V$ as per Def.~\ref{def: efficient quantum access matrix} in time $O(\polylog(nm))$.
    \end{enumerate}
    Similarly, given a vector $\Vec{x} \in \C^n$ stored in this data structure, we can create access to $\ket{\Vec{x}} = \frac{1}{\norm{\Vec{x}}} \sum_{i=1}^n x_i \ket{i}$ in $O(\polylog(n))$ time and update/store a new entry in $O(\log(n))$.
\end{theorem}

For our QOMP, we require not only access to the full dictionary $D$ but also to dynamically updated subdictionaries $D_\Lambda$.
Doing so requires being able to modify the matrix access unitary $V$ selecting the columns.
The QRAM and KP-Trees framework naturally supports this: one can maintain auxiliary norm trees for $\Lambda$ and $\overline{\Lambda}$ and update them when adding a column to the active set. Each update costs $O(\log m)$ classical time, after which quantum access to the sets remains efficient, as per Def.~\ref{def: qomp  quantum access sets}.

In the standard KP-tree construction~\cite{kerenidisP17recommendation} (Figure~\ref{fig:column tree}), the global tree encodes $\|A\|_F^2$, so state preparation is normalized by $\|A\|_F$.  
This dependence propagates into block-encodings built from such access.  
Subsequent work~\cite{kerenidis2020gradient} showed that one can modify what is stored in QRAM to obtain a smaller normalization factor $\mu(A)$.  

\begin{definition}[Parameter $\mu_p(A)$]
\label{def:mu}
    Let $A \in \C^{n \times m}$. 
    Then, $\mu_p(A) =  \sqrt{s_{2p}(A)s_{2(1-p)}(A^T)}$, where $s_p(A) = \max_{i \in [n]} \norm{\Vec{a}_{i, \cdot}}_p^p$.
\end{definition}

This is achieved by factoring
\begin{align}
\frac{A}{\mu(A)} = P \circ Q,
\end{align}
where $\circ$ denotes entrywise multiplication and the rows of $P$ and columns of $Q$ are normalized in $\ell_2$. 
The corresponding QRAM trees store $P$ and $Q$, enabling amplitude encodings with normalization $\mu(A)$. 

Different values of $p$ yield different $\mu_p(A)$; in practice one may preprocess a constant set of values and select the best.  
This preprocessing is performed once when storing the dictionary and can be amortized over many signals.  
Both normalizations, $\|A\|_F$ and $\mu(A)$, extend naturally to subdictionaries $D_\Lambda$ by maintaining separate trees selecting the columns.

We emphasize that $\mu_p(A)$ here is a QRAM efficiency parameter and should not be confused with the \emph{mutual incoherence} $\mu$ of a dictionary, which appears in our recovery guarantees.
Both notations are standard in their respective literatures, and we retain them for continuity; the meaning will be clear from context.

\emph{In summary:} The QRAM model provides a unified framework for efficient quantum access: one can store a signal vector or a dictionary in KP-trees with linear-time preprocessing, prepare amplitude encodings in $\Ot(1)$ time, and maintain dynamic access to subdictionaries $D_\Lambda$ with logarithmic update costs.  
Normalization can be chosen between $\|A\|_F$ and $\mu(A)$ depending on preprocessing, and the same mechanism applies to both the target signal and the dictionary.  
Thus the QRAM abstraction supports all the access assumptions required for QOMP with polylogarithmic overhead in $n,m$.

%%%%%%%%%%%%%%%%%%%%%%%%%%%%%%%%%%%%%%%%%%%%%%%%%%%%%%%%%%%%%%%%%%%%%%%%%%%%%%%%%%%%%%%%%%%%%
%%%%%%%%%%%%%%%%%%%%%%%%%%%%%%%%%%%%%%%%%%%%%%%%%%%%%%%%%%%%%%%%%%%%%%%%%%%%%%%%%%%%%%%%%%%%%

\subsection{Algorithmic primitives}
We now turn to the algorithmic primitives that underpin QOMP. 
Our algorithm relies on variations of well-established quantum tools for boosting success probabilities, estimating overlaps, and performing linear-algebraic transformations. 
In particular, we make use of amplitude amplification and estimation, inner product estimation, quantum minimum/maximum finding, block-encodings combined with quantum singular value transformation (QSVT), sparse tomography in an orthonormal basis, and techniques for amplifying success probabilities or converting between Las Vegas and Monte Carlo algorithms. Additional background on QSVT and polynomial approximations is deferred to Appendix~\ref{appendix: singular value transformation}.

%%%%%%%%%%%%%%%%%%%%%%%%%%%%%%%%%%%%%%%%%%%%%%%%%%%%%%%%%%%%%%%%%%%%%%%%%%%%%%%%%%%%%%%%%%%%%

\subsubsection{Amplitude amplification and estimation}
We will use both amplitude amplification and estimation~\cite{brassard2002quantum}.
Suppose we have a unitary $U$ (with inverse and controlled implementations) such that
\begin{align}
\label{eq: amp amp or est}
    U\ket{0} = a\ket{\vec{x}, 1} + b|\vec{G}, 0\rangle.
\end{align}
where the ancilla qubit flags the \emph{good} subspace by $\ket{1}$.
Amplitude amplification prepares a state $|\vec{\psi} \rangle$ such that $||\vec{\psi}\rangle - \ket{\vec{x}}| \leq \epsilon$ with high probability, while amplitude estimation outputs a probability estimate $p$ satisfying $|p - |a|^2| \leq \epsilon$ with high probability.

Like Grover's algorithm, standard amplitude amplification alternates reflections about the \emph{initial state} and the \emph{bad} subspace.
If the initial success amplitude $a$ is unknown, naively iterating these reflections risks \emph{overshooting} the target state~\cite{grover1996fast, grover1997quantum}, drifting away from the \quoted{good} state that we want to prepare.
To avoid this, we use fixed-point amplification, which requires only a lower bound on $|a|$ and eliminates the risk of overshooting.
We follow the block-encoding formulation of \citet{gilyen2019quantum}; see also the original construction by \citet{yoder2014fixed}.

\begin{theorem}[Fixed-point amplitude amplification~{\cite[Theorem 27, arxiv]{gilyen2019quantum}}]
\label{thm: qomp  fixed point amp amp}
    Let $U$ be a unitary and $\Pi$ be an orthogonal projector such that $a|\vec{\psi}_G\rangle = \Pi U |\vec{\psi}_0\rangle$, and $a > \delta > 0$.
    There is a unitary circuit $\widetilde{U}$ such that $\|\vec{\psi}_G\rangle - \widetilde{U}|\vec{\psi}_0\rangle\|_2 \leq \epsilon$, which uses a single ancilla qubit and consists of $O\left(\frac{\log(1/\epsilon)}{\delta}\right)$ $U$, $U^\dagger$, $C_{\Pi}NOT$, $C_{\ketbra{\psi_0}{\psi_0}}NOT$ and $e^{i\phi \sigma_z}$ gates.
\end{theorem}

In this formulation, the projector identifies the \quoted{good} subspace. 
For instance, in the setting of Eq.~(\ref{eq: amp amp or est}) one may take $\Pi= \ketbra{1}{1}$, so that $\Pi U |\vec{\psi}_0\rangle = a\ket{\vec{x}, 1}$.

We next turn to amplitude estimation. 
The textbook routine from \citet[Theorem 12]{brassard2002quantum} estimates $|a|^2$ without overshooting concerns.
However, in our applications, we require an estimate of $|a|$ itself.
A simple modification together with a stability bound for $\sin$ suffices.

\begin{lemma}[Error propagation $\sin(\theta)$]
\label{lem: qomp error sin}
    Let $a = \sin(\theta)$ and $\overline{a} = \sin(\overline{\theta})$ with $0 \leq \theta,\overline{\theta} \leq 2\pi$, then $\abs{\overline{\theta} - \theta} \leq \epsilon \implies \abs{a - \overline{a}} \leq \epsilon$.
\end{lemma}
\begin{proof}
    The Mean Value (or Lagrange) Theorem states that $f'(c) = \frac{f(b) - f(a)}{b-a}$, where $f' = \dv{f}{x}$ for some $c \in (a,b)$ and $f$ continuous in $[a,b]$, differentiable in $(a,b)$. 
    From this, we can write $\abs{\sin{\theta} - \sin{\overline{\theta}}} \leq \cos(c)\epsilon$, for $c \in (\theta-\epsilon, \theta+\epsilon)$. 
    Using $\cos(x) \leq 1$, we have $\abs{\overline{a} - a} \leq \epsilon$.
\end{proof}

\begin{theorem}[Absolute value amplitude estimation]
\label{thm: amplitude estimation amplitude}
There is a quantum algorithm which takes as input one copy of a quantum state $\ket{\varphi}$, a unitary transformation $U=2\ketbra{\vec{\varphi}}{\vec{\varphi}} - \I$, a unitary transformation $V=\I - 2P$ for some projector $P$, and an integer $t$. 
The algorithm outputs $\overline{a}$, an estimate of $a = \sqrt{\bra{\vec{\varphi}}P\ket{\vec{\varphi}}}$,
such that
    \begin{align}
        \abs{\overline{a} - a} &\leq \frac{\pi}{t}
    \end{align}
with probability at least $8/\pi^2$, using exactly $t$ evaluations of $U$ and $V$.
\end{theorem}
\begin{proof}
    The proof follows \citet[Theorem 12]{brassard2002quantum}, until the estimation of the angle $\theta$: $\abs{\overline{\theta} - \theta} \leq \frac{\pi}{t}$. Then, using the bound of Lemma~\ref{lem: qomp error sin}, we conclude the algorithm by outputting $\overline{a} = \sin(\overline{\theta})$.
\end{proof}

%%%%%%%%%%%%%%%%%%%%%%%%%%%%%%%%%%%%%%%%%%%%%%%%%%%%%%%%%%%%%%%%%%%%%%%%%%%%%%%%%%%%%%%%%%%%%

\subsubsection{Inner product estimation}
Throughout the paper we will need to perform inner products between amplitude encoded vectors. 
We report a result from \citet{kerenidis2019qmeans} and tailor it to our needs.

\begin{theorem}[Inner product estimation \cite{kerenidis2019qmeans}]
\label{thm:innerproductestimation}
Let there be quantum access to the matrices $V \in \mathbb{R}^{n \times m}$ and $C \in \mathbb{R}^{k \times m}$ through the unitaries
$U_v: \ket{i}\ket{0} \rightarrow \ket{i}\ket{\vec{v}_{i, \cdot}}$ and $U_c: \ket{j}\ket{0} \rightarrow \ket{j}\ket{\vec{c}_{j,\cdot}}$, that run in time $T_v$, $T_c$, respectively.
Then, for any $\epsilon>0$, there exists a quantum algorithm that computes
$\ket{i}\ket{j}\ket{0}  \rightarrow   \ket{i}\ket{j} |\overline{\braket{\vec{v}_{i,\cdot}}{\vec{c}_{j,\cdot}}} \rangle$,
such that
$| \overline{\braket{\vec{v}_{i,\cdot}}{\vec{c}_{j,\cdot}}}-\braket{\vec{v}_{i,\cdot}}{\vec{c}_{j,\cdot}} | \leq  \epsilon$,
with high probability
in time $\widetilde{O}\left(\frac{1}{ \epsilon}(T_v + T_c)\right)$.
\end{theorem}

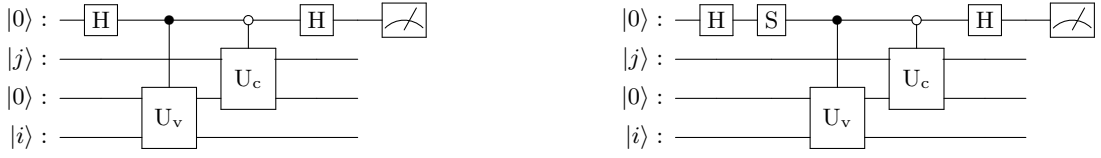
\begin{figure}[t]
    %\centering
    \begin{subfigure}{0.47\linewidth}
    \scalebox{1}{ 
        \Qcircuit @C=1.0em @R=0.2em @!R { \\
            \nghost{{\ket{0}} :  } & \lstick{{\ket{0}} :  } & \gate{\mathrm{H}} & \ctrl{2} & \ctrlo{1} & \gate{\mathrm{H}} & \qw & \meter\\
            \nghost{{\ket{j}} :  } & \lstick{{\ket{j}} :  } & \qw & \qw & \multigate{1}{\mathrm{U_{c}}} & \qw & \qw  \\
            \nghost{{\ket{0}} :  } & \lstick{{\ket{0}} :  } & \qw & \multigate{1}{\mathrm{U_{v}}} & \ghost{\mathrm{U_{b}}} & \qw & \qw  \\
            \nghost{{\ket{i}} :  } & \lstick{{\ket{i}} :  } & \qw & \ghost{\mathrm{U_{v}}} & \qw & \qw & \qw \\
        \\ }
    }
    \caption{The probability of measuring the auxiliary qubit in the state $\ket{1}$ is $P=\frac{1 - \mathrm{Re}[\langle \Vec{v}_i | \Vec{c}_j \rangle ]}{2}$.}
    \label{subfig: real part}
    \end{subfigure}
    \begin{subfigure}{0.47\linewidth}
    \scalebox{1}{ 
        \Qcircuit @C=1.0em @R=0.2em @!R { \\
            \nghost{{\ket{0}} :  } & \lstick{{\ket{0}} :  } & \gate{\mathrm{H}} & \gate{\mathrm{S}} & \ctrl{2} & \ctrlo{1} & \gate{\mathrm{H}} & \qw & \meter\\
            \nghost{{\ket{j}} :  } & \lstick{{\ket{j}} :  } & \qw & \qw & \qw & \multigate{1}{\mathrm{U_{c}}} & \qw & \qw  \\
            \nghost{{\ket{0}} :  } & \lstick{{\ket{0}} :  } & \qw & \qw & \multigate{1}{\mathrm{U_{v}}} & \ghost{\mathrm{U_{b}}} & \qw & \qw  \\
            \nghost{{\ket{i}} :  } & \lstick{{\ket{i}} :  } & \qw & \qw & \ghost{\mathrm{U_{v}}} & \qw & \qw & \qw \\
        \\ }
    }
    \caption{The probability of measuring the auxiliary qubit in the state $\ket{1}$ is $P=\frac{1 - \mathrm{Im}[\langle \Vec{v}_i | \Vec{c}_j \rangle]}{2}$.}
    \label{subfig: imaginary part}
    \end{subfigure}
    \caption{Quantum circuit to estimate $\langle \Vec{v}_i| \Vec{c}_j \rangle$. 
    Here $U_v\ket{i}\ket{0} = \ket{i}\ket{\Vec{v}_i}$ and $U_c\ket{j}\ket{0} = \ket{j} |\Vec{c}_j \rangle$. }
    \label{fig: IPE estimation circuit}
\end{figure}

This theorem uses a modified version of the Hadamard test followed by amplitude estimation, and in the current form only deals with real valued amplitudes.
However, when used with complex amplitudes, the same procedure estimates the real part of the inner product (see Figure~\ref{subfig: real part}). 
Modifying the circuit with an $S$ gate $\left(S = \begin{bmatrix}
        1 & 0 \\
        0 & i
\end{bmatrix} \right)$ 
enables estimating the imaginary part as well. 
The two resulting circuits are shown in Figure~\ref{fig: IPE estimation circuit}. 
Since we only implement inner products between a specific vector and the columns of a matrix, we also simplify the required quantum access.
These modifications lead us to formulate the following corollary.

\begin{corollary}
[Complex inner products estimation] 
\label{coro: complex inner}
    Let there be quantum access to a matrix $V \in \C^{n \times d}$ and a vector $\vec{c} \in \C^d$ through the following unitaries
    $U_v: \ket{i}\ket{0} \rightarrow \ket{i}\ket{\Vec{v}_i}$, and $U_c:\ket{0} \rightarrow |\Vec{c}\rangle$ and inverse, that
    can be controlled and executed in times $T_v$ and $T_c$ respectively. Let the norm $\|\vec{c}\|$ be known and let there be quantum access to the norms of $V$ via $\ket{i}\ket{0} \rightarrow \ket{i}\ket{\|\Vec{v}_i\|}$ in time $T_N$. For any $\delta > 0$ and $\epsilon>0$, there exist quantum algorithms that compute:
    \begin{itemize}
        \item $\ket{i}\ket{0}  \rightarrow   \ket{i} | \mathrm{Re}[\overline{(\Vec{v}_i,\Vec{c})}] \rangle$ where  $|\mathrm{Re}[\overline{(\Vec{v}_i,\Vec{c})}]-\mathrm{Re}[(\Vec{v}_i, \Vec{c})]| \leq  \epsilon$ w.p. $\geq  1-\delta$
        \item $\ket{i}\ket{0}  \rightarrow   \ket{i} | \mathrm{Im}[\overline{(\Vec{v}_i,\Vec{c})}] \rangle$ where  $|\mathrm{Im}[\overline{(\Vec{v}_i,\Vec{c})}]-\mathrm{Im}[(\Vec{v}_i, \Vec{c})]| \leq  \epsilon$ w.p. $\geq  1-\delta$
    \end{itemize}
    in time $\widetilde{O}\left(\frac{ \norm{\Vec{v}_i} \|\Vec{c}\|}{\epsilon}(T_v+T_c) \log(1/\delta) + T_N\right)$.
\end{corollary}

%%%%%%%%%%%%%%%%%%%%%%%%%%%%%%%%%%%%%%%%%%%%%%%%%%%%%%%%%%%%%%%%%%%%%%%%%%%%%%%%%%%%%%%%%%%%%

\subsubsection{Finding the minimum/maximum}
Finally, the last quantum subroutine is an algorithm to find the index of the minimum value in a list of real numbers.

\begin{theorem}[Finding the minimum {\cite{durr1996findingminimum}}] 
\label{thm:finding minimum}
Let there be quantum access to a vector $\Vec{u} \in \R^n$ via the operation $\ket j \ket{0} \to \ket j \ket{ u_j}$ in time $T$.
Then, we can find the minimum $u_{\min} = \min_{j\in[n]} |u_j|$ and its index $j_{\min} = \argmin_{j\in[n]} u_j$ w.p. greater than $1-\delta$ in time $O\left(T \sqrt n \log \left (\frac{1}{\delta}\right) \right)$.
\end{theorem}

This result is due to \citet{durr1996findingminimum}.
It builds on Grover's search and queries the list quadratically less than its classical counterpart whenever the search is unstructured.
This routine requires access to the list's state preparation unitary, its inverse, and their controlled versions.
However, we often build an approximation of the state preparation unitary that we need.
Instead of having a map $\ket j \ket{0} \to \ket j \ket{ u_j}$, we have an approximation that produces $\ket j \ket{0} \to \sqrt{1-\delta_1}\ket j \ket{ \vec{u}_j} + \sqrt{\delta_1}\ket{j} |\vec{G}\rangle$, where $|\vec{G}\rangle$ is a garbage state orthogonal to $\ket{\vec{u}_j}$ and $\ket{\vec{u}_j}$ is a state-vector that upon measurement yields an output $\overline{u}_j$ such that $\abs{\overline{u}_j - u_j} \leq \epsilon$.
This is the case, for instance, when the entries of $u_j$ are computed using another quantum algorithm with an amplitude estimation routine, like with the \emph{Complex inner product estimation} from Corollary~\ref{coro: complex inner}.
\citet{wiebeKS15nearestneighbor} and \citet{chen2021quantum} show that we can still find the minimum using the approximated unitary, in time $\widetilde{O}(T\sqrt{n})$ with polylogarithmic overhead.

\begin{theorem}[Finding the minimum with an approximate unitary {\cite{chen2021quantum}}]
\label{thm: finding approximate version}
    Let $\delta, \epsilon \in (0,1)$. 
    Let there be quantum access to a vector $\vec{u} \in \R^n$ via a unitary that computes $\ket{j}\ket{0} \rightarrow \ket{j}\ket{u_j}$ in time $T$ and such that for every $j \in [n]$, after measuring the state $\ket{u_j}$, with high probability the measurement outcome $\overline{u}_j$ satisfies $\abs{\overline{u}_j - u_j} \leq \epsilon$.
    There exists a quantum algorithm that finds an index $j$ such that $u_j \leq \min_{k \in [n]} u_k + 2\epsilon$ with probability $\geq 1 - \delta$ in time $O(T\sqrt{n}~\polylog(n,\delta^{-1}))$.
\end{theorem}

In our QOMP algorithm, we are interested in finding the index of the \emph{maximum} value in a list, rather than the minimum.
However, if we have access to a unitary (or an approximation) $U: \ket{j}\ket{0}\to \ket{j}\ket{u_j}$, we can implement $U: \ket{j}\ket{0}\to \ket{j}\ket{-u_j}$ with arithmetic operations in $\widetilde{O}(1)$ time and leverage $\argmin(x) = \argmax(-x)$ to turn this routine into what we need.
Another observation is that we can modify the initial state to find the minimum (or maximum) among a subset of the indices of the list. 
The following corollary incorporates these two observations.

\begin{corollary}[Finding the maximum with an approximate unitary] 
\label{coro: finding approximate maximum}
    Let $\delta, \epsilon \in (0,1)$. 
    Let there be quantum access to a vector $\vec{u} \in \R^n$ via a unitary that computes $\ket{j}\ket{0} \rightarrow \ket{j}\ket{u_j}$ in time $T_u$ and such that for every $j \in [n]$, after measuring the state $\ket{u_j}$, with high probability the measurement outcome $\overline{u}_j$ satisfies $\abs{\overline{u}_j - u_j} \leq \epsilon$.    
    Let there also be access to a unitary $U_S$ that prepares a uniform superposition of a subset $S\subseteq [n]$ of indices, of size $d$, such that it can be implemented, inverted, and controlled in time $T_S$.
    There exists a quantum algorithm that finds an index $j$ such that $u_j \geq \max_{k \in S} (u_k - 2\epsilon)$ with probability $\geq 1 - \delta$ in time $O((T_u + T_S)\sqrt{d}~\polylog(n,\delta^{-1}))$.
\end{corollary}

%%%%%%%%%%%%%%%%%%%%%%%%%%%%%%%%%%%%%%%%%%%%%%%%%%%%%%%%%%%%%%%%%%%%%%%%%%%%%%%%%%%%%%%%%%%%%

\subsubsection{Block-encodings, singular value transformation, and linear systems}
\label{section: block encodings}
A central tool in modern quantum algorithms is the ability to represent a matrix as a block of a larger unitary, known as a \emph{block-encoding}. 
Block-encodings allow us to simulate matrix transformations using quantum singular value transformation (QSVT), which in turn enables algorithmic primitives such as projections, pseudoinverse computation, and linear system solving. 
We summarize the main definitions and results that we will require.

\begin{definition} [Block-encoding \cite{chakraborty2019power, gilyen2019quantum}] 
\label{def: block-encoding}
Suppose that $A$ is an s-qubit operator, $\alpha, \epsilon \in \mathbb{R}^+$ and $q \in N$. We say that the $(s + q)$-qubit unitary $U_A$ is an $(\alpha, q, \epsilon)$ block-encoding of $A$, if
\begin{equation}
\label{eq: block encoding error}
    \left\|A-\alpha(\bra{0}^{\otimes q}\otimes \mathbb{I})U_A(\ket{0}^{\otimes q}\otimes \mathbb{I})\right\| \leq \epsilon,
\end{equation}
where $\|\cdot\|$ is the operator norm. 
\end{definition}

In words, a block-encoding embeds a (generally non-unitary, non-Hermitian) matrix $A$ into the top-left block of a larger unitary $U_A$, up to a normalization factor $\alpha$. 
This normalization satisfies $\alpha \geq \|A\|$, and can be tuned depending on the access model.
Note that a block-encoding of $A$ is roughly equivalent to a block-encoding of $A/\alpha$.

\begin{lemma}
\label{lemma: mini block encoding}
    Let $U_A$ be an $(\alpha, a, \epsilon)$ block-encoding of a matrix $A$. Then, $U_A$ is a $(1, a, \epsilon/\alpha)$ block-encoding of $A/\alpha$.
\end{lemma}
\begin{proof}
    Observe the definition of block-encoding and divide Eq.~(\ref{eq: block encoding error}) by $\alpha$.
\end{proof}

With quantum access to a matrix $A$ as in Def.~\ref{def: efficient quantum access matrix}, block-encodings can be obtained essentially with negligible overhead.
We report a result whose proof can be found in \citet[proof of Lemma 25, arxiv version]{chakraborty2019power}.

\begin{theorem}[Block-encoding from quantum access {\cite{chakraborty2019power}}]
\label{thm: qomp  block encoding from qa}
    Let there be quantum access to a matrix $A \in \C^{n \times m}$ as per Def~\ref{def: efficient quantum access matrix} in times $T_U$, $T_V$.
    Then there exist unitaries $U_R, U_L$ that can be implemented in time $\Ot(T_U + T_V)$ such that $U_R^\dagger U_L$ is a $(\norm{A}_F, \ceil{\log(n+m)}, \epsilon)$-block-encoding of $A$.
\end{theorem}

We can also create block-encodings of subdictionaries. 
Indeed, given access to the dictionary $D$ and a set $\Lambda$, we can use $U=U_D$ and $V=U_\Lambda$ to obtain a block-encoding of the restricted dictionary $D_\Lambda$.

If the matrix is stored in a QRAM-based data structure (Sec.~\ref{sec: qram}), we obtain analogous guarantees with normalization factor $\|A\|_F$ or the refined $\mu_p(A)$ parameter (Def.~\ref{def:mu}). 

\begin{theorem}[Implementing block-encodings from quantum data structures {\cite[Theorem 4]{chakraborty2019power}}] 
\label{lem: structure to block encoding}
Let $A \in \C^{n\times m}$.
\begin{enumerate}
    \item Fix $p \in [0,1]$. If $A^{(p)}$, and $(A^{(1-p)})^T$ are stored in quantum accessible data structures, then there exist unitaries $U_R$ and $U_L$ that can be implemented in time $O(\polylog(nm/\epsilon))$ such that $U_R^\dagger U_L$ is a $(\mu_p(A), \ceil{log(n+m+1)}, \epsilon)$-block-encoding of $\overline{A}$.
    \item On the other hand, if $A$ is stored in quantum accessible data structure, then there exist unitaries $U_R$ and $U_L$ that can be implemented in time $O(\polylog(nm/\epsilon))$ such that $U_R^\dagger U_L$ is a $(\norm{A}_F, \ceil{log(n+m+1)}, \epsilon)$-block-encoding of $\overline{A}$.
\end{enumerate}
\end{theorem}
Even in this case, we can efficiently implement block-encodings of subdictionaries $D_\Lambda$.

Our main reason to consider block-encodings is that they can be combined with polynomial transformations of singular values to apply approximate matrix functions. 
\begin{definition}[Singular value transformation]
    Let $A \in \C^{n\times m}$ be a matrix with singular value decomposition $A = \sum_i \sigma_{i}\ket{u_i}\langle{v_i}|$.
    We define singular value transformation by a polynomial $P \in \C[x]$ as
    \begin{equation}
        P^{(SV)}(A) = \begin{cases}
            \sum_i P(\sigma_{i})\ket{u_i}\bra{v_i} & \text{if $P$ is odd}\\
            \sum_i P(\sigma_{i})\ket{v_i}\bra{v_i} & \text{if $P$ is even}.
        \end{cases}
    \end{equation}
    $P$ is odd if all coefficients of even powers of $x$ are $0$ and even if all coefficients of odd powers of $x$ are $0$.
\end{definition}

In practice, SVT is implemented through QSVT~\cite{gilyen2019quantum}, which applies polynomial approximations of desired functions of $A$ by composing controlled block-encodings. 
The circuit complexity scales linearly with the polynomial degree, enabling approximations of spectral projectors, inverses, and more.
We include more details about QSVT circuits, polynomial approximations, and projections in Appendix~\ref{appendix: singular value transformation}.

As a central application, QSVT enables efficient quantum linear system solvers.  
We use the following result from~\citet{chakraborty2022quantum}, which refines the HHL~\cite{HHL2009} approach using block-encodings and variable-time amplitude amplification~\cite{ambainis2012variable}.
We adapt the formulation using our Lemma~\ref{lemma: mini block encoding} to our convenience.

\begin{theorem}[Quantum Linear Systems via QSVT~{\cite[Theorem 28]{chakraborty2022quantum}}]
\label{theorem: QSVT linear systems}
    Let $\epsilon, \delta >0 $. Let $A$ be a matrix such that its non-zero singular values lie in $[\gamma, \alpha]$. Suppose that for $\epsilon = o\left(\frac{\gamma^3\delta}{\alpha^2 \log^2 (\frac{\alpha}{\gamma\delta})}\right)$, we have access to $U_A$ which is an $(\alpha, a, \epsilon)$-block-encoding of $A$, implemented with cost $T_A$. Let there be quantum access to $|\vec{b}\rangle$ in cost $T_b$. Then there exists a quantum algorithm that outputs a state $\overline{\ket{\vec{x}}}$ such that  $\| \overline{\ket{\vec{x}}} - \frac{A^+|\vec{b}\rangle}{\|A^+|\vec{b}\rangle\|}\| \leq \delta$ at a cost of 
    \begin{align}
        O\left(\frac{\alpha}{\gamma} \log(\frac{\alpha}{\gamma}) \left(T_A \log\left(\frac{\alpha}{\gamma\delta}\right) + T_b\right)\right)
    \end{align}
    using $O(\log{}(\frac{\alpha}{\gamma}))$ additional qubits.
\end{theorem}

In summary, block-encodings provide a unifying interface for linear-algebraic primitives in quantum algorithms. 
Whether obtained from oracle-based access (Theorem~\ref{thm: qomp  block encoding from qa}) or from QRAM-based data structures (Theorem~\ref{lem: structure to block encoding}), they allow QSVT to implement functions of $A$, including pseudoinverses as in Theorem~\ref{theorem: QSVT linear systems}. 
This will be the key tool enabling projections in QOMP and coefficient recovery in tomography.

%%%%%%%%%%%%%%%%%%%%%%%%%%%%%%%%%%%%%%%%%%%%%%%%%%%%%%%%%%%%%%%%%%%%%%%%%%%%%%%%%%%%%%%%%%%%%
\subsubsection{Sparse tomography in an orthogonal basis}
We recall a useful result from \citet{van2022quantum}, which addresses sparse tomography in the computational basis when an upper bound on the sparsity $k$ is known.
Intuitively, in this regime, tomography should depend only on the precision $\epsilon$ and the number of significant coefficients $k$ rather than on the full dimension $N$.

\begin{theorem}[Orthogonal sparse tomography~{\cite{van2022quantum}}]
\label{theorem: sparse tomography - orthogonal}
    Let $\ket{\varphi} = \sum_{j \in [d]} \alpha_j\ \ket{j}$ be a quantum state, and $U\ket{0} = \ket{\varphi}$. Let $0 < \delta < 1$, and let $k$ be such that $|\{j \in [d]: |\alpha_j| \geq \epsilon\sqrt{\frac{k}{N}}\}| \leq k$.
    There is a quantum algorithm that, with probability at least $1-\delta$, outputs a $O(k \log(k)\log(1/\delta))$-sparse $\overline{\vec{\alpha}}$ such that $\|\overline{\vec{\alpha}} - \vec{\alpha}\| \leq \epsilon$ using $\Ot(\frac{k}{\epsilon}\mathrm{polylog}(1/\delta))$ applications of $U$ and its inverse, and polynomially many additional gates.
\end{theorem}

The threshold condition guarantees that at most $k$ coefficients of $\ket{\varphi}$ are significantly larger than $\epsilon \sqrt{k/N}$.
In particular, if $\ket{\varphi}$ is exactly $k$-sparse, then the condition holds automatically: precisely $k$ amplitudes are nonzero, and all others vanish.
This theorem establishes that in an orthogonal basis, sparse tomography requires query complexity nearly linear in $k$ and $1/\epsilon$, independent of the ambient dimension $N$.
We will leverage this primitive in our analysis of sparse tomography with non-orthogonal dictionaries, to recover the coefficients once we learn the sparse support.

%%%%%%%%%%%%%%%%%%%%%%%%%%%%%%%%%%%%%%%%%%%%%%%%%%%%%%%%%%%%%%%%%%%%%%%%%%%%%%%%%%%%%%%%%%%%%

\subsubsection{Las Vegas, Monte Carlo, and success probability}
\label{section: amplification of success probabilities}
To conclude this background section, we report some useful tools to tame randomized algorithms.
Las Vegas and Monte Carlo are two terms that indicate two different families of randomized algorithms.
Las Vegas algorithms are algorithms that always output the correct answer, but whose running time is a random variable.
On the other hand, Monte Carlo algorithms have a bounded running time, but their outputs are correct with a certain probability.
We first show how to turn Las Vegas algorithms of known expected time into Monte Carlo, and then discuss how to boost the success probability of Monte Carlo algorithms.

The main tool to turn a Las Vegas algorithm into a Monte Carlo is a famous result in probability.
\begin{theorem}[Markov's inequality]
\label{theorem: markovs inequality}
    Let $X$ be a non-negative random variable and $a > 0$, then $\mathrm{Pr}[X \geq a] \leq \frac{E(X)}{a}$.
\end{theorem}
Indeed, consider a Las Vegas algorithm.
Let $X$ be a random variable expressing the its running time, and $E(X)$ its expected value. 
Terminating the algorithm when the running time exceeds $4E(X)$ returns the correct solution with probability $\geq 2/3$.
In fact, $\mathrm{Pr}[X \geq 4E(X)] \leq 1/4$.

In our work, we prefer to deal with Monte Carlo algorithms, but we will encounter algorithms that have both a random running time and a certain probability of success.
One can turn them into worst-case bounded time algorithm using Markov's inequality at the expense of the success probability.
However, we always find ways to boost success probabilities and still obtain an algorithm with a deterministic worst-case time.
Throughout the work, we will require that an algorithm terminates \emph{with high probability} (e.g., with probability $\geq 2/3$). 
The exact success probability is not relevant for the asymptotic complexity.
In fact, for any algorithm succeeding with probability sufficiently higher than $1/2$, we can efficiently boost the success probability to an arbitrary value $\geq 1-\delta$ by running the routine $O(\log(1/\delta))$ and processing the outputs.
We use two amplification methods, depending on the algorithm's output range.

The first method arbitrarily boosts the success probability of any randomized approximation algorithm which outputs an $\epsilon$-estimate of a real value with high probability by taking the median of the outputs across several runs.
This result is known as powering lemma or median lemma, and we report it using the formulation of \citet{montanaro2015quantum}.

\begin{lemma}[Powering lemma {\cite{jerrum1986random}}]
    \label{lemma: powering lemma (median)}
    Let $\mathcal{A}$ be a (classical or quantum) algorithm which aims to estimate some quantity $\mu$, and whose output $\overline{\mu}$ satisfies $\abs{\mu - \overline{\mu}} \leq \epsilon$ except with probability $\gamma$, for some fixed $\gamma < 1/2$. Then, for any $\delta > 0$, it suffices to repeat $\mathcal{A}$ $O(\log(1/\delta))$ times and take the median to obtain an estimate which is accurate to within $\epsilon$ with probability at least $1 - \delta$.
\end{lemma}

Similarly, if an algorithm ranges over a finite set, we can boost its success probability by majority vote.

\begin{lemma}[Discrete amplification lemma]
    \label{lemma: amplification lemma (majority)}
    Let $\mathcal{A}$ be a (classical or quantum) algorithm whose outputs lie in a finite set $X$.
    On every input, $\mathcal{A}$ returns the correct value except with probability $\gamma$, for some fixed $\gamma < 1/2$. Then, for any $\delta > 0$, it suffices to repeat $\mathcal{A}$ $O(\log(1/\delta))$ times and return the element that appears most often to obtain the correct result with probability at least $1 - \delta$.
\end{lemma}

One way to prove this result is to observe that the expected number of times that we obtain the correct value over $t$ repetitions is $E[Success] = (1-\gamma)t > t/2$ and continue with Chernoff's bound.

Finally, one last useful result is the union bound, or Boole's inequality, which helps us bound the failure probability of a process using the failure probability of its subprocesses.

\begin{theorem}[Union bound]
\label{theorem: union bound}
    Let $X_1, X_2, \dots, X_n$ be a family of events. Then, $\mathrm{Pr}[\cup_{i \in[n]}X_i] \leq \sum_{i \in [n]} \mathrm{Pr}[X_i]$.
\end{theorem}

As an example, imagine an iterative algorithm with failure probability bounded by $1/3$ at each iteration. 
In this case, the overall failure probability of the algorithm is given by the probability that one or more of these failures happen, meaning the union of these events. 
If the algorithm has $K$ iterations, then the union bound helps us bound the total probability of failure by $K/3$.
In general, if an iteration has a failure probability bounded by some $\delta$, then the total failure probability is bounded by $K\delta$.
If we want the overall procedure to succeed with probability $\geq 2/3$ we need to require $\delta < 1/(3K)$, and we can use one of the amplification bounds above to make the overall algorithm terminate with high probability using $O(\log(K))$ overhead per iteration.

In conclusion, we can carry out our algorithms' analysis by considering the success instances and then bound the failure probability using a combination of the two amplification results above plus the union bound.
Furthermore, whenever we have a routine that terminates in random time and we have a classical estimate for expectated time, we can always turn it into an algorithm with a deterministic worst-case running time by terminating it after a certain number of iterations, thanks to Markov's inequality.

%%%%%%%%%%%%%%%%%%%%%%%%%%%%%%%%%%%%%%%%%%%%%%%%%%%%%%%%%%%%%%%%%%%%%%%%%%%%%%%%%%%%%%%%%%%%%
%%%%%%%%%%%%%%%%%%%%%%%%%%%%%%%%%%%%%%%%%%%%%%%%%%%%%%%%%%%%%%%%%%%%%%%%%%%%%%%%%%%%%%%%%%%%%
%%%%%%%%%%%%%%%%%%%%%%%%%%%%%%%%%%%%%%%%%%%%%%%%%%%%%%%%%%%%%%%%%%%%%%%%%%%%%%%%%%%%%%%%%%%%%
%%%%%%%%%%%%%%%%%%%%%%%%%%%%%%%%%%%%%%%%%%%%%%%%%%%%%%%%%%%%%%%%%%%%%%%%%%%%%%%%%%%%%%%%%%%%%
%%%%%%%%%%%%%%%%%%%%%%%%%%%%%%%%%%%%%%%%%%%%%%%%%%%%%%%%%%%%%%%%%%%%%%%%%%%%%%%%%%%%%%%%%%%%%
%%%%%%%%%%%%%%%%%%%%%%%%%%%%%%%%%%%%%%%%%%%%%%%%%%%%%%%%%%%%%%%%%%%%%%%%%%%%%%%%%%%%%%%%%%%%%
%%%%%%%%%%%%%%%%%%%%%%%%%%%%%%%%%%%%%%%%%%%%%%%%%%%%%%%%%%%%%%%%%%%%%%%%%%%%%%%%%%%%%%%%%%%%%
%%%%%%%%%%%%%%%%%%%%%%%%%%%%%%%%%%%%%%%%%%%%%%%%%%%%%%%%%%%%%%%%%%%%%%%%%%%%%%%%%%%%%%%%%%%%%

\section{The Quantum Orthogonal Matching Pursuit (QOMP) algorithm}
\label{sec:qomp}
Orthogonal Matching Pursuit (OMP) is one of the most widely used classical algorithms for sparse approximation. 
It reconstructs a signal iteratively, building its support one element at a time while maintaining the residual orthogonal to the selected dictionary vectors.  
In this section we first recall the structure of classical OMP, emphasizing its distinction from the earlier Matching Pursuit algorithm, and then present our quantum analogue, QOMP. The quantum version inherits the greedy spirit of OMP while addressing the unique challenges of the quantum setting, such as the inability to store or directly update residuals across iterations.  
We will later analyze the cost of each quantum iteration in both the Oracular-Circuit and QRAM models.

%%%%%%%%%%%%%%%%%%%%%%%%%%%%%%%%%%%%%%%%%%%%%%%%%%%%%%%%%%%%%%%%%%%%%%%%%%%%%%%%%%%%%%%%%%%%%
%%%%%%%%%%%%%%%%%%%%%%%%%%%%%%%%%%%%%%%%%%%%%%%%%%%%%%%%%%%%%%%%%%%%%%%%%%%%%%%%%%%%%%%%%%%%%

\subsection{The classical Orthogonal Matching Pursuit}
Orthogonal Matching Pursuit (OMP)~\cite{pati1993orthogonal} is a classical greedy algorithm for sparse approximation.  
It operates iteratively: starting from the full signal as an initial residual, at each step it selects one dictionary element (also called atom) to add to the support, updates the approximation, and redefines the residual as the part of the signal not yet explained by the span of the selected atoms.  

OMP improves on the earlier Matching Pursuit algorithm of \citet{mallat1993matching}, where atoms may be reselected because the residual is not fully re-optimized at each step.
In contrast, OMP recomputes the orthogonal projection of the signal onto the span of the active atoms after every update.
This guarantees that no atom is chosen twice, keeps the residual orthogonal to the current support, and underlies OMP’s stronger recovery guarantees under incoherence assumptions.

We present two equivalent formulations in Algorithms~\ref{alg:omp}--\ref{alg:omp v2}.  
The first emphasizes the least-squares update of the coefficients, while the second makes explicit the projection-based residual:
\begin{align}
    \vec{r}^{(k)} = \vec{s} - D_{\Lambda^{(k)}}D_{\Lambda^{(k)}}^+\vec{s},
\end{align}
where $\Lambda^{(k)}$ is the support set after $k$ iterations and $D_{\Lambda^{(k)}}$ the associated subdictionary.  
From this point on, we omit the superscript $k$ and treat $\Lambda$ as the current support, with equalities interpreted as assignment when the context is iterative.
This projection-based formulation is the one we will adopt in the quantum setting, as it enables an error-resetting strategy: the residual is always recomputed directly from the signal and the support, rather than accumulated across steps.

The computational cost of OMP is dominated by two tasks: (i) the \emph{sweep stage}, computing inner products of the residual with all dictionary atoms to select the next index, and (ii) the orthogonal projection onto the active set.  
In a naive implementation, one iteration costs
\begin{align}
    O(nm + nk^2 + k^3),
\end{align}
where $n$ is the signal dimension, $m$ the dictionary size, and $k$ the iteration count.  
Using more advanced techniques such as the Matrix Inversion Lemma~\cite{sturm2012comparison}, the cost can be reduced to
\begin{align}
    O(nk + mk),
\end{align}
at the expense of additional memory.  
Despite algorithmic optimizations, each iteration remains dominated by the sweep stage (computing correlations with all atoms and selecting the greatest) and the orthogonal projection onto the active set.  
These are precisely the operations we target for quantum acceleration.

\begin{figure}[t]
\centering
    \begin{minipage}[t]{0.48\linewidth}
        \begin{algorithm}[H]
            \caption{Orthogonal Matching Pursuit (OMP)}
            \label{alg:omp}
            \begin{algorithmic}[1]
            
                \Statex \textbf{Input} Signal $\vec{s} \in \mathbb{C}^{n}$, dictionary $D \in \mathbb{C}^{n \times m}$, sparsity threshold $L \in \mathbb{N}$, residual threshold $\epsilon \in \mathbb{R}_{>0}$.
                \Statex \textbf{Output} Vector $\vec{x} \in \mathbb{C}^m$ s.t. $\|\vec{s}-D\vec{x}\| \leq \epsilon$ and $\|\vec{x}\|_0 \leq L$ or FAIL if exceeding $L$ iterations.
                \vspace{5pt}
                \hrule
                \vspace{5pt}
                
                \State Initialize $\vec{r} = \vec{s}$,\, $\vec{x} = 0^{\otimes m},\, k=0, \, \Lambda=\emptyset$ 
                \While {not ($k>L \;\mathrm{or} \; \|\vec{r}\|_2 \leq \epsilon$)}
                    \State $j^* = \argmax_{j \in [m] \setminus \Lambda} ( |\langle \vec{d}_j, \vec{r} \rangle| )$
                    \State $\Lambda= \Lambda \cup \, j^*$
                    \State $\vec{x} = \argmin_{\vec{x}} \norm{\vec{s}-D_\Lambda \vec{x}}_2^2$
                    \State $\vec{r}=\vec{s}-D_\Lambda\vec{x}$
                    \State $k = k + 1$
                \EndWhile
                \State Output $\vec{x}$ if $k \leq L$; Else FAIL.
        \end{algorithmic}
        \end{algorithm}
    \end{minipage}%
    \hfill
    \begin{minipage}[t]{0.48\linewidth}
        \begin{algorithm}[H]
        \caption{Alternative OMP formulation}
        \label{alg:omp v2}
        \begin{algorithmic}[1]
        
            \Statex \textbf{Input} Signal $\vec{s} \in \mathbb{C}^{n}$, dictionary $D \in \mathbb{C}^{n \times m}$, sparsity threshold $L \in \mathbb{N}$, residual threshold $\epsilon \in \mathbb{R}_{>0}$.
            \Statex \textbf{Output} Vector $\vec{x} \in \mathbb{C}^m$ s.t. $\|\vec{s}-D\vec{x}\| \leq \epsilon$ and $\|\vec{x}\|_0 \leq L$ or FAIL if exceeding $L$ iterations.
            \vspace{5pt}
            \hrule
            \vspace{5pt}
            
            \State Initialize $\|\vec{r}\|_2 = \|\vec{s}\|_2,\, k=0, \, \Lambda=\emptyset$ 
            \While {not ($k>L \;\mathrm{or} \; \|\vec{r}\|_2 \leq \epsilon$)}
                \For {\textbf{all} $j \in [m] \setminus \Lambda$}
                    \If{$k == 0$}
                        \State $z_j = |\inner{\vec{d}_j}{\vec{s}}|$
                    \Else
                        \State $z_j = |\inner{\vec{d}_j}{\vec{s}-D_\Lambda D_\Lambda^+\vec{s}}|$
                    \EndIf
                \EndFor
                \State $j^* = \argmax_j ( z_j )$. 
                \State $\Lambda= \Lambda \cup \, j^*$
                \State $\|\vec{r}\|_2 = \|\vec{s}-D_\Lambda D_\Lambda^+\vec{s}\|_2$
                \State $k = k + 1$
            \EndWhile
            \State Output $\vec{x} = \argmin_{\vec{x}} \norm{\vec{s}-A\vec{x}}_2^2$ if $k \leq L$; Else FAIL.
        \end{algorithmic}
        \end{algorithm}
    \end{minipage}
\end{figure}

%%%%%%%%%%%%%%%%%%%%%%%%%%%%%%%%%%%%%%%%%%%%%%%%%%%%%%%%%%%%%%%%%%%%%%%%%%%%%%%%%%%%%%%%%%%%%
%%%%%%%%%%%%%%%%%%%%%%%%%%%%%%%%%%%%%%%%%%%%%%%%%%%%%%%%%%%%%%%%%%%%%%%%%%%%%%%%%%%%%%%%%%%%%

\subsection{Quantum Orthogonal Matching Pursuit}
\label{section: quantum orthogonal matching pursuit}
QOMP is the quantum analogue of OMP, built on the projection-based formulation of Algorithm~\ref{alg:omp v2}.
In this view, the residual at each step is defined by
\begin{align}
    \vec{r} = \vec{s} - D_{\Lambda}D_{\Lambda}^+\vec{s},
\end{align}
where $\Lambda$ is the current support.

This formulation is central to our quantum design: it enables an \emph{error-resetting strategy}, where each residual is recomputed as an exact projection depending only on the input state and the active support.

QOMP preserves the greedy structure of OMP, but re-engineers its iteration body with quantum subroutines, enabling handling quantum signals and dictionaries.
A classical controller orchestrates the algorithm, updating the support set and managing iteration counts, while the quantum device executes the expensive primitives: inner product estimation, maximum-finding, block-encoded projections, and residual norm estimation.
The result is a hybrid scheme that preserves the spirit of OMP while leveraging quantum resources to accelerate its computational bottlenecks.

\smallskip
\textbf{1. Initialization.} The classical computer initializes two variables, an iteration counter and the set of selected atoms $k=0; ~ \Lambda = \emptyset$.

\smallskip
\textbf{2. Atom selection.} This step is the main body of an iteration and requires executing multiple quantum circuits.
The task consists of computing multiple inner products and extracting the index of the one basis vector having the highest overlap with the residual, in absolute value.

The main difficulty is to prepare access to an oracle that allows querying the absolute values of the inner products
\begin{align}
\label{eq: Oi}
    O_i: \ket{j}\ket{0} \rightarrow \ket{j}\ket{z_j},
\end{align}
where $z_j$ approximates $|\inner{\vec{d}_j}{\vec{r}}|$ to error $\epsilon_i$ (i.e., $|z_j - |\inner{\vec{d}_j}{\vec{r}}|| \leq \epsilon_i$) with high probability.
Using this oracle and the access to the complement set $\overline{\Lambda} = [m]\setminus \Lambda$ (Def.~\ref{def: qomp  quantum access sets}), one can use the \emph{Finding the maximum with an approximate unitary} algorithm of Corollary~\ref{coro: finding approximate maximum} to identify the index $j$ of the best basis vector.

To prepare access to $O_i$ we leverage the following equation
\begin{align}
    z_j \simeq |\inner{\vec{d}_j}{\vec{r}}| = |\inner{\vec{d}_j}{\vec{s}} - \inner{\vec{d}_j}{\vec{\phi}}|,
\end{align}
where $\vec{\phi} = D_\Lambda D_\Lambda^+ \vec{s}$.

The strategy is to prepare the real and imaginary part of the two inner products in four registers and combine them with in a fifth register through arithmetic expressions, to reproduce the formula
\begin{align}
    |\inner{\vec{d}_j}{\vec{r}}| = (\|\vec{s}\|\Re[\bracket{\vec{d}_j}{\vec{s}}] - \|\vec{\phi}\|\Re[\bracket{\vec{d}_j}{\vec{\phi}}])^2 + (\|\vec{s}\|\Im[\bracket{\vec{d}_j}{\vec{s}}] - {\|\vec{\phi}\|}\Im[\bracket{\vec{d}_j}{\vec{\phi}}])^2.
\end{align}

First, we can compute $\bracket{\vec{d}_j}{\vec{s}}$ using the \emph{Complex inner product estimation} of Corollary~\ref{coro: complex inner} with quantum access to the dictionary via $U_D$, and to the signal via $U_s$.
This way, we can implement the mappings
\begin{align}
    \ket{j}\ket{0} &\rightarrow \ket{j}\ket{\Re[z_{1j}]},\\
    \ket{j}\ket{0} &\rightarrow \ket{j}\ket{\Im[z_{1j}]},
\end{align}
where $\Re[z_{1j}]$ approximates the real part of $\bracket{\vec{d}_j}{\vec{s}}$ and $\Im[z_{1j}]$ the imaginary part.
Then, we use the same method on different registers to compute $\bracket {\vec{d}_j}{\vec{\phi}}$ using quantum access to the dictionary $U_D$ and to an approximation of $\vec{\phi}$ via a unitary $U_\phi$, which we will discuss in a moment.
Again, using Corollary~\ref{coro: complex inner}, we can implement the mappings
\begin{align}
    \ket{j}\ket{0} &\rightarrow \ket{j}\ket{\Re[z_{2j}]},\\
    \ket{j}\ket{0} &\rightarrow \ket{j}\ket{\Im[z_{2j}]},
\end{align}
where $\Re[z_{2j}]$ approximates the real part of $\bracket{\vec{d}_j}{ \vec{\phi}}$ and $\Im[z_{2j}]$ the imaginary part.
Finally, through these mappings, and access to the classical norm of $\|\vec{s}\|$ and to an approximation $\overline{\|\vec{\phi}\|}$ of the norm of $\vec{\phi}$, we can implement 
\begin{align}
    z_j = (\|\vec{s}\|\Re[z_{1j}] - \overline{\|\vec{\phi}\|}\Re[z_{2j}])^2 + (\|\vec{s}\|\Im[z_{1j}] - \overline{\|\vec{\phi}\|}\Im[z_{2j}])^2
\end{align}
with some arithmetic.
This whole procedure effectively implements the oracle $O_i$ of Eq. (\ref{eq: Oi}) coherently.
To regulate the probability of failure, we can use the \emph{Powering Lemma} of Lemma~\ref{lemma: powering lemma (median)}.

To conclude the implementation of $O_i$ and the atom selection step, we need to discuss how to create access to $\vec{\phi} = D_\Lambda D_\Lambda^+ \vec{s}$ and estimate its norm, which is necessary for each iteration following the first one.
Using quantum access to the dictionary via the unitary $U_D$ and to the set $\Lambda$ via the unitary $U_\Lambda$, we can create quantum access to the matrix $D_\Lambda$ (Def.~\ref{def: efficient quantum access matrix}), and consequently, a block-encoding of $D_\Lambda$ (Theorem~\ref{thm: qomp  block encoding from qa}).
With access to the unitary block-encoding and to the signal via $U_s$, we can use the following result.

\begin{theorem}[Column space projection]
\label{theorem: column space projection}
    Let $\epsilon > 0$ be a precision parameter. 
    Let $U_A$ be a $(\alpha, q, \epsilon_A)$-block-encoding of a matrix $A \in \C ^{n\times m}$, implementable in time $T_A$, and let a lower bound $\gamma \leq \sigma_{\min}(A)$ be known. 
    Let there be quantum access to a vector $\Vec{x} \in \C^{n}$ of known norm $\norm{\vec{x}}_2$ in time $T_x$ via a unitary $U_x$. 
    Then, there exists a constant $c \in \R^+$ such that if $\epsilon_A \leq \frac{\norm{AA^+\vec{x}}^2\gamma^2 \epsilon^2}{c\norm{\vec{x}}^2\alpha\log^2(\norm{\vec{x}}/(\norm{AA^+\vec{x}}\epsilon))}$ there are quantum algorithms that:
    \begin{enumerate}
        \item Create a quantum state $|\vec{\phi}\rangle$ such that $\norm{|\vec{\phi}\rangle - \ket{AA^+ \vec{x}}}_2 \leq \epsilon$ in expected time $\widetilde{O}\left(\frac{\norm{\vec{x}}}{\norm{AA^+ \vec{x}}}( \frac{\alpha}{\gamma} T_A + T_x)\right)$ if $\norm{AA^+\vec{x}} \neq 0$ and otherwise runs forever. 
        \item Produce an estimate $t$ such that $\abs{t - \norm{AA^+\vec{x}}_2}\leq \epsilon$ with high probability  in time $\widetilde{O}\left(\frac{1}{\epsilon}(\frac{\alpha}{\gamma} T_A + T_x)\right)$;
        \item Produce an estimate $t$ such that $\abs{t - \norm{AA^+\vec{x}}_2}\leq \epsilon\norm{AA^+ \vec{x}}_2$ with high probability \\in expected time $\widetilde{O}\left(\frac{1}{\epsilon}\frac{\norm{\vec{x}}}{\norm{AA^+ \vec{x}}}(\frac{\alpha}{\gamma} T_A + T_x)\right)$.
    \end{enumerate}
\end{theorem}

This theorem allows us to provide access to $U_\phi$ and to estimate the norm $\|\vec{\phi}\|$, concluding the atom selection process.
The main intuition behind this result is that $D_\Lambda D_\Lambda^+ = UU^\dagger$, where $D_\Lambda = U\Sigma V^\dagger$ and $D_\Lambda^+ = V\Sigma^{-1}U^\dagger$ are singular value decompositions.
We can then apply a polynomial approximation of a constant function $f(x) = 1$ in the interval $[\frac{\gamma}{\alpha}, 1]$ to the singular values of $D_\Lambda$ and $D_\Lambda$ using Quantum Singular Value Transformation (QSVT)~\cite{gilyen2019quantum, chakraborty2019power} on their block encodings. 
We finally apply the block-encoding of $UU^\dagger$ to the state $\ket{\vec{s}}$ and estimate the norm using the \emph{amplitude estimation} routine from Theorem~\ref{thm: amplitude estimation amplitude} or amplify the relevant quantum state $|\vec{\phi}\rangle$ with \emph{Fixed point amplitude amplification} from Theorem~\ref{thm: qomp  fixed point amp amp}.
We defer the full proof of Theorem~\ref{theorem: column space projection} to Appendix~\ref{appendix: column space}.

Once we obtain the index of the best atom for the current iteration, the classical computer can proceed to update the set of chosen atoms $\Lambda = \Lambda \cup j^*$, update $U_{\Lambda}$ and $U_{\overline{\Lambda}}$, and increment the iteration counter $k = k+1$.

\smallskip
\textbf{3. Exit condition.}
The exit condition is $(k > L \text{ or } \|\vec{r}\|_2 \leq \epsilon)$.
The classical computer can easily evaluate the first inequality, as it stores both the iteration counter and the threshold. 
On the other hand, it will require the execution of quantum circuits to estimate $\|\vec{r}\|$.

The computation of the norm is based on
\begin{align}
\label{eq: norm estimation}
    \|\vec{r}\| = \|\vec{s} - \vec{\phi}\| = \| \|\vec{s}\| \ket{\vec{s}} - \|\vec{\phi}\| |\vec{\phi}\rangle  \|.
\end{align}

We have access to $\ket{\vec{s}}$ through $U_s$ and we have a classical value for $\|\vec{s}\|$.
Moreover, using \emph{Column space projection} from Theorem~\ref{theorem: column space projection}, we have access to an approximation of $\ket{\phi}$ via $U_\phi$, and to a classical estimate of $\|\vec{\phi}\|$. 
Using these tools, we can compute the residual's norm through the following result.

\begin{figure}[t]
    \centering
    \scalebox{1}{ 
        \Qcircuit @C=1.0em @R=0.2em @!R { \\
            \nghost{{\ket{0}} :  } & \lstick{{\ket{0}} :  } & \qw & \gate{\mathrm{U_{v}}} & \gate{\mathrm{U_{c}}} & \qw & \qw & \qw & \qw & \qw \\
            \nghost{{\ket{0}} :  } & \lstick{{\ket{0}} :  } & \gate{\mathrm{H}} & \ctrl{-1} & \ctrlo{-1} & \ctrl{1} & \ctrlo{1} & \gate{\mathrm{H}} & \qw & \qw\\
            \nghost{{\ket{0}} :  } & \lstick{{\ket{0}} :  } & \qw & \qw & \qw & \gate{\mathrm{R_v}} & \gate{\mathrm{R_c}} & \qw & \qw & \qw \\
        \\ }
    }
    \caption{State preparation circuit for estimating $\norm{\vec{v}-\vec{c}}$, when $\norm{\vec{v}}, \norm{\vec{c}}$ are classically known. The most significant qubit is the one at the top.
    The gate $\mathrm{R_v}$ performs the rotation $\ket{0} \rightarrow \sqrt{1-\frac{1}{\norm{\vec{c}}^2}}\ket{0} + \frac{1}{\norm{\vec{c}}}\ket{1}$, and similarly $\mathrm{R_c}$ performs $\ket{0} \rightarrow \sqrt{1-\frac{1}{\norm{\vec{v}}^2}}\ket{0} + \frac{1}{\norm{\vec{v}}}\ket{1}$.
    At the end of the circuit, the amplitude of the two least significant qubits in the state $\ket{1}\ket{1}$ is $\frac{\norm{\vec{v}-\vec{c}}}{2\norm{\vec{v}}\norm{\vec{c}}}$.}
    \label{fig: weighted norm estimation circuit}
\end{figure}
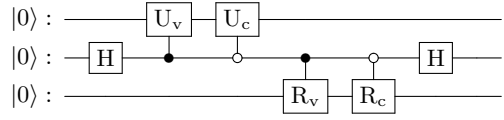

\begin{theorem}[Weighted Euclidean distance estimation]
\label{theorem: weighted euclidean distance estimation}
    Let there be quantum access to two unit vectors $\vec{v} \in \mathbb{C}^{n}$ and $\vec{c} \in \mathbb{C}^{n}$, through unitaries
    $U_v: \ket{0} \rightarrow \ket{\vec{v}}$ and $U_c: \ket{0} \rightarrow \ket{\vec{c}}$ that run in time $T_v$ and $T_c$.
    Let $\alpha, \beta \in \C$ be two weights.
    Then, for any $\delta > 0$ and $\epsilon>0$, there exists a quantum algorithm that computes an estimate of
    $d = \norm{\alpha \vec{v} - \beta \vec{c}}$,
    such that
    $\abs{\overline{d}-d} \leq  \epsilon$
    with probability greater than $1-\delta$,
    in time $O\left((T_a + T_b)\frac{|\alpha||\beta|}{\epsilon}\log(1/\delta)\right)$.    
\end{theorem}

The algorithm consists of executing amplitude estimation on the circuit described in Figure~\ref{fig: weighted norm estimation circuit}. 
Appendix~\ref{appendix: euclidean distance estimation} details the analysis of the routine.
This algorithm allows us to obtain an estimate of $\|\vec{r}\|$ and conclude the evaluation of the error condition.

\smallskip
\textbf{4. Output.} When the exit condition is met, QOMP outputs the set of chosen atoms $\Lambda$, or FAIL if the number of iterations exceeded $L$.

\subsubsection{Iteration cost in the Oracular-Circuit model}
Chaining together these steps, we can bound the expected cost of a single QOMP iteration in the Oracular-Circuit model.
We denote by $T_s$ the cost of accessing the signal through $U_s$, by $T_D$ the cost of accessing the dictionary through $U_D$, by $T_\Lambda$ and $T_{\overline{\Lambda}}$ the costs of accessing the active set $\Lambda$ and its complement, and by $T_U$ the classical time required to update the circuits for $U_\Lambda$ and $U_{\overline{\Lambda}}$ when inserting a new element into $\Lambda$.

\begin{theorem}[QOMP Iteration's Cost]
\label{theorem: QOMP iteration cost}
    Let there be quantum access to the dictionary $D \in \C^{n\times m}$ (Def.~\ref{def: qomp  quantum access dictionary}), the signal $\vec{s} \in \C^n$ (Def.~\ref{Def:efficient quantum access vector}) and the sets $\Lambda$, $\overline{\Lambda}$ (Def.~\ref{def: qomp  quantum access sets}). 
    Let $\norm{\vec{s}} \geq 1$, let $\epsilon_i, \epsilon_f > 0 $ be precision parameters, and $\gamma \leq \sigma_{\min}(D_\Lambda)$ a lower bound on the smallest singular value of the current matrix $D_\Lambda$, whose columns are the chosen atoms in $\Lambda$.
    With high probability, at the $k^{\mathrm{th}}$ iteration, the QOMP algorithm selects the atom
    \begin{align}
        j^* = \argmax_{j \in {\overline{\Lambda}}} \left(\abs{\inner{\vec{d}_j}{\vec{r}}} - 2\epsilon_i \right) \quad \text{ s.t. } \quad \forall j \in \overline{\Lambda}: \abs{\overline{\inner{\vec{d}_j}{\vec{r}}} - \inner{\vec{d}_j}{\vec{r}}} \leq \epsilon_i,
    \end{align}
    and evaluates the exit condition on an estimate $\overline{\norm{\vec{r}}}_2$ such that $\abs{\overline{\norm{\vec{r}}}_2 - \norm{\vec{r}}_2} \leq \epsilon_f$, all in expected time
    \begin{align}
        \Ot \left( \sqrt{m} T_{\overline{\Lambda}} +
        \norm{\vec{s}}^2 \left(\frac{\sqrt{m}}{\epsilon_i} + \frac{1}{\epsilon_f}\right)
        \left(T_s + 
        (T_D + T_{\Lambda})\frac{\sqrt{k}}{\gamma}\right)
        \right)
    \end{align}
    plus additional classical time $T_U$ to update quantum access to $\Lambda$ and $\overline{\Lambda}$ (Def.~\ref{def: qomp  quantum access sets}).
\end{theorem}

A detailed derivation of this bound, including the handling of approximation errors, is given in Appendix~\ref{appendix: QOMP error and running time analysis}.
Here we emphasize two features that are central to the efficiency of QOMP.
First, errors from approximate subroutines do not propagate across iterations: the residual is always recomputed as a projection depending only on the input signal and the current support. 
\emph{The only way an error carries forward is through the unlikely event of selecting an incorrect atom.}
Second, the complexity scales with the conditioning of the subdictionary. 
Since $D_\Lambda$ consists of columns of $D$, one may always take $\gamma = \sigma_{\min}(D)$ as an iteration-independent bound, ensuring uniform guarantees across iterations. 
Tighter bounds on $\sigma_{\min}(D_\Lambda)$ can be assumed or computed if desired, at the expense of additional classical or quantum computation.

%%%%%%%%%%%%%%%%%%%%%%%%%%%%%%%%%%%%%%%%%%%%%%%%%%%%%%%%%%%%%%%%%%%%%%%%%%%%%%%%%%%%%%%%%%%%%
%%%%%%%%%%%%%%%%%%%%%%%%%%%%%%%%%%%%%%%%%%%%%%%%%%%%%%%%%%%%%%%%%%%%%%%%%%%%%%%%%%%%%%%%%%%%%

\subsubsection{Iteration cost in the QRAM model}
We now analyze the iteration cost of QOMP in the QRAM model.
The starting point is Theorem~\ref{theorem: QOMP iteration cost}, which bounds the runtime in the Oracular-Circuit setting in terms of the access costs $T_s$, $T_D$, $T_\Lambda$, $T_{\overline{\Lambda}}$, and $T_U$ together with the block-encoding normalization factor.
In the QRAM model, these access costs are polylogarithmic.

Moreover, QRAM access enables block-encodings with improved normalization.
In the Oracular-Circuit model, the normalization factor is $|A|_F$, but in the QRAM model it can be reduced to $\mu(A)$ (Definition~\ref{def:mu}) using the decomposition of Theorem~\ref{lem: structure to block encoding}.
This refinement lowers the dependence of the projection step on the size of $A$, since the QSVT polynomial approximations now scale with $\mu(D_\Lambda)$ rather than the Frobenius norm.

\begin{corollary}[QOMP Iteration's Cost in the QRAM model]
\label{corollary: iteration cost QRAM}
    In the QRAM cost model, the $k^\mathrm{th}$ iteration of QOMP, with the same guarantees as above, takes expected time
    \begin{align}
        \Ot\left(\norm{\vec{s}}^2\frac{\mu(D_\Lambda)}{\gamma} \left( \frac{\sqrt{m}}{\epsilon_i} + \frac{1}{\epsilon_f}\right)\right).
    \end{align}
\end{corollary}
\begin{proof}
    Substituting $T_s, T_D, T_\Lambda, T_{\overline{\Lambda}}, T_U \in \widetilde{O}(1)$ into the bound of Theorem~\ref{theorem: QOMP iteration cost} eliminates the explicit dependence on data-access costs.
    The remaining dependence comes from the block-encoding normalization factor.
    By Theorem~\ref{lem: structure to block encoding}, QRAM-based block-encodings of $D_\Lambda$ admit normalization $\alpha = \mu(D_\Lambda)$ rather than $\alpha = \|D_\Lambda\|_F$, yielding the stated complexity. 
    The normalization $\mu(D_\Lambda)$ can be retrieved by classically stored data structures.
\end{proof}

\begin{table}[ht]
    \centering
        \begin{tabular}{ |c|c|c| }
         \hline
         Algorithm & Time complexity & Memory\\
         \hline
         Naive & $nm + nk + nk^2 + k^3$ & $nm$\\
         \hline
         Chol-1 & $nm + nk + k^2$ & $m^2 + nm + k + k^2$\\
         \hline
         Chol-2 & $mk + k^2$ & $m^2 + nm + k + k^2$\\
         \hline
         QR-1 & $nm + nk$ & $nm + nk + k^2$\\
         \hline
         QR-2 & $nk + mk + k^2$ & $m^2 + nm + nk + k^2$\\
         \hline
         MIL & $nk+mk$ & $m^2 + nm + nk$\\
         \hline
         \cellcolor{green!25} QOMP (This work) & \cellcolor{green!25} $\norm{\vec{s}}^2 \frac{\mu(D_\Lambda)}{\gamma}\left( \frac{\sqrt{m}}{\epsilon_i} + \frac{1}{\epsilon_f}\right)$ & \cellcolor{green!25} $nm\log(nm)$\\
        \hline
        \end{tabular}
    \caption{Asymptotic iteration costs of different classical implementations of OMP~\cite{sturm2012comparison} vs QOMP.
    The memory cost of $\mathrm{QOMP}$ is expressed in number of QRAM cells.}
    \label{table: iteration cost quantum}
\end{table}

This result highlights the power of quantum-accessible data structures.
Table~\ref{table: iteration cost quantum} compares the resulting bound against several classical implementations of OMP reported by \citet{sturm2012comparison}. Naive methods scale as $O(nm)$ per iteration, while optimized variants such as those using the Matrix Inversion Lemma achieve $O(nk+mk)$. By contrast, QOMP achieves sublinear scaling in $m$ through its $\sqrt{m}$ dependence, at the price of a QRAM memory requirement of $O(nm\log(nm))$ cells.

The normalization parameter $\mu(D_\Lambda)$ is always upper bounded by $\|D_\Lambda\|_F = \sqrt{k}$, while the conditioning parameter can be set to $\gamma = \sigma_{\min}(D)$ for a fixed dictionary, or estimated more carefully at additional cost.
Approximation errors scale with the signal norm, so rescaling the input simply rescales the tolerated precision, and the two effects typically balance.
Under reasonable error tolerances, and provided the dictionary is reasonably well-conditioned, the iteration cost of QOMP reduces to roughly $\Ot(\sqrt{km})$. This represents a genuine polynomial improvement over naive $O(nm)$ methods and nearly quadratic savings compared with the fastest classical implementations.
In the high-dimensional regime where $m$ is large, this positions QOMP as a genuine acceleration over classical algorithms, contingent on QRAM access times approaching those of classical RAM - a regime unlikely in the near term but conceivable in the longer horizon of scalable fault-tolerant quantum architectures.

%%%%%%%%%%%%%%%%%%%%%%%%%%%%%%%%%%%%%%%%%%%%%%%%%%%%%%%%%%%%%%%%%%%%%%%%%%%%%%%%%%%%%%%%%%%%%
%%%%%%%%%%%%%%%%%%%%%%%%%%%%%%%%%%%%%%%%%%%%%%%%%%%%%%%%%%%%%%%%%%%%%%%%%%%%%%%%%%%%%%%%%%%%%
%%%%%%%%%%%%%%%%%%%%%%%%%%%%%%%%%%%%%%%%%%%%%%%%%%%%%%%%%%%%%%%%%%%%%%%%%%%%%%%%%%%%%%%%%%%%%
%%%%%%%%%%%%%%%%%%%%%%%%%%%%%%%%%%%%%%%%%%%%%%%%%%%%%%%%%%%%%%%%%%%%%%%%%%%%%%%%%%%%%%%%%%%%%
%%%%%%%%%%%%%%%%%%%%%%%%%%%%%%%%%%%%%%%%%%%%%%%%%%%%%%%%%%%%%%%%%%%%%%%%%%%%%%%%%%%%%%%%%%%%%
%%%%%%%%%%%%%%%%%%%%%%%%%%%%%%%%%%%%%%%%%%%%%%%%%%%%%%%%%%%%%%%%%%%%%%%%%%%%%%%%%%%%%%%%%%%%%
%%%%%%%%%%%%%%%%%%%%%%%%%%%%%%%%%%%%%%%%%%%%%%%%%%%%%%%%%%%%%%%%%%%%%%%%%%%%%%%%%%%%%%%%%%%%%
%%%%%%%%%%%%%%%%%%%%%%%%%%%%%%%%%%%%%%%%%%%%%%%%%%%%%%%%%%%%%%%%%%%%%%%%%%%%%%%%%%%%%%%%%%%%%

\section{Exact Sparse Recovery with QOMP}
\label{sec:guarantees}
In this section, we analyze the ability of Orthogonal Matching Pursuit and its quantum analogue, QOMP, to recover the exact sparse representation of a signal.
We begin by reviewing the classical theory based on mutual incoherence, which provides clean and widely adopted guarantees. These results set the stage for our quantum extension.

In the exact recovery problem, we are given a dictionary $D \in \C^{n \times m}$ and a signal $\vec{s} \in \C^n$, and we seek the sparsest coefficient vector $\vec{x} \in \C^m$ such that
\begin{align}
\label{eq:p0}
    \vec{x}^* = \argmin_{\vec{x} \in \C^m} \|\vec{x}\|_0
    \quad \text{subject to} \quad D\vec{x} = \vec{s}.
\end{align}
Let $\Lambda_{\mathrm{opt}} \subset [m]$ denote the support of $\vec{x}^*$, i.e., the indices of the atoms used in the unique optimal representation.
We write $A_{\mathrm{opt}}$ for the submatrix of $D$ containing the columns indexed by $\Lambda_{\mathrm{opt}}$ (with zeros elsewhere), so that $A_{\mathrm{opt}}\vec{x} = \vec{s}$, and $B_{\mathrm{opt}}$ for the complementary submatrix (i.e., $D = A_{\mathrm{opt}} + B_{\mathrm{opt}}$).

%%%%%%%%%%%%%%%%%%%%%%%%%%%%%%%%%%%%%%%%%%%%%%%%%%%%%%%%%%%%%%%%%%%%%%%%%%%%%%%%%%%%%%%%%%%%%
%%%%%%%%%%%%%%%%%%%%%%%%%%%%%%%%%%%%%%%%%%%%%%%%%%%%%%%%%%%%%%%%%%%%%%%%%%%%%%%%%%%%%%%%%%%%%

\subsection{Classical recovery guarantees and mutual incoherence}
Sparse recovery has been studied extensively in the last two decades, both in compressed sensing and in approximation theory.
A central line of work has characterized the conditions under which greedy methods such as OMP provably recover the optimal support in polynomial time.
The first such guarantee is the Exact Recovery Condition (ERC) of \citet{tropp2004greed}, which formalizes the requirement that OMP selects a correct atom at every iteration.

\begin{theorem}[Exact Recovery for OMP]
\label{thm: qomp  exact recovery omp}
    A sufficient condition for OMP to recover the sparsest representation of the input signal is that $\max_{\vec{\psi}} \|A_{\mathrm{opt}}^+\vec{\psi}\|_1 < 1$, where $\psi$ ranges over the columns of $B_{\mathrm{opt}}$.
\end{theorem}

Intuitively, this condition ensures that the atoms in the optimal support dominate the correlations with the residual, so that OMP will not be misled into selecting an atom outside $\Lambda_{\mathrm{opt}}$.

Since the optimal support is unknown a priori, the ERC is often specialized to dictionary-wide properties.
The most common is \emph{mutual incoherence}, which measures the largest normalized correlation between distinct atoms.

\begin{definition} [Mutual Incoherence]
For a set of vectors $\vec{x}_i \in \mathbb{C}^m$, $i \in [n]$, the mutual incoherence $\mu \in \R^+$ is the largest absolute value of normalized correlation between these vectors: $\mu= \mathrm{max}_{i,j \in [n], i \neq j} \frac{\abs{\inner{\vec{x}_i}{\vec{x}_j}}}{\norm{\vec{x}_i}_2\norm{\vec{x}_j}_2}$.
\end{definition}

When $\mu$ is small, atoms are nearly orthogonal, which makes them easier to distinguish. The following corollary gives a clean incoherence-based recovery condition.

\begin{corollary}[MI condition for OMP]
\label{coro: qomp incoherence condition omp 2}
    OMP recovers every superposition of $K$ atoms from $D$ in $K$ iterations if  
    \begin{align}
        K < \frac{1}{2}(\mu^{-1} + 1).
    \end{align}
\end{corollary}

This recovery condition is sharp in the general case, as it would fail for any $\ceil{\frac{1}{2}(\mu^{-1} + 1)}$ atoms from an equiangular tight frame with $m=n+1$ vectors~\cite{tropp2004greed}.
Moreover, this condition also guarantees uniqueness of the recovered solution.

%%%%%%%%%%%%%%%%%%%%%%%%%%%%%%%%%%%%%%%%%%%%%%%%%%%%%%%%%%%%%%%%%%%%%%%%%%%%%%%%%%%%%%%%%%%%%
%%%%%%%%%%%%%%%%%%%%%%%%%%%%%%%%%%%%%%%%%%%%%%%%%%%%%%%%%%%%%%%%%%%%%%%%%%%%%%%%%%%%%%%%%%%%%

\subsection{Quantum recovery guarantees}
The recovery analysis of QOMP builds on the classical theory of OMP, but its adaptation to the quantum setting requires new ingredients.
In the classical case, the Exact Recovery Condition (Theorem~\ref{thm: qomp exact recovery omp}) and its incoherence-based corollary (Corollary~\ref{coro: qomp incoherence condition omp 2}) ensure that OMP selects a correct atom at every iteration, relying on exact evaluations of inner products between the residual and the dictionary atoms.

QOMP, in contrast, can only access approximate inner products, obtained through quantum estimation routines.
The central technical issue is therefore to prove that these approximation errors do not accumulate across iterations, and that the greedy selection rule continues to succeed under bounded quantum error.
This is made possible by the algorithm’s \emph{error-resetting strategy}: rather than updating the residual incrementally, QOMP defines it afresh at each iteration as the orthogonal projection of the signal onto the complement of the chosen support.
As a consequence, no error carries over from earlier steps; the only approximation that matters at iteration $k$ is the precision of the oracle used to compare candidate atoms.

Formally, the atom selection oracle $O_i$ of Eq.~(\ref{eq: Oi}) returns an estimate of the correlations $|\bracket{\vec{d}_j}{\vec{r}}|$ up to error $\epsilon_i$.
Exact recovery is guaranteed provided that, despite this slack, the optimal atoms remain distinguishable from the rest.
The following theorem makes this requirement precise by introducing a parameter $\eta \in (0,1)$ that quantifies the tolerated estimation error relative to the signal.

\begin{theorem}[Exact Recovery for QOMP]
\label{thm: qomp  exact recovery qomp}
    Let $\eta \in (0,1)$.
    Let the error on the inner product estimation be $\epsilon_i \leq \eta\min_{k \in [K]}(\|A^\dagger_{\mathrm{opt}} \vec{r}\|_\infty)/2$.
    A sufficient condition for QOMP to recover the sparsest representation of the input signal is that
    \begin{align}
        \max_{\vec{\psi}} \|A_{\mathrm{opt}}^+\vec{\psi}\|_1 < 1 - \eta
    \end{align}
    where $\vec{\psi}$ ranges over the columns of $B_{\mathrm{opt}}$.
\end{theorem}
\begin{proof}
    In the original proof of Theorem~\ref{thm: qomp  exact recovery omp}, \citet[Theorem 3.1]{tropp2004greed} makes sure that OMP selects an atom from the optimal set at each iteration by imposing that $\rho(\vec{r}) \defeq \frac{\|B^\dagger_{\mathrm{opt}}\vec{r}\|_\infty}{\|A^\dagger_{\mathrm{opt}}\vec{r}\|_\infty} < 1$, meaning that the inner products with the optimal atoms is always greater than the suboptimal ones.
    Then, the crucial step is to show that
    \begin{align}
    \label{eq: qomp  rho ineq}
        \rho(\vec{r}) \leq \max_{\vec{\psi}} \|A^+_{\mathrm{opt}} \vec{\psi}\|_1
    \end{align}
    where $\vec{\psi}$ ranges over the columns of $B_{\mathrm{opt}}$. 
    Their proof of Theorem~\ref{thm: qomp  exact recovery omp} follows directly from this equation.
    
    We approach our proof similarly, and make use of the equation above.
    To recover the optimal subset of atoms, QOMP needs to succeed at each iteration. 
    Assume that the first $k-1$ iterations succeeded. 
    At the $k^{\mathrm{th}}$ iteration, QOMP selects the atom $j^* = \argmax_{j \in \overline{\Lambda}} |\bracket{\vec{d}_j}{\vec{r}}| - 2\epsilon_i$, where $\epsilon_i$ is the error on the inner products $|\langle\vec{d}_j, \vec{r}\rangle - z_j| \leq \epsilon_i$, as by the approximate oracle $O_i$ (\ref{eq: Oi}) and Corollary~\ref{coro: finding approximate maximum}.
    Thus, requiring that QOMP selects an atom from $\Lambda_{\mathrm{opt}}$ is equivalent to asking for $\|A^\dagger_{\mathrm{opt}}\vec{r}\|_\infty - 2\epsilon_i> \|B^\dagger_{\mathrm{opt}}\vec{r}\|_\infty$.
    This leads to the inequality $\frac{\|B^\dagger_{\mathrm{opt}}\vec{r}\|_\infty}{\|A^\dagger_{\mathrm{opt}}\vec{r}\|_\infty} < 1 - \frac{2\epsilon_i}{\|A^\dagger_{\mathrm{opt}}\vec{r}\|_\infty}$.
    Defining $\frac{2\epsilon_i}{\|A^\dagger_{\mathrm{opt}}\vec{r}\|_\infty} = \eta$ (hence, asking $\epsilon_i \leq \eta \frac{\|A^\dagger_{\mathrm{opt}} \vec{r}\|_\infty}{2}$) and using Eq.~(\ref{eq: qomp  rho ineq}), we derive the sufficient condition $\max_{\vec{\psi}} \|A^+_{\mathrm{opt}} \vec{\psi}\|_1 < 1 - \eta$.
    To select the best atom in all the iterations, we need $\epsilon_i \leq \eta \min_{k \in [K]}(\|A^\dagger_{\mathrm{opt}} \vec{r}\|_\infty)/2$.
\end{proof}

This theorem should be read as a genuine strengthening of the classical analysis: Tropp’s ERC~\cite{tropp2004greed} ensures success when $\max_{\vec{\psi}}|A_{\mathrm{opt}}^+\vec{\psi}|_1 < 1$, while in QOMP the bound becomes $< 1-\eta$.
The parameter $\eta$ directly quantifies robustness: smaller $\epsilon_i$ (more accurate inner product oracles) allow recovery under weaker conditions, while larger $\epsilon_i$ require stronger incoherence.
Specializing to mutual incoherence yields the following corollary, which extends the classical incoherence condition to the quantum domain.

\begin{corollary}[Incoherence condition for QOMP]
\label{coro: qomp incoherence condition omp}
    Let $\eta \in (0,1)$.
    Let the error on the inner product estimation be $\epsilon_i \leq \eta\min_{k \in [K]}(\|A^\dagger_{\mathrm{opt}} \vec{r}^{(k)}\|_\infty)/2$.
    Then, QOMP selects an atom from $\Lambda_\mathrm{opt}$ at each iteration for any superposition of $K$ atoms from $D$ if
    \begin{align}
        K < \frac{(1-\eta)}{(2-\eta)}(\mu^{-1} + 1).
    \end{align}
\end{corollary}
\begin{proof}
    \citet[Proof of Theorem 3.5]{tropp2004greed} shows $\max_{\vec{\psi}} \|A^+_{\mathrm{opt}}\vec{\psi}\|_1 \leq \frac{K\mu}{\ 1 - (K-1)\mu}$.
    We leverage this equation to prove our result.   
    Theorem~\ref{thm: qomp  exact recovery qomp} states that QOMP performs exact recovery in $K$ steps if $\max_{\vec{\psi}} \|A^+_{\mathrm{opt}}\vec{\psi}\|_1 < 1 - \eta$ and $\epsilon_i \leq \eta\min_{k \in [K]}(\|A^\dagger_{\mathrm{opt}} \vec{r}^{(k)}\|_\infty)/2$.
    Hence, imposing $\frac{K\mu}{\ 1 - (K-1)\mu} < 1 - \eta$ and solving for $K$, we obtain $K < \frac{(1-\eta)}{(2-\eta)}(\mu^{-1} + 1)$.
\end{proof}

Compared with the classical incoherence bound $K < \tfrac{1}{2}(\mu^{-1}+1)$, the quantum condition includes the multiplicative factor $\tfrac{1-\eta}{2-\eta}$, which smoothly interpolates between the classical threshold (as $\eta \to 0$) and stricter requirements under finite oracle error.
This reflects the fact that QOMP must guard against approximate comparisons while still following the greedy atom-selection rule.

Overall, these results show that QOMP inherits the same structural recovery guarantees as OMP, up to an explicit slack that reflects the accuracy of the quantum estimation procedures.
In this sense, the classical theory of exact recovery carries over essentially unchanged, provided the precision of the oracles is calibrated appropriately.
This observation clarifies that the introduction of quantum subroutines, while substantially reducing the iteration cost, does not compromise the conditions under which greedy sparse recovery succeeds.

The guarantees proved above allow us to go beyond algorithmic analysis and apply QOMP to the concrete task quantum sparse recovery and tomography with respect to arbitrary dictionaries.

%%%%%%%%%%%%%%%%%%%%%%%%%%%%%%%%%%%%%%%%%%%%%%%%%%%%%%%%%%%%%%%%%%%%%%%%%%%%%%%%%%%%%%%%%%%%%
%%%%%%%%%%%%%%%%%%%%%%%%%%%%%%%%%%%%%%%%%%%%%%%%%%%%%%%%%%%%%%%%%%%%%%%%%%%%%%%%%%%%%%%%%%%%%
%%%%%%%%%%%%%%%%%%%%%%%%%%%%%%%%%%%%%%%%%%%%%%%%%%%%%%%%%%%%%%%%%%%%%%%%%%%%%%%%%%%%%%%%%%%%%
%%%%%%%%%%%%%%%%%%%%%%%%%%%%%%%%%%%%%%%%%%%%%%%%%%%%%%%%%%%%%%%%%%%%%%%%%%%%%%%%%%%%%%%%%%%%%
%%%%%%%%%%%%%%%%%%%%%%%%%%%%%%%%%%%%%%%%%%%%%%%%%%%%%%%%%%%%%%%%%%%%%%%%%%%%%%%%%%%%%%%%%%%%%
%%%%%%%%%%%%%%%%%%%%%%%%%%%%%%%%%%%%%%%%%%%%%%%%%%%%%%%%%%%%%%%%%%%%%%%%%%%%%%%%%%%%%%%%%%%%%
%%%%%%%%%%%%%%%%%%%%%%%%%%%%%%%%%%%%%%%%%%%%%%%%%%%%%%%%%%%%%%%%%%%%%%%%%%%%%%%%%%%%%%%%%%%%%
%%%%%%%%%%%%%%%%%%%%%%%%%%%%%%%%%%%%%%%%%%%%%%%%%%%%%%%%%%%%%%%%%%%%%%%%%%%%%%%%%%%%%%%%%%%%%

\section{Learning sparse quantum states}
\label{sec: learning sparse quantum states}
We now leverage QOMP to address the problem of \emph{exact quantum sparse recovery} and \emph{sparse quantum tomography}. 
In this task,  one is given quantum access to a target pure state $\ket{\vec{s}}$ and to a dictionary $D=\{\vec{d}_1,\ldots,\vec{d}_m\}$, with the promise that $\ket{\vec{s}}$ admits an exact $K$-sparse representation in $D$. 
The goal is to recover, up to error $\epsilon$, a concise classical description of $\ket{\vec{s}}$ in terms of a small subset of dictionary vectors.  
This is the natural analogue of compressed sensing in quantum information, and it provides a concrete setting in which the structural guarantees of QOMP translate into provable improvements for tomography.

The learning problem separates naturally into two stages.  
First, one must identify the \emph{support}; i.e., the subset $\Lambda_{\mathrm{opt}}$ of at most $K$ atoms whose span contains $\ket{\vec{s}}$, or a subset $\Lambda \subseteq \Lambda_{\mathrm{opt}}$ containing an $\epsilon$-approximation of $\ket{\vec{s}}$.  
Second, once $\Lambda$ has been recovered, one must estimate the coefficients of the expansion of $\ket{\vec{s}}$ in that subdictionary.  
We address each of these stages in turn.

\subsection{Recovering the support}
Support recovery is the combinatorial core of sparse tomography. 
Classically, algorithms such as OMP succeed under incoherence assumptions guaranteeing that an atom from the optimal support is identified at every iteration.  
Our analysis in the previous section shows that QOMP inherits these guarantees in the quantum setting, provided the inner product oracle is accurate to within a slack factor $\eta$.  
The challenge is to convert these structural guarantees into query-complexity bounds when $\ket{\vec{s}}$ and $D$ are accessible only via state-preparation unitaries.

The following theorem establishes such a guarantee: if $\ket{\vec{s}}$ admits a $K$-sparse representation in $D$ and the dictionary obeys the usual incoherence bounds, then QOMP identifies a support $\Lambda \subseteq \Lambda_{\mathrm{opt}}$ of size at most $K$ such that $\ket{\vec{s}}$ lies within $\epsilon$ of $\mathrm{span}\{\vec{d}_j : j \in \Lambda\}$, with high probability.  
The query complexity is $\widetilde{O}(\tfrac{K^{3/2}}{\gamma \eta}\tfrac{\sqrt{m}}{\epsilon})$ to the state-preparation oracles and $\widetilde{O}(\tfrac{K^2}{\gamma \eta}\tfrac{\sqrt{m}}{\epsilon})$ to the dictionary oracles, together with polynomially many additional quantum and classical resources.

\begin{theorem}[Sparse recovery with QOMP]
\label{theorem: sparse recovery with QOMP}
    Let $\epsilon, \eta \in (0,1)$. 
    Let there be quantum access to $\ket{\vec{s}} \in \C^n$ and $D \in \C^{n \times m}$ via state preparation unitaries $U_s$, $U_D$, inverses, and controlled versions. 
    Suppose that $\ket{\vec{s}}$ admits an exact $K$-sparse representation in $D$, where $K$ is a known upper bound on the sparsity.
    That is, there exists a subset $\Lambda_{\mathrm{opt}} \subseteq [m]$ with $|\Lambda_{\mathrm{opt}}|\leq K$ such that $\ket{\vec{s}} \in \mathrm{span}\{d_j : j \in \Lambda_{\mathrm{opt}}\}$.
    Let $\mu = \max_{i \neq j} |\bracket{d_i}{d_j}|$ denote the mutual incoherence of $D$.
    If 
    \begin{align}
        K < \frac{1 - \eta}{2 - \eta}\left(\frac{1}{\mu} + 1\right),
    \end{align} 
    then the QOMP algorithm, run for at most $K$ iterations or until the estimated residual norm is $\leq \epsilon/2$, with parameters $\epsilon_i \leq \eta\frac{1}{\sqrt{K}}\gamma \epsilon$, $\epsilon_f = \epsilon/2$, where $\gamma$ is a lower bound on $\sigma_{\min}(D_{\Lambda_\mathrm{opt}})$, satisfies the following:
    \begin{enumerate}
        \item It uses a total of $\widetilde{O}(\frac{K^{3/2}}{\gamma \eta}\frac{\sqrt{m}}{\epsilon})$ queries to $U_s$, $U_s^\dagger$, and their controlled versions.
        \item It uses a total of $\widetilde{O}(\frac{K^{2}}{\gamma \eta}\frac{\sqrt{m}}{\epsilon})$ queries to $U_D$, $U_D^\dagger$, and their controlled versions.
        \item It uses polynomially many other quantum and classical resources.
        \item It outputs a support $\Lambda \subseteq \Lambda_{\mathrm{opt}}$ of size at most $K$ whose span contains a vector within $\epsilon$ of $\ket{\vec{s}}$, with high probability.
    \end{enumerate}
\end{theorem}
\begin{proof}
    We first bound $\|D_{\Lambda_\mathrm{opt}}^\dagger \vec{r}\|_\infty$ and the running time, and then establish the approximation guarantee.
    \smallskip
    
    1, 2) Corollary~\ref{coro: qomp incoherence condition omp} ensures that if $K < \frac{1 - \eta}{2 - \eta}\left(\frac{1}{\mu} + 1\right)$ and $\epsilon_i \leq \eta\|D_{\Lambda_\mathrm{opt}}^\dagger \vec{r}\|_\infty/2$, then, at each iteration, QOMP selects an atom from the optimal set $\Lambda_{\mathrm{opt}}$ with high probability. 
    While $\overline{\|\vec{r}\|}>\epsilon/2$ we have $\|\vec{r}\|_2>\epsilon$, and therefore
    \begin{align}
        %https://math.stackexchange.com/questions/1500789/lower-bound-for-the-supreme-norm-of-one-matrix-vector-product?rq=1
        \|D_{\Lambda_\mathrm{opt}}^\dagger \vec{r}\|_\infty \geq \frac{1}{\sqrt{K}}\|D_{\Lambda_\mathrm{opt}}^\dagger \vec{r}\|_2 \geq \frac{1}{\sqrt{K}}\sigma_{\min}(D_{\Lambda_\mathrm{opt}})\|\vec{r}\|_2 > \frac{1}{\sqrt{K}}\sigma_{\min}(D_{\Lambda_\mathrm{opt}})\epsilon.
    \end{align}
    Hence, $\epsilon_i \leq \frac{1}{\sqrt{K}}\gamma \epsilon$, with $\gamma \leq \sigma_{\min}(D_{\Lambda_\mathrm{opt}})$, suffices for correct selection at each iteration, with high probability.
    Conditioning on success of all iterations, the procedure selects only elements of $\Lambda_{\mathrm{opt}}$, so the output support $\Lambda$ satisfies $\Lambda\subseteq\Lambda_{\mathrm{opt}}$ and $|\Lambda|\le K$.
    By the \emph{Union bound} and the \emph{Discrete amplification lemma} (Sec.~\ref{section: amplification of success probabilities}), this holds with high probability with only $\Ot(1)$ overhead.

    Using \emph{QOMP iteration cost} (Theorem~\ref{theorem: QOMP iteration cost}) and that the algorithm runs for at most $K$ iterations, the expected number of queries to $U_s$, $U_s^\dagger$, $U_D$, $U_D^\dagger$, and their controlled versions are $\widetilde{O}\left(\frac{K^{1.5}}{\gamma\eta}\frac{\sqrt{m}}{\epsilon}\right)$ and $\widetilde{O}\left(\frac{K^{2}}{\gamma\eta}\frac{\sqrt{m}}{\epsilon}\right)$.
    Since this expectation is expressed in terms of known parameters $(\gamma,\epsilon,\eta,K)$, \emph{Markov’s inequality} yields a worst-case bound with the same scaling (see Theorem~\ref{theorem: markovs inequality} and Sec.~\ref{section: amplification of success probabilities}).
    \smallskip

    3) Accessing and updating $\Lambda$ and its complement can be implemented in $O(\mathrm{poly}(m))$ time without QRAM (and in fact $O(\mathrm{poly}(K,\log m))$ suffices, though we do not rely on this refinement).
    The remaining classical routines and the $1$- and $2$-qubit gates used in QOMP’s subroutines are polynomial in the problem parameters.
    \smallskip
    
    4) The exit rule guarantees the stated approximation. 
    The algorithm halts only when the estimated residual obeys $\overline{\|\vec{r}\|}\le \epsilon/2$; since the estimator has additive error at most $\epsilon/2$, this implies $\|\vec{r}\|\le\epsilon$ at termination.
    Because each successful iteration adds an atom from $\Lambda_{\mathrm{opt}}$, the final support $\Lambda\subseteq\Lambda_{\mathrm{opt}}$ has $|\Lambda|\leq K$, and there exists $\ket{\tilde s} = D_\Lambda D_\Lambda^+ \ket{\vec{s}} \in\mathrm{span}\{d_j:j\in\Lambda\}$ with $\|\ket{\tilde s}-\ket{\vec{s}}\|_2\le\epsilon$.
    This occurs with high probability by the amplification argument above.
\end{proof}

Here, a central point is that all approximation errors in the iteration analysis can be rewritten in terms of controlled quantities, such as $\epsilon$, $K$, and $\sigma_{\min}(D_\Lambda)$.  
Because the running time is expressed in terms of these parameters, the expected query complexity can be lifted to a worst-case bound via \emph{Markov’s inequality}.  
Moreover, one can always conservatively replace the instance-dependent $\sigma_{\min}(D_\Lambda)$ by the global bound $\sigma_{\min}(D)$, which is fixed and iteration-independent.

This result should be contrasted with the $\Theta(N/\epsilon)$ lower and upper bounds for general tomography of $N$-dimensional pure states~\cite{van2023quantum}.  
Without structural assumptions, $\Omega(N/\epsilon)$ queries to the state-preparation unitary are unavoidable to approximate an arbitrary dense state up to $\ell_2$-error $\epsilon$.  
By exploiting sparsity in incoherent dictionaries, QOMP reduces this scaling to $\widetilde{O}(\sqrt{N}/\epsilon)$ queries when $m=O(N)$ and $K=\widetilde{O}(1)$ with well-conditioned support ($\sigma_{\min}\in\widetilde{\Omega}(\mathrm{polylog}(N)^{-1})$).  

Finally, while Theorem~\ref{theorem: sparse recovery with QOMP} guarantees recovery of a subset $\Lambda \subseteq \Lambda_{\mathrm{opt}}$, it is natural to ask whether the full optimal support can also be recovered.  
This is possible under a mild \emph{identifiability assumption}, namely that no smaller support yields an $\epsilon$-approximation to $\ket{\vec{s}}$.
In that case, the algorithm cannot terminate early, and the exact support is recovered.

\begin{corollary}[Exact sparse recovery with QOMP] 
\label{corollary: exact sparse recovery with QOMP}
    Suppose the assumptions of Theorem~\ref{theorem: sparse recovery with QOMP} hold. If, in addition, every vector $\vec{y}$ supported on fewer than $|\Lambda_{\mathrm{opt}}|$ columns satisfies $\|\ket{\vec{s}} - D\vec{y}\|_2 > \epsilon$, then the procedure from Theorem~\ref{theorem: sparse recovery with QOMP} recovers the full optimal support $\Lambda_{\mathrm{opt}}$ with high probability, solving problem $\mathcal{QP}_0$ in polynomial time.
\end{corollary}
\begin{proof}
    Theorem~\ref{theorem: sparse recovery with QOMP} ensures that the algorithm outputs $\Lambda \subseteq \Lambda_{\mathrm{opt}}$ with $|\Lambda|\le K$ such that $\ket{\vec{s}}$ is $\epsilon$-approximated from $\mathrm{span}\{d_j:j\in\Lambda\}$, with high probability. 
    The additional assumption rules out any $\epsilon$-approximation with support smaller than $|\Lambda_{\mathrm{opt}}|$, so the algorithm cannot stop early. 
    Hence $\Lambda=\Lambda_{\mathrm{opt}}$, with high probability.
\end{proof}

Thus, under standard incoherence assumptions and a natural identifiability condition, QOMP recovers the full optimal support $\Lambda_{\mathrm{opt}}$ with high probability, \emph{solving $\mathcal{QP}_0$ in polynomial time}.

\subsection{Recovering the coefficients}
Once the support has been identified, the remaining task is to recover the coefficients of $\ket{\vec{s}}$ in the subdictionary $D_\Lambda$.  
At this stage the problem reduces to solving a sparse quantum linear system: we seek $\vec{x}$ supported on $\Lambda$ such that $D_\Lambda \vec{x} \approx \ket{\vec{s}}$.  
This formulation highlights the role of QOMP as a reduction: it converts the combinatorial search over $\binom{m}{K}$ supports into a well-posed estimation problem of dimension $K$.  

From the perspective of tomography, this reduction is significant.  
In the absence of further structure, learning an arbitrary dense $N$-dimensional pure state requires $\Theta(N/\epsilon)$ queries even with access to the state-preparation unitary~\cite{van2023quantum}.  
By contrast, under sparsity assumptions in an incoherent dictionary, once the support has been identified, tomography requires only estimating $K$ coefficients, where $K \ll N$.  
The query complexity is therefore governed by $K$ and the conditioning parameter $\sigma_{\min}(D_\Lambda)$, rather than the ambient dimension $N$.
While support recovery remains the dominant cost in the overall procedure, this reduction is what makes sparse tomography feasible: one pays a $\widetilde{O}(\frac{K^{3/2}}{\gamma \eta}\frac{\sqrt{m}}{\epsilon})$ overhead to identify the correct subspace, but subsequent coefficient recovery adds only polynomial dependence in $K$ and $\gamma$.
Under suitable conditions, we remember that the overhead can drop to $\widetilde{O}(\sqrt{N}/\epsilon)$, enabling substantial query savings in large Hilbert spaces.

The next lemma shows that one can efficiently prepare a normalized quantum state proportional to the optimal coefficient vector $D_\Lambda^+\ket{\vec{s}}$.

\begin{lemma}(Coefficients state preparation)
\label{lemma: intermediate tomography}
        Assume the hypotheses of Theorem~\ref{theorem: sparse recovery with QOMP} and let $\Lambda$ be the output support upon success.
        There exists an algorithm that prepares a state $\ket{\vec{x}}$ such that $\|\ket{\vec{x}} - \frac{D_{\lambda}^+\ket{\vec{s}}}{\|D_{\lambda}^+\ket{\vec{s}}\|}\| \leq \epsilon$ using $\Ot(\frac{\sqrt{K}}{\gamma}\mathrm{polylog}(1/\epsilon))$ queries to $U_s$, $U_D$, their inverses and controlled versions, and polynomially many other $1$- and $2$-qubit gates.
\end{lemma}
\begin{proof}
    Using $U_D$ and polynomially many gates to create $U_\Lambda$, we can create a $(\sqrt{K}, \ceil{\log(n+K)},\epsilon_0)$ block-encoding of $D_\Lambda$ in time $\Ot(T_D + T_\Lambda)$ (Theorem~\ref{thm: qomp  block encoding from qa}). $D_\Lambda$ has singular values in $[\gamma, \sqrt{K}]$.
    Choosing $\epsilon_{0}$ to satisfy Theorem~\ref{theorem: QSVT linear systems} (absorbed in polylog factors) we run the \emph{Quantum linear systems via QSVT} routine and conclude the proof. 
\end{proof}

Building on this, one can obtain a sparse classical description of the coefficients, with guarantees both on approximation quality and on query complexity.

\begin{corollary}(Sparse coefficients tomography)
\label{corollary: sparse coefficient tomography}
    Suppose the assumptions of Theorem~\ref{theorem: sparse recovery with QOMP} hold, and let $\Lambda$ be the support returned by QOMP upon success when run to residual $\epsilon/4$.
    There exists an algorithm that, with probability $\geq 1-\delta$, outputs a $O(K\log(K)\log(1/\delta))$-sparse classical vector $\vec{y} \in \C^m$ such that $\|\ket{\vec{s}} - \frac{D_\Lambda \vec{y}}{\|D_\Lambda \vec{y}\|}\| \leq \epsilon$ using 
    \begin{align}
        \Ot\left(\kappa(D_\Lambda)\frac{K^{3/2}}{\gamma}\frac{1}{\epsilon}\mathrm{polylog}(1/\delta)\right) \quad \left(\text{i.e., } \Ot\left(\frac{K^2}{\gamma^2}\frac{1}{\epsilon}\mathrm{polylog}(1/\delta)\right) \right)
    \end{align} queries to $U_s$, $U_D$, their inverses and controlled versions, and polynomially many other $1$- and $2$-qubit gates. Here $\kappa(D_\Lambda)$ upper bounds $\frac{\|D_\Lambda\|}{\sigma_{\min}(D_\Lambda)}$.
\end{corollary}
\begin{proof}
    By success of QOMP with residual parameter $\epsilon_0$, there exists $\vec{x} \in \C^{m}$ with support $|\Lambda| \leq K$ such that 
    \begin{align}
        \|\ket{\vec{s}} - D_\Lambda \vec{x}\| \leq \epsilon_0 \quad\quad\quad \text{and we may take }\vec{x} = D_\Lambda^+ \ket{\vec{s}}.
    \end{align}
    Set the unit vector $\ket{\vec{x}} \eqdef \vec{x}/\|\vec{x}\|.$

    \emph{Proxy coefficients preparation. } 
    By Lemma~\ref{lemma: intermediate tomography}, for any $\epsilon_1 > 0$ we can produce a state $\overline{\ket{\vec{x}}}$ such that $\|\overline{\ket{\vec{x}}} - \ket{\vec{x}}\| \leq \epsilon_1$ using $\Ot(\frac{\sqrt{K}}{\gamma}\mathrm{polylog}(1/\epsilon_1))$ queries to $U_s$, $U_D$, their inverses and controlled versions, and polynomially many other $1$- and $2$-qubit gates.
    Choose $\epsilon_1 < \min(\epsilon_t\sqrt{k/N}, \epsilon_t/2)$, where $\epsilon_t > 0$ will be set later. 
    The second inequality ensures that no off-support entry of $\overline{\ket{\vec{x}}}$ can exceed the threshold $\epsilon_t\sqrt{K/N}$, hence at most $K$ entries of $\ket{\overline{\vec{x}}}$ are $\geq \epsilon_t\sqrt{K/N}$.

    \emph{Sparse coefficients tomography.} 
    Apply \emph{Orthogonal sparse tomography} from Theorem~\ref{theorem: sparse tomography - orthogonal} to $|\overline{\vec{x}}\rangle$.
    With probability greater than $1-\delta$, this returns a $O(K\log(K)\log(1/\delta))$-sparse classical vector $\vec{y} \in \C^{m}$ with $\|\vec{y} - |\overline{\vec{x}}\rangle\| \leq \epsilon_t$ using a total of $\Ot(\frac{K^{3/2}}{\gamma\epsilon_t}\polylog(1/\delta))$ queries to $U_s$, $U_D$, their inverses and controlled versions, and polynomially many other $1$- and $2$-qubit gates.
    By the triangle inequality and our choice of $\epsilon_1$, $\|\vec{y} - |\vec{x}\rangle\| \leq \epsilon_t+\epsilon_1 \leq \frac{3}{2}\epsilon_t$.

    \emph{Error propagation.} 
    Then, $\left\| \ket{\vec{s}} - \frac{D_\Lambda \vec{y}}{\|D_\Lambda\vec{y}\|}\right\| \leq \underbrace{\left\|\ket{\vec{s}} - \frac{D_\Lambda \ket{\vec{x}}}{\|D_\Lambda \ket{\vec{x}}\|}\right\|}_{\text{(I)}} + \underbrace{\left\| \frac{D_\Lambda \ket{\vec{x}}}{\|D_\Lambda \ket{\vec{x}}\|} - \frac{D_\Lambda \vec{y}}{\|D_\Lambda\vec{y}\|}\right\|}_{\text{(II)}}$.
    
    For (I), use the triangle inequality and the colinear difference (two vectors on the same line differ by the difference of their norms): $\left\|\ket{\vec{s}} - \frac{D_\Lambda \ket{\vec{x}}}{\|D_\Lambda \ket{\vec{x}}\|}\right\| \leq \|\ket{\vec{s}} - D_\Lambda\vec{x}\| + \|\|D_\Lambda \vec{x}\| - 1\| \leq \epsilon_0 + \underbrace{\|\|D_\Lambda \vec{x}\| - \|\ket{\vec{s}}\|\|}_{\text{use reverse triangular}} \leq 2\epsilon_0.$

    For (II), we use $\|D_\Lambda \ket{\vec{x}}\| \geq \sigma_{\min}(D_\Lambda)$ and obtain
    $\left\| \frac{D_\Lambda \ket{\vec{x}}}{\|D_\Lambda \ket{\vec{x}}\|} - \frac{D_\Lambda \vec{y}}{\|D_\Lambda\vec{y}\|}\right\| \leq \left\| \frac{D_\Lambda (\ket{\vec{x}} - \vec{y}) }{\|D_\Lambda \ket{\vec{x}}\|}\right\| + \left\| D_\Lambda \vec{y}\left(\frac{1}{\|D_\Lambda \ket{\vec{x}}\|} - \frac{1}{\|D_\Lambda\vec{y}\|}\right)\right\| \leq 2\left\| \frac{D_\Lambda (\ket{\vec{x}} - \vec{y}) }{\|D_\Lambda \ket{\vec{x}}\|}\right\| \leq 3\frac{\|D_\Lambda\|}{\sigma_{\min}(D_\Lambda)}\epsilon_t$.

    Combining, we have $\| \ket{\vec{s}} - \frac{D_\Lambda \vec{y}}{\|D_\Lambda\vec{y}\|}\| \leq 2 \epsilon_0 + 3\kappa(D_\Lambda)\epsilon_t.$
    Choosing parameters $\epsilon_0 \leq \epsilon/4$, $\epsilon_T \leq \epsilon/(6\kappa(D_\Lambda))$, $\epsilon_1 < \min(\epsilon_t\sqrt{K/N}, \epsilon_t/2)$, we bound $\|\ket{\vec{s}} - \frac{D_\Lambda \vec{y}}{\|D_\Lambda\vec{y}\|}\| \leq \epsilon.$
    Substituting these in the time complexity concludes the proof.

    As a final remark, since the columns of $D_\Lambda$ are unit-norm, then $\frac{\|D_\Lambda\|}{\sigma_{\min}(D_\Lambda)} \leq \frac{\sqrt{K}}{\gamma}$, where $\gamma$ lower bounds $\sigma_{\min}(D_\Lambda)$.
\end{proof}

The complexity of coefficient recovery scales as $\Ot \left(\frac{K^2}{\gamma^2\epsilon}\right)$ queries, where $\gamma$ lower bounds $\sigma_{\min}(D_\Lambda)$.  
Since $K=\widetilde{O}(1)$ in the sparse regime of primary interest, this overhead is negligible compared to support recovery.  

Furthermore, in scenarios where QRAM access is available, the overall complexity can be further reduced to
\begin{align}
    \widetilde{O}\left(\frac{K}{\epsilon} \kappa(D_\Lambda) \mu(D_{\Lambda})\right),
\end{align}
where $\mu(D_\Lambda)$ is the normalization parameter (Def.~\ref{def:mu}).
We view this as a refinement under stronger architectural assumptions, rather than a prerequisite for our main guarantees.  

In summary, these results establish the first framework systematic for efficient \emph{sparse quantum tomography in non-orthogonal, overcomplete dictionaries}.  
They complement prior work on low-rank compressed sensing for quantum states~\cite{gross2010quantum}, demonstrating that sparsity in incoherent dictionaries also enables provable polynomial improvements over dense tomography.
Conceptually, QOMP shows that structural promises beyond rank can be leveraged for pure states in a fully quantum setting, and that approximate quantum subroutines can be orchestrated to yield rigorous end-to-end recovery guarantees.  

Beyond their theoretical significance, these guarantees suggest several directions for practical use.  
First, the recovered coefficients enable \emph{approximate state preparation}: an $\epsilon$-close copy of the target state can be reconstructed from only a handful of dictionary vectors, potentially yielding simpler unitaries than those that originally generated the state.  
Second, parties who agree on a dictionary could in principle communicate only the sparse coefficient vector rather than the full state, reminiscent of how the JPEG compression format uses the discrete cosine transform to transmit compressed images.  
This analogy highlights the possibility of compact, structured, and interpretable representations of quantum states tailored to specific tasks.  
Finally, sparse coefficient vectors also serve as low-dimensional features for downstream quantum or classical learning tasks.  
These applications remain speculative, but they illustrate how sparse tomography may serve not only as a tool for efficient reconstruction, but also as a bridge between quantum algorithms, information theory, and the modeling of physical systems.  

After proving the main results of this work, we now turn to a meta-task: estimating the incoherence parameter itself, which underlies the guarantees.

%%%%%%%%%%%%%%%%%%%%%%%%%%%%%%%%%%%%%%%%%%%%%%%%%%%%%%%%%%%%%%%%%%%%%%%%%%%%%%%%%%%%%%%%%%%%%
%%%%%%%%%%%%%%%%%%%%%%%%%%%%%%%%%%%%%%%%%%%%%%%%%%%%%%%%%%%%%%%%%%%%%%%%%%%%%%%%%%%%%%%%%%%%%
%%%%%%%%%%%%%%%%%%%%%%%%%%%%%%%%%%%%%%%%%%%%%%%%%%%%%%%%%%%%%%%%%%%%%%%%%%%%%%%%%%%%%%%%%%%%%
%%%%%%%%%%%%%%%%%%%%%%%%%%%%%%%%%%%%%%%%%%%%%%%%%%%%%%%%%%%%%%%%%%%%%%%%%%%%%%%%%%%%%%%%%%%%%
%%%%%%%%%%%%%%%%%%%%%%%%%%%%%%%%%%%%%%%%%%%%%%%%%%%%%%%%%%%%%%%%%%%%%%%%%%%%%%%%%%%%%%%%%%%%%
%%%%%%%%%%%%%%%%%%%%%%%%%%%%%%%%%%%%%%%%%%%%%%%%%%%%%%%%%%%%%%%%%%%%%%%%%%%%%%%%%%%%%%%%%%%%%
%%%%%%%%%%%%%%%%%%%%%%%%%%%%%%%%%%%%%%%%%%%%%%%%%%%%%%%%%%%%%%%%%%%%%%%%%%%%%%%%%%%%%%%%%%%%%
%%%%%%%%%%%%%%%%%%%%%%%%%%%%%%%%%%%%%%%%%%%%%%%%%%%%%%%%%%%%%%%%%%%%%%%%%%%%%%%%%%%%%%%%%%%%%

\section{Quantum estimation of the mutual incoherence}
The guarantees for QOMP and sparse tomography rely on structural properties of the dictionary, most notably the mutual incoherence parameter $\mu$.
In practice, one may not know $\mu$ a priori, especially when dealing with large or data-driven dictionaries, and being able to estimate it efficiently is therefore a useful primitive.
This motivates a final problem: \emph{given quantum access to the dictionary $D$, can we estimate its mutual incoherence faster than classically?}

Recall that we can assume the columns of $D$ have unit $\ell_2$ norm, without loss of generality, and define $\mu = \max_{i, \in [m], i\neq j}  |\langle \vec{d}_i | \vec{d}_j \rangle |$.

\begin{theorem}[Estimating the mutual incoherence]
\label{theorem: Estimating the mutual incoherence}
    Let there be quantum access to a dictionary $D\in \C^{n \times m}$ with $\ell_2$ unit norm columns in time $T_D$.
    There exists a quantum algorithm that estimates the mutual incoherence $\mu = \max_{i, \in [m], i\neq j} | \langle \vec{d}_i | \vec{d}_j \rangle |$  of $D$ to absolute error $\epsilon$ with high probability in $\Ot(T_D(m/\epsilon))$ time.
\end{theorem}
\begin{proof}
    We can use quantum access to $D$ to perform inner products in superposition, creating an oracle that performs
    \begin{align}
        O_{ij}: \ket{i}\ket{j}\ket{0} &\rightarrow \ket{i}\ket{j}| \overline{|\langle \vec{d}_i | \vec{d}_j \rangle | }\rangle
    \end{align}
    where $\abs{\overline{|\langle \vec{d}_i | \vec{d}_j \rangle | } - \langle \vec{d}_i | \vec{d}_j \rangle} \leq \epsilon$, in time $O(1/\epsilon)$ (using the same considerations and resources of Sec.~\ref{section: quantum orthogonal matching pursuit} for the absolute value).
    Then, we can create access to a state $|\vec{\phi}\rangle = \frac{1}{\sqrt{m(m-1)}} \sum_{i=0}^{m-1}\sum_{j=0, j \neq i}^{m-1} \ket{i}\ket{j}\ket{0}$ in $\Ot(\polylog(m))$ time.
    Finally, we use \emph{Finding the maximum with an approximate unitary} from Corollary~\ref{coro: finding approximate maximum} with the oracle $O_{ij}$ that approximates the inner products to extract the index and the value of the mutual incoherence, with a total cost of $\Ot(T_D(m/\epsilon))$.
\end{proof}

From a classical perspective, estimating $\mu$ is straightforward but expensive: one must compute $O(m^2)$ inner products of $n$-dimensional vectors, for a total cost of $O(m^2 n)$.
More sophisticated classical algorithms based on $\ell_1$-sampling (see, e.g., Lemma 3 in \cite{rebentrost2021quantumising}) could reduce this to $O(m^2/\epsilon^2)$ for $\epsilon$-accurate estimation.
In contrast, the proposed quantum routine achieves the same accuracy in $\widetilde{O}(m/\epsilon)$ dictionary queries in the Oracular-Circuit model (or total time in the QRAM model), providing a quadratic improvement in the dictionary size.

%%%%%%%%%%%%%%%%%%%%%%%%%%%%%%%%%%%%%%%%%%%%%%%%%%%%%%%%%%%%%%%%%%%%%%%%%%%%%%%%%%%%%%%%%%%%%
%%%%%%%%%%%%%%%%%%%%%%%%%%%%%%%%%%%%%%%%%%%%%%%%%%%%%%%%%%%%%%%%%%%%%%%%%%%%%%%%%%%%%%%%%%%%%
%%%%%%%%%%%%%%%%%%%%%%%%%%%%%%%%%%%%%%%%%%%%%%%%%%%%%%%%%%%%%%%%%%%%%%%%%%%%%%%%%%%%%%%%%%%%%
%%%%%%%%%%%%%%%%%%%%%%%%%%%%%%%%%%%%%%%%%%%%%%%%%%%%%%%%%%%%%%%%%%%%%%%%%%%%%%%%%%%%%%%%%%%%%
%%%%%%%%%%%%%%%%%%%%%%%%%%%%%%%%%%%%%%%%%%%%%%%%%%%%%%%%%%%%%%%%%%%%%%%%%%%%%%%%%%%%%%%%%%%%%
%%%%%%%%%%%%%%%%%%%%%%%%%%%%%%%%%%%%%%%%%%%%%%%%%%%%%%%%%%%%%%%%%%%%%%%%%%%%%%%%%%%%%%%%%%%%%
%%%%%%%%%%%%%%%%%%%%%%%%%%%%%%%%%%%%%%%%%%%%%%%%%%%%%%%%%%%%%%%%%%%%%%%%%%%%%%%%%%%%%%%%%%%%%
%%%%%%%%%%%%%%%%%%%%%%%%%%%%%%%%%%%%%%%%%%%%%%%%%%%%%%%%%%%%%%%%%%%%%%%%%%%%%%%%%%%%%%%%%%%%%

\section{Conclusion}
\label{sec:conclusions}

In this work we introduced and studied \emph{quantum sparse recovery}, the problem of reconstructing a quantum state that admits a sparse representation in an \emph{overcomplete, non-orthogonal dictionary}.  
Our results delineate both the limitations and the opportunities of this problem.  
On the negative side, we showed that quantum sparse recovery is NP-hard in full generality, even with access to state-preparation and dictionary oracles and inverses.  
On the positive side, we designed Quantum Orthogonal Matching Pursuit (QOMP), the first greedy quantum sparse recovery algorithm that operates directly on quantum states, faithfully mirroring the classical OMP while remaining stable under iteration.  
QOMP achieves provable recovery guarantees under standard incoherence assumptions and yields the first framework for \emph{sparse tomography in non-orthogonal dictionaries}, reducing the query complexity below the tight bounds known for general pure-state tomography.  
In particular, in sparse regimes with $K=\widetilde{O}(1)$, $m=O(N)$, and well conditioned support (i.e., $\sigma_{\min} (D_\Lambda) \geq \gamma \in \Omega(\mathrm{polylog}(N)$), QOMP achieves query complexity $\widetilde{O}(\sqrt{N}/\epsilon)$, improving polynomially over the tight $\Theta(N/\epsilon)$ bound for general pure-state tomography~\cite{van2023quantum}.  

\smallskip
Beyond these core contributions, we also analyzed QOMP in the QRAM model, where it offers per-iteration polynomial speedups, and developed a quantum procedure to estimate the mutual incoherence of a dictionary, a key parameter in sparse recovery.  
Together, these results identify the boundary between hardness and tractability, and open the door to structured regimes where sparsity can be harnessed for quantum speedups.  

\smallskip
Our findings raise several directions for future work.  
First, while QOMP inherits the guarantees of OMP under incoherence, it remains an open question whether alternative quantum algorithms (possibly inspired by convex relaxations such as $\ell_1$ minimization) can achieve stronger guarantees in different regimes or further improve query and time efficiency.  
Second, the application of sparse tomography to physically motivated dictionaries deserves further exploration: can incoherent dictionaries derived from physical symmetries, tensor networks theory, or variational ans\"atze yield practical speedups in learning and simulation?  
Moreover, classical sparse recovery is routinely used as a building block for dictionary learning problems.
Our setting suggests an analogous quantum task: learn a dictionary of quantum states that yields sparse representations for states drawn from a given process, algorithmic family, or probability distribution. 
How can we learn quantum dictionaries efficiently?
Third, the role of sparsity in quantum machine learning remains largely unexplored: sparse coefficients may serve as interpretable features, much like in classical data science.  
Finally, the hardness result invites a deeper complexity-theoretic study of which dictionary or states structural promises make quantum sparse recovery efficient, and how this connects to the broader landscape of quantum learning theory.

\smallskip
We also note that QOMP does not recover the optimal query complexity known for orthogonal dictionaries, somewhat similarly to how quantum $\ell_1$-regularization methods~\cite{chen2021quantum} fail to match known lower bounds.  
It would be interesting to investigate optimal query- and time-efficient algorithms for general incoherent dictionaries.  

\smallskip
In summary, quantum sparse recovery provides a new lens on one of the most fundamental primitives in quantum information.  
It bridges ideas from compressed sensing, learning theory, and quantum algorithms, and shows that sparsity in non-orthogonal dictionaries - a useful resource in classical signal processing - can also enable genuine quantum advantages.  
We hope that this work will stimulate further research at the intersection of these fields, bringing both conceptual insights and practical tools for the efficient characterization and use of quantum states.  

\begin{acknowledgments}
A.B. and S.Z. would like to thank Professors Ferruccio Resta and Donatella Sciuto for their support.  
A.B. thanks Ignacio Cirac for his support at MPQ and for many insightful discussions.  
He is also particularly grateful to Prof. Giacomo Boracchi for his inspiring lectures on sparse representations, to Alessandro Luongo, Rolando Somma, and Ronald de Wolf for valuable discussions on the quantum preliminaries, to Marten Folkertsma for discussions on the NP-hardness proof, and to Patrick Rebentrost for hosting him at CQT during part of this project.  
A.B. would also like to thank Andrea Bonvini for his help with Figure~\ref{fig: sparse cube}.  
This work originated with the supervision of the M.Sc.\ thesis of S.V.~\cite{vanerio2021quantum}, was developed further in Part I of the Ph.D.\ thesis of A.B.~\cite{armando2024quantum}, and reached completion during the time at MPQ.  
A.B.'s research was partially funded by THEQUCO as part of the Munich Quantum Valley, supported by the Bavarian State Government through the Hightech Agenda Bayern Plus. 
Additional financial support was provided by ICSC - \quoted{National Research Centre in High Performance Computing, Big Data and Quantum Computing,} Spoke 10, funded by the European Union - NextGenerationEU, under grant \emph{PNRR-CN00000013-HPC}.
\end{acknowledgments}

\providecommand{\noopsort}[1]{}\providecommand{\singleletter}[1]{#1}%

\onecolumngrid
\appendix

%%%%%%%%%%%%%%%%%%%%%%%%%%%%%%%%%%%%%%%%%%%%%%%%%%%%%%%%%%%%%%%%%%%%%%%%%%%%%%%%%%%%%%%%%%%%%
%%%%%%%%%%%%%%%%%%%%%%%%%%%%%%%%%%%%%%%%%%%%%%%%%%%%%%%%%%%%%%%%%%%%%%%%%%%%%%%%%%%%%%%%%%%%%

\section{Weighted Euclidean distance estimation}
\label{appendix: euclidean distance estimation}

We discuss the implementation of Theorem~\ref{theorem: weighted euclidean distance estimation}.
Previous work have already described routines to estimate the squared Euclidean distance between two quantum states to which we have quantum access~\cite{kerenidis2019qmeans}.
The basic circuit is represented in Figure~\ref{fig:norm estimation circuit}, with the proof concluding through amplitude amplification and powering lemma, to boost the success probability.

\begin{figure}[h]
    \centering
    \scalebox{1}{ 
        \Qcircuit @C=1.0em @R=0.2em @!R { \\
            \nghost{{\ket{0}} :  } & \lstick{{\ket{0}} :  } & \qw & \gate{\mathrm{U_{v}}} & \gate{\mathrm{U_{c}}} & \qw & \qw & \qw \\
            \nghost{{\ket{0}} :  } & \lstick{{\ket{0}} :  } & \gate{\mathrm{H}} & \ctrl{-1} & \ctrlo{-1} & \gate{\mathrm{H}} & \qw & \qw\\
        \\ }
    }
    \caption{Circuit estimating $\norm{\ket{\vec{v}}-\ket{\vec{c}}}$.
    The absolute value amplitude of $\ket{1}$ in the auxiliary qubit (at the bottom) after the circuit is $\frac{\norm{\ket{\vec{v}}-\ket{\vec{c}}}}{2}$.}
    \label{fig:norm estimation circuit}
\end{figure}
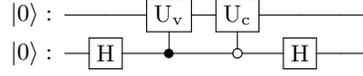

In our QOMP, we want to compute the Euclidean distance between two classical vectors $\|\vec{v} - \vec{c}\|$, represented through access to two quantum states and a classical representation of the vector norms.
To estimate $\|\|\vec{v}\|\ket{\vec{v}} - \|\vec{c}\|\ket{\vec{c}}\|$, we develop a generic routine that allows us to estimate $\|\alpha\ket{\vec{v}} - \beta\ket{\vec{c}}\|$, for some generic weights $\alpha$ and $\beta$.
The main building block of our routine is the circuit in Figure~\ref{fig: weighted norm estimation circuit}, which is a modification of the one in Figure~\ref{fig:norm estimation circuit} inspired by the state preparation routine of \citet[Section III, Method 1]{shao2018linear}.

This circuit prepares a state $\frac{1}{2}\left(\frac{1}{\beta}\ket{\vec{v}} - \frac{1}{\alpha}\ket{\vec{c}} \right)\ket{1,1} + \ket{\psi^\perp}$ where $\ket{\psi^\perp}$ is supported entirely on states whose last two qubits are orthogonal to $\ket{1,1}$.
Then, we can observe that $\mathrm{Pr}[\ket{1,1}] = \frac{1}{|2\alpha \beta|^2} \sum_{i=0}^{n-1}  \abs{\alpha \frac{\vec{v}_i}{\|\vec{v}\|} - \beta\frac{\vec{c}_i}{\|\vec{c}\|}}^2 = \frac{\|\alpha\ket{\vec{v}} - \beta\ket{\vec{c}}\|^2}{\abs{2\alpha \beta}^2}$.
Using \emph{Absolute value amplitude estimation} (Theorem~\ref{thm: amplitude estimation amplitude}) we can estimate $\sqrt{\mathrm{Pr}[\ket{1,1}]} = \frac{\|\alpha\ket{\vec{v}} - \beta\ket{\vec{c}}\|}{2|\alpha||\beta|}$ to precision $\epsilon_1 \leq \frac{\epsilon}{2|\alpha||\beta|}$ with probability greater than $\pi^2/8$ in time $O((T_v + T_c)\frac{|\alpha||\beta|}{\epsilon})$. 
Multiplying the estimate $\overline{a}$ by $2|\alpha||\beta|$, we obtain $2|\alpha||\beta|(|\overline{a} - \sqrt{\mathrm{Pr}[\ket{1,1}]}|) \leq \epsilon$, which was our goal.
Finally, with the \emph{Powering lemma} (Lemma~\ref{lemma: powering lemma (median)}) we can arbitrarily increase the success probability to $1-\delta$ with a multiplicative overhead of $O(\log(1/\delta))$.

Although beyond the interests of this work, we can adapt this algorithm to compute 
\begin{align}
    \ket{i}\ket{j}\ket{0} \rightarrow \ket{i}\ket{j}|\overline{\|\alpha_i\ket{\vec{v}_i} - \beta_j\ket{\vec{c}_j}\|}\rangle,    
\end{align}
provided that we have unitaries implementing both $\ket{i} \rightarrow \ket{\vec{v}_i}$, $\ket{j} \rightarrow \ket{\vec{c}_j}$ and $\ket{i} \rightarrow \ket{\alpha_i}$, $\ket{j} \rightarrow \ket{\beta_j}$.

%%%%%%%%%%%%%%%%%%%%%%%%%%%%%%%%%%%%%%%%%%%%%%%%%%%%%%%%%%%%%%%%%%%%%%%%%%%%%%%%%%%%%%%%%%%%%
%%%%%%%%%%%%%%%%%%%%%%%%%%%%%%%%%%%%%%%%%%%%%%%%%%%%%%%%%%%%%%%%%%%%%%%%%%%%%%%%%%%%%%%%%%%%%

\section{Column space projection with block-encodings and QSVT}
\label{appendix: column space}

In this section, we prove Theorem~\ref{theorem: column space projection}.
Given a block-encoding of a matrix $A$ and quantum access to a state $\ket{\vec{x}}$, our goal is to prepare access to a quantum state that approximates $\frac{AA^+\ket{\vec{x}}}{\|AA^+\ket{\vec{x}}\|}$ and to estimate the norm $\|AA^+\ket{\vec{x}}\|$.

Let $A=U\Sigma V^\dagger$ be the singular value decomposition of $A$, then we can observe that $A^+ = V\Sigma^{-1} U^\dagger$ and $AA^+\ket{\vec{x}} = UU^\dagger \ket{x}$, which is a projection on the column space of $A$.
More importantly, we will use that $UU^\dagger \ket{x}=f(A)f(A^\dagger)\ket{\vec{x}}$, where $f(A)$ is a function mapping the singular values of $A$ to the constant value $1$ ($f(A)=UV^\dagger$ and $f(A^\dagger)=f(A)^\dagger$).

We use block-encodings, singular value transformation, and amplitude amplification and estimation to prove our theorem.
This Appendix is structured as follows.
The first section discusses how to approximate access to a state $\frac{A\ket{\vec{x}}}{\|A\ket{\vec{x}}\|}$ and estimate $\|A\ket{\vec{x}}\|$ using a block-encoding of a matrix $A$ and quantum access to $\ket{\vec{x}}$.
The second section describes Quantum Singular Value Transformation (QSVT) and discusses how to implement an approximate block-encoding of $f(A)$.
The last section puts everything together and concludes the proof.

\subsection{Matrix-vector multiplication and norm estimation}
Given a block-encoding $U$ of a matrix $A$, and quantum access to a vector $\vec{x}$, our goal is to produce a state $\ket{A\vec{x}} = \frac{A\Vec{x}}{\norm{A\vec{x}}}$ and to be able to estimate $\norm{A\vec{x}}$.
The algorithms from this section apply the block-encoding onto the quantum state and perform amplitude amplification to produce the desired state or estimation to estimate the norm.
An earlier version of the result that we are going to use appeared in \citet{bellante2023quantum}.
We state the differences in the proof.

\begin{theorem}[Matrix multiplication and norm estimation]
\label{thm: qomp matrix mul} 
    Let $U_A$ be a $(\alpha, q, \epsilon_0)$-block-encoding of a matrix $A \in \C ^{n\times m}$, implementable in time $T_A$. Let there be quantum access to a vector $\Vec{x} \in \C^{m}$ in time $T_x$. 
    Let $\epsilon > 0$.
    There exist quantum algorithms that output:
    \begin{enumerate}
        \item A classical estimate $\overline{t}$ of $t = \frac{\norm{A\Vec{x}}}{\norm{\Vec{x}}}$ such that $\abs{t - \overline{t}} \leq \epsilon$ with high probability in time $O\left((T_A + T_U)\frac{\alpha}{\epsilon}\right)$, provided $\epsilon_0 \leq \frac{\epsilon}{c}$, for any known constant $c$.
        \item A classical estimate $t$ of $\norm{A\Vec{x}}$ such that $\abs{\norm{A\Vec{x}} - t} \leq \epsilon\norm{A\Vec{x}}$ with high probability \\
        in expected time $\widetilde{O}\left((T_A + T_X)\frac{\alpha}{\epsilon} \frac{\norm{\Vec{x}}}{\norm{A\Vec{x}}}\right)$ if $\frac{\norm{A\Vec{x}}}{\norm{\Vec{x}}} \neq 0$ and otherwise runs forever, \\
        provided $\epsilon_0 \leq \frac{\epsilon}{c}$, for any known constant $c$.
        \item A quantum state $\ket{\Vec{z}}$ such that $\norm{\ket{\Vec{z}} - \frac{A\Vec{x}}{\norm{A\Vec{x}}}} \leq \epsilon$ in time $\widetilde{O}\left((T_A + T_X)\frac{\alpha}{\gamma}\right)$, provided that we know some lower bound $\gamma \leq \frac{\norm{A\Vec{x}}}{\norm{\Vec{x}}}$ and that $\epsilon_0 \leq \frac{\epsilon\gamma}{3}$.
        \item A quantum state $\ket{\Vec{z}}$ such that $\norm{\ket{\Vec{z}} - \frac{A\Vec{x}}{\norm{A\Vec{x}}}} \leq \epsilon$ in expected time ${\scriptstyle \widetilde{O}\left((T_A + T_X)\alpha\frac{\norm{\Vec{x}}}{\norm{A\Vec{x}}}\right)}$ if $\frac{\norm{A\Vec{x}}}{\norm{\Vec{x}}} \neq 0$ and otherwise runs forever, provided that $\epsilon_0 \leq \frac{\epsilon \norm{A\Vec{x}}}{3\norm{\Vec{x}}}$.
    \end{enumerate}
\end{theorem}
\begin{proof}
From the definition of block-encoding (Def.~\ref{def: block-encoding}), we have
$\left\|A-\alpha(\bra{0}^{\otimes q}\otimes \mathbb{I})U_A(\ket{0}^{\otimes q}\otimes \mathbb{I})\right\| \leq \epsilon_0.$

Let us define $A' = (\bra{0}^{\otimes q}\otimes \mathbb{I})U_A(\ket{0}^{\otimes q}\otimes \mathbb{I})$ as the matrix on the top-left corner of $U_A$, such that $\norm{\frac{A}{\alpha} - A'} \leq \frac{\epsilon_0}{\alpha}$, considering the possible zero-padding that makes $A$ a square matrix with size equal to a power of two.
Then,
\begin{align}
    U_A (I^{\otimes q} \otimes U_x) \ket{0} &= U_A \ket{0\cdots 0}\ket{\Vec{x}}\\
    &= \begin{bmatrix} A' & \cdot \\
    \cdot & \cdot
    \end{bmatrix}\begin{bmatrix} \Vec{x} \\
      \vec{0}
    \end{bmatrix} = \ket{0}^{q}A'\ket{\Vec{x}} + |0^{\perp}\rangle,
    \label{eq: qomp  matrix vector mul}
\end{align}
where $|0^{\perp}\rangle$ is unnormalized, with the first $q$ qubits orthogonal to the all-zero state $\ket{0}^{q}$.
The probability of measuring the first $q$ qubits in the state $\ket{0}^q$ is $ \mathrm{Pr}[\ket{0}^q] = \norm{A'\ket{\Vec{x}}}^2$.
\smallskip

1, 2) The two proofs proceed as in \citet[Appendix A, Theorem IV.6]{bellante2023quantum}.
Both use \emph{Absolute value amplitude estimation} (Theorem~\ref{thm: amplitude estimation amplitude}) on $\ket{0}^{q}$ and the second relies on \citet[Appendix D]{chowdhury2021computing} to obtain a multiplicative error bound of $\frac{\|A\vec{x}\|}{\|\vec{x}\|}$ and multiply the resulting estimate by $\|\vec{x}\|$. 
The estimation routine from \citet{chowdhury2021computing} is the reason why the second algorithm might not terminate.
\smallskip

3) Let $\gamma \leq \frac{\norm{A\Vec{x}}}{\norm{\Vec{x}}}$ be a lower bound.
We can run \emph{Fixed-point amplitude amplification} from Theorem~\ref{thm: qomp  fixed point amp amp} on the state of Eq.~(\ref{eq: qomp  matrix vector mul}), instead of amplitude estimation.
In our case, $\ket{\psi_0} = \ket{\vec{x}}$, $U=U_A$, $\Pi = \ket{0}^q\bra{0}^q$ and $a\ket{\psi'_G} = \norm{A'\ket{\vec{x}}} \frac{A'\ket{\vec{x}}}{\norm{A'\ket{\vec{x}}}}$. 
Since $\norm{A'\ket{\vec{x}}} \geq \frac{\norm{A\ket{\vec{x}}}}{\alpha} - \frac{\epsilon_0}{\alpha} \geq  \frac{\gamma - \epsilon_0}{\alpha}$, assuming $\epsilon_0 < \gamma$ we can run the fixed-point amplitude amplification routine with target precision $\epsilon/3$ for $O(\frac{\alpha}{\gamma - \epsilon_0}\log(1/\epsilon))$ rounds to obtain a quantum state $\ket{\psi''_G}$ that is $\epsilon/3$ close to $\ket{\psi'_G} = \frac{A'\vec{x}}{\norm{A'\vec{x}}}$. 

We proceed by studying how far $\ket{\psi''_G}$ is from $\frac{A\vec{x}}{\norm{A\vec{x}}}$.
Recall that $\norm{\norm{A\ket{\vec{x}}} - \alpha\norm{A'\ket{\vec{x}}}} \leq \norm{A\ket{\vec{x}} - \alpha A'\ket{\vec{x}}} \leq \epsilon_0$.
Then,
\begin{align}
    \norm{\frac{A\vec{x}}{\norm{A\vec{x}}} - |\vec{\psi}''_G\rangle } &\leq 
    \norm{\frac{A\vec{x}}{\norm{A\vec{x}}} - |\vec{\psi}'_G\rangle} + \frac{\epsilon}{3}\\
    &\leq 
    \norm{\frac{A\vec{x}}{\norm{A\vec{x}}} - \alpha\frac{A'\vec{x}}{\norm{A\vec{x}}}}
    + \norm{\alpha\frac{A'\vec{x}}{\norm{A\vec{x}}} - |\vec{\psi}'_G\rangle} + \frac{\epsilon}{3}\\ 
    &\leq \frac{\epsilon_0}{\norm{A\vec{x}}} + \norm{A'\vec{x}}\abs{\frac{\alpha\norm{A'\vec{x}} - \norm{A\vec{x}}}{\norm{A'\vec{x}}\norm{A\vec{x}}}} + \frac{\epsilon}{3}\\
    &\leq 2\frac{\epsilon_0}{\norm{A\vec{x}}} + \frac{\epsilon}{3}.
\end{align}
Choosing $\epsilon_0 \leq \frac{\epsilon\gamma}{3}$, we bound the above by $\epsilon$. 
Since we are bounding a norm between two quantum states, the reasonable range for $\epsilon$ should be $(0,2]$,
For any $\epsilon \in (0,2)$, the rounds of amplitude estimation become $O(\frac{\alpha}{\gamma}\log(1/\epsilon))$.
For any $\epsilon \geq 2$, outputting the $\ket{0}$ state would do.
\smallskip

4) The proof is similar to the above, but we need a routine to determine the lower bound $\gamma$. 
We can use the second result of this Theorem to obtain a relative-error estimate of $\mu = \frac{\norm{A\vec{x}}}{\norm{\vec{x}}}$.
We can run the relative error estimation routine with error $1/2$ to obtain an estimate $\overline{\mu}$ such that $\frac{1}{2}\mu \leq \overline{\mu} \leq \frac{3}{2}\mu$, in expected time $\widetilde{O}\left((T_A + T_X)\alpha \frac{\norm{\Vec{x}}}{\norm{A\Vec{x}}}\right)$.
Then, we set our lower bound to $\gamma = \frac{2}{3}\overline{\mu}$, obtaining $\frac{1}{3}\mu \leq \gamma \leq \mu$, and run the fixed-point amplitude amplification routine as in the proof above.
The randomness of the running time is due to the relative error estimation of the lower bound and the success probability (upon termination) can be adjusted through the \emph{Powering lemma} (Lemma~\ref{lemma: powering lemma (median)}) and t.
\end{proof}

If $A$ and $\vec{x}$ are stored in a quantum data structure in QRAM, then we obtain the following corollary.

\begin{corollary}[Matrix-vector multiplication with quantum data structures]
\label{coro:matrix mul}
Let $A \in \C^{n\times m}$ and $\Vec{x} \in \C^{m}$ stored in a quantum data structure.
There exist quantum algorithms that output:
    \begin{enumerate}
        \item A classical estimate $\overline{t}$ of $t = \frac{\norm{A\Vec{x}}}{\norm{\Vec{x}}}$ such that $\abs{t - \overline{t}} \leq \epsilon$ with high probability in time $\widetilde{O}\left(\frac{\mu(A)}{\epsilon}\right)$.
        \item A classical estimate $t$ of $\norm{A\Vec{x}}$ such that $\abs{\norm{A\Vec{x}} - t} \leq \eta\norm{A\Vec{x}}$ with high probability in expected time $\widetilde{O}\left(\frac{\mu(A)}{\epsilon} \frac{\norm{\Vec{x}}}{\norm{A\Vec{x}}}\right)$ if $\frac{\norm{A\Vec{x}}}{\norm{\Vec{x}}} \neq 0$ and otherwise runs forever.
        \item A quantum state $\ket{\Vec{z}}$ such that $\norm{\ket{\Vec{z}} - \frac{A\Vec{x}}{\norm{A\Vec{x}}}} \leq \epsilon$ in time $\widetilde{O}\left(\frac{\mu(A)}{\gamma}\right)$, provided that we know some bound $\gamma \leq \frac{\norm{A\Vec{x}}}{\norm{\Vec{x}}}$.
        \item A quantum state $\ket{\Vec{z}}$ such that $\norm{\ket{\Vec{z}} - \frac{A\Vec{x}}{\norm{A\Vec{x}}}} \leq \epsilon$ in expected time $\widetilde{O}\left(\mu(A)\frac{\norm{\Vec{x}}}{\norm{A\Vec{x}}}\right)$ if $\frac{\norm{A\Vec{x}}}{\norm{\Vec{x}}} \neq 0$ and otherwise runs forever. 
    \end{enumerate}
\end{corollary}

The proof requires creating a block-encoding of $A$ and using Theorem~\ref{thm: qomp matrix mul}.
It follows closely the one of \citet[Appendix A, Corollary IV.7]{bellante2023quantum}.

%%%%%%%%%%%%%%%%%%%%%%%%%%%%%%%%%%%%%%%%%%%%%%%%%%%%%%%%%%%%%%%%%%%%%%%%%%%%%%%%%%%%%%%%%%%%%
\subsection{Quantum singular value transformation and polynomial approximations}
\label{appendix: singular value transformation}

We revisit Quantum Singular Value Transformation (QSVT) and state a handy corollary for QSVT by odd real polynomials.
The following theorem shows how to implement polynomial QSVT on a block-encoded matrix $A$, combining Corollary 18, Lemma 19, and Definition 15 of the arxiv version of \citet{gilyen2019quantum} in one statement.

\begin{theorem}[Quantum singular value transformation by real polynomials{~\cite{gilyen2019quantum}}]
\label{theorem:Singular value transformation by real polynomials}
    Let $U \in \C^{n\times n}$ be a unitary matrix and $\Pi, \widetilde{\Pi} \in \C^{n \times n}$ be two orthogonal projectors.
    Suppose that $P \in \R[x]$ is an either even or odd degree-$d$ polynomial such that $\forall x \in [-1,1]:\abs{P(x)}\leq 1$.

    Then, there exist $\vec{\Phi} \in \R^d$, such that
    \begin{equation}
        P^{(SV)}(\widetilde{\Pi}U\Pi) = \begin{cases}
            (\bra{+} \otimes \widetilde{\Pi})(\ketbra{0}{0}\otimes U_\Phi + \ketbra{1}{1} \otimes U_{-\Phi})(\ket{+}\otimes \Pi) & \text{if $d$ is odd}\\
            (\bra{+} \otimes \Pi)(\ketbra{0}{0}\otimes U_\Phi + \ketbra{1}{1} \otimes U_{-\Phi})(\ket{+}\otimes \Pi) & \text{if $d$ is even}.
        \end{cases}
    \end{equation}
    The unitary
    \begin{equation}
        U_\Phi = \begin{cases}
            e^{i\phi_1(2\widetilde{\Pi} - I)}U\prod_{j=1}^{(d-1)/2}\left(e^{i\phi_{2j}(2\Pi - I)}U^\dagger e^{i\phi_{2j+1}(2\widetilde{\Pi} - I)}U\right) & \text{if $d$ is odd}\\
            \prod_{j=1}^{d/2}\left(e^{i\phi_{2j-1}(2\Pi - I)}U^\dagger e^{i\phi_{2j}(2\widetilde{\Pi} - I)}U\right) & \text{if $d$ is even}
        \end{cases}
    \end{equation}
    can be implemented using a single ancilla qubit and $O(d)$ uses of $U$, $U^\dagger$, $C_{\Pi}NOT$, $C_{\widetilde{\Pi}}NOT$ and single qubit gates. Similarly, for its controlled versions.
\end{theorem}

Here, a $C_{\Pi}NOT$ for a projector $\Pi$ is the controlled operation $C_{\Pi}NOT = \Pi \otimes X + (I-\Pi) \otimes I$ and the block-encoded matrix is $A=\widetilde{\Pi} U \Pi$.
Moreover, \citet[Lemma 19, arxiv version]{gilyen2019quantum} shows how to efficiently implement $e^{i\phi (2\Pi - I)}$ using a single auxiliary qubit as $e^{i\phi (2\Pi - I)} = C_{\Pi}NOT(I \otimes e^{-i\phi \sigma_z})C_{\Pi}NOT$, leading to an efficient $U_\Phi$.

In this paper, we focus on the application of real and odd polynomials.
Before stating our main corollary, we include a lemma that relates the error in the block-encoding to the resulting one on the polynomial SVT.
This lemma is a simplification of \citet[Lemma 22, arxiv version]{gilyen2019quantum} for real and odd polynomials.

\begin{lemma}[Robustness of singular value transformation~{\cite{gilyen2019quantum}}]
\label{lemma: qomp robustness of singular value transformation}
If $P \in \R[x]$ is an even or odd degree-$d$ polynomial such that $\forall x\in[-1,1]: \abs{P(x)\leq 1}$, moreover $A, \widetilde{A} \in \C^{N \times N}$ are matrices of operator norm at most $1$, then we have that
\begin{equation}
    \norm{P^{(SV)}(A) - P^{(SV)}(\widetilde{A})} \leq 4d\sqrt{\norm{A - \widetilde{A}}}.
\end{equation}
\end{lemma}
\begin{proof}
    We report the difference from \citet[Lemma 22, arxiv version]{gilyen2019quantum}.
    First, we can always use their Corollary 10 to make our real polynomial satisfy the conditions of their Corollary 8.
    Using the polynomial obtained by Corollary 10, we can prove the correctness by replacing their equation
    \begin{align}
        \norm{P^{(SV)}(A) - P^{(SV)}(\widetilde{A}/(1+\epsilon))} &= \norm{\Pi'U_{\Phi}\Pi - \Pi'\overline{U}_{\Phi}\Pi} \leq \norm{U_{\Phi} - \overline{U}_{\Phi}}
        &\leq d\norm{U - \overline{U}} \leq 2d\sqrt{\norm{A-\widetilde{A}}}
    \end{align}
    with 
    \begin{align}
        & \quad~ \norm{P^{(SV)}(A) - P^{(SV)}(\widetilde{A}/(1+\epsilon))} = \\
        &=\norm{(\bra{+} \otimes \Pi')(\braket{0}{0}\otimes U_\Phi + \braket{1}{1} \otimes U_{-\Phi})(\ket{+}\otimes \Pi) - (\bra{+} \otimes \Pi')(\braket{0}{0}\otimes \overline{U}_\Phi + \braket{1}{1} \otimes \overline{U}_{-\Phi})(\ket{+}\otimes \Pi)} \\
        &\leq \norm{\frac{\Pi'U_{\Phi}\Pi}{2} + \frac{\Pi'U_{-\Phi}\Pi}{2} - \frac{\Pi'\overline{U}_{\Phi}\Pi}{2} - \frac{\Pi'\overline{U}_{-\Phi}\Pi}{2} } \leq 
        \frac{\norm{U_{\Phi} - \overline{U}_{\Phi}}}{2} + \frac{\norm{U_{-\Phi} - \overline{U}_{-\Phi}}}{2}\\
        &\leq d\norm{U - \overline{U}} \leq 2d\sqrt{\norm{A-\widetilde{A}}}.
    \end{align}
    The proof then concludes like theirs.
\end{proof}

We are now ready to state our handy corollary for QSVT by real odd polynomials, which provides us guarantees on the accuracy a block-encoding of $P^{(SV)}\left(\frac{A}{\alpha}\right)$.

\begin{corollary}[QSVT by real and odd polynomial]
\label{corollary:QSVT by real and odd polynomial}
    Let $\delta \in [0,1]$ be a precision parameter. 
    Let $A \in \C^{n\times m}$ be a matrix with singular value decomposition $A = \sum_{i} \sigma_i\ket{u_i}\langle v_i^\dagger|$.
    Let $P \in \R[x]$ be an odd polynomial such that $\forall x \in [-1,1]: \abs{P(x)}\leq 1$. 
    Let $U_A$ be an $(\alpha, q, \epsilon)$-block-encoding of $A$, implementable in time $T_A$, with $\epsilon \leq \frac{\alpha \delta^2}{16d^2}$.
     
    Then, we can implement a $(1, q+2, \delta)$-block-encoding $U_P$ of 
    \begin{equation}
        P^{(SV)}\left(\frac{A}{\alpha}\right) \defeq \sum_{k=1}^r P\left(\frac{\sigma_k}{\alpha}\right)\ket{u}\langle v^\dagger|
    \end{equation}
    in time $O(dT_A)$.
\end{corollary}
\begin{proof}
By the definition of block-encoding (Def.~\ref{def: block-encoding}), $U_A$ is a $(1, q, \frac{\epsilon}{\alpha})$-block-encoding of $A'=\frac{A}{\alpha}$. 
Indeed,
    \begin{align}
        \norm{A'- (\langle 0|^{\otimes a}\otimes I ) U_A (|0\rangle^{\otimes a}\otimes I)}
        \leq \epsilon/\alpha.
    \end{align}
    Let $\widetilde{\Pi} = (\langle 0|^{\otimes q}\otimes I )$, $\Pi = (|0\rangle^{\otimes q} \otimes I)$ and $\widetilde{\Pi}U_A\Pi = \widetilde{A}$, so that $\|{A' - \widetilde{A}}\| \leq \epsilon/\alpha$.
    By Corollary~\ref{theorem:Singular value transformation by real polynomials}, we can implement $P^{(SV)}(\widetilde{A})$ in time $O(dT_A)$ using at most other $2$ auxiliary qubits and by Lemma~\ref{lemma: qomp robustness of singular value transformation}, we have $\norm{P^{(SV)}(A/\alpha) - P^{SV}(\widetilde{A})} \leq 4d\sqrt{\epsilon/\alpha}$. 
    To achieve final precision $\delta$, we require $\epsilon \leq \frac{\alpha \delta^2}{16d^2}$.
\end{proof}

In this section, we assumed that $\vec{\Phi}$ - the vector of rotations used in SVT - is available with sufficient (ideal) precision. 
In general, it is possible to classically compute $\vec{\Phi}$ to arbitrary precision $\xi$ in time $O(\mathrm{poly}(d,\log(1/\xi)))$~\cite{gilyen2019quantum}.

\subsubsection{Polynomial approximation of Sign and Step}
\label{appendix: polynomial approximation of sign and step}
We conclude this section by stating a real and odd polynomial approximation of the sign and step functions.
First, we report a result on the sign function.

\begin{lemma}[Polynomial approximation of the sign function~{\cite[Lemma 25, arxiv version]{gilyen2019quantum}}] 
\label{lemma: qomp Polynomial approximation of the sign function}
For all $\delta > 0$,  $\epsilon \in (0, 1/2)$ there exists an efficiently computable odd polynomial $P \in \R[x]$ of degree $n=O\left( \frac{\log(1/\epsilon)}{\delta}\right)$, such that
\begin{itemize}
    \item $\forall x \in [-2, 2]: \abs{P(x)} \leq 1$, and
    \item $\forall x \in [-2, 2]\setminus(-\delta, \delta): \abs{P(x) - \mathrm{sign}(x)} \leq \epsilon$.
\end{itemize}    
\end{lemma}

\begin{figure}[t]
    \centering
    \includegraphics[width=0.5\linewidth]{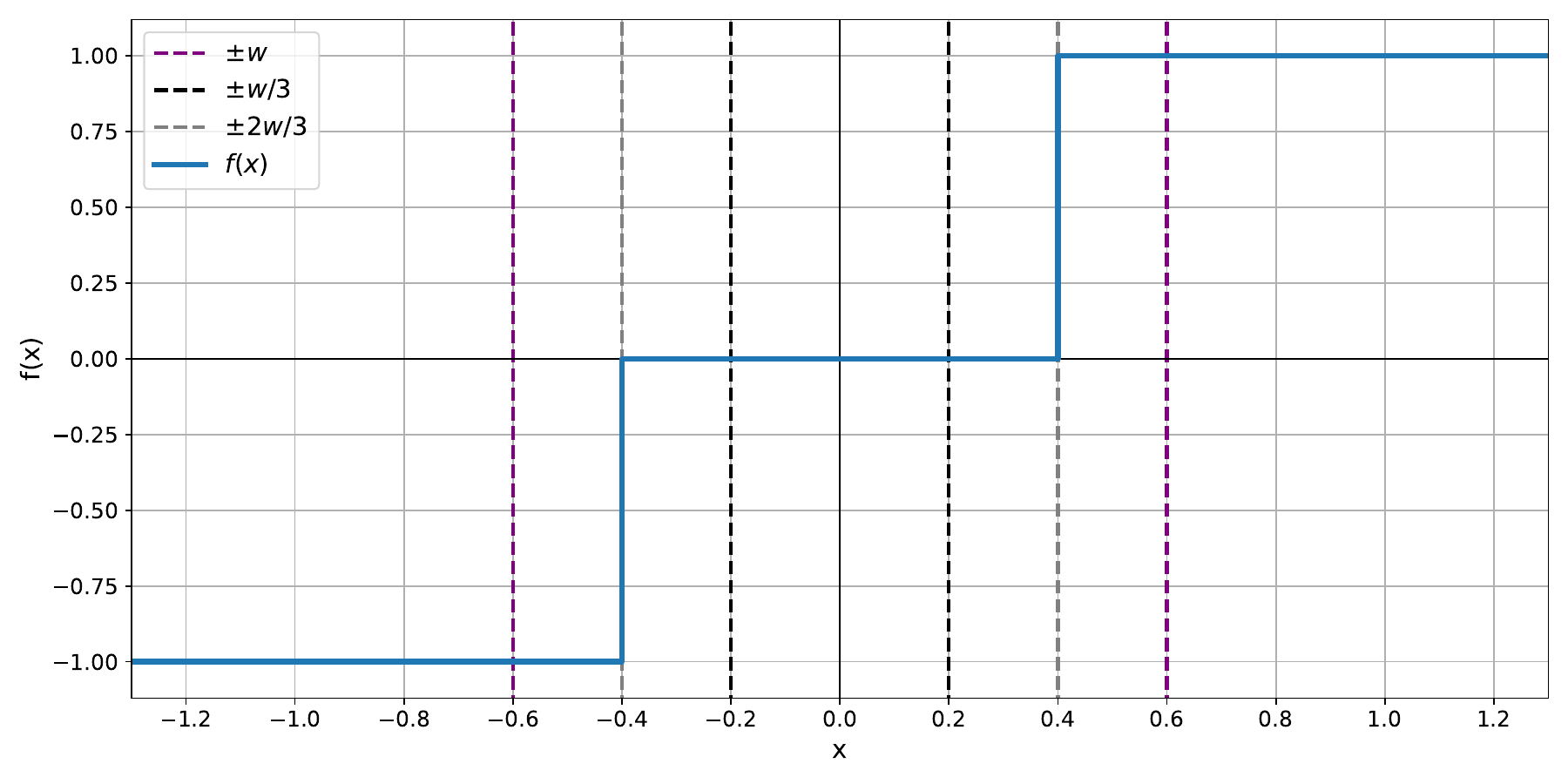}
    \caption{Antisymmetric step function $f(x) = \frac{1}{2}(\mathrm{sign}(x+\frac{2w}{3})+\mathrm{sign}(x-\frac{2w}{3}))$ with $w = 0.6$. It acts as a step function for $x>0$, with the step at $2w/3$.}
    \label{fig:antisymmetric step}
\end{figure}

If we want to make sure that the function is close to $0$ in a small interval around $x = 0$, we can approximate the step function on a positive domain (for $x \geq 0$) with a real odd polynomial, and complete the polynomial on $[-1,0)$ with an antysimmetric step function.
To perform the polynomial approximation we can create an antisymmetric step function (Figure~\ref{fig:antisymmetric step}) by manipulating the sign function and using its real odd polynomial approximation.
While we will use the sign function in our proofs, we expect that in practice it could be better to use a polynomial approximation of the antisymmetric step function, as forcing the function to be close to $0$ might help suppress errors even further.

\begin{lemma}[Polynomial approximation of the antisymmetric step function]
\label{lemma: qomp Polynomial approximation of the antisymmetric step function}
Let $w \in (0,1)$ and $\epsilon \in (0, 1/2)$.
There exists an efficiently computable odd polynomial $P \in \R[x]$ of degree $O\left(\frac{1}{w}\log(1/\epsilon)\right)$, such that 
    \begin{itemize}
        \item $\forall x \in [+w, 1]: \abs{1 - P(x)} \leq \epsilon$ and $\forall x \in [-1, -w]: \abs{-1 - P(x)} \leq \epsilon$.
        \item $\forall x \in [-w/3, +w/3]: \abs{P(x)} \leq \epsilon$,
        \item $\forall x \in [-1, 1]: \abs{P(x)} \leq 1$,
    \end{itemize}    
\end{lemma}
\begin{proof}
    Let $f(x) = \frac{1}{2}(\mathrm{sign}(x+\frac{2w}{3})+\mathrm{sign}(x-\frac{2w}{3}))$. It is easy to verify that an $\epsilon$-approximation of $f(x)$ might satisfy our needs, as
    \begin{itemize}
        \item $\forall x \in [w, +\infty): f(x) = 1$ and $\forall x \in (-\infty, -w]: f(x) = -1$,
        \item $\forall x \in [-w/3, +w/3]: f(x) = 0$.
    \end{itemize}
    We are going to build a polynomial approximation of $f(x)$ for all $x \in [-1, -w] \cup [-w/3, +w/3] \cup [+w, 1]$ starting from the one of the sign function.    
    We can use Lemma~\ref{lemma: qomp Polynomial approximation of the sign function} to construct a real odd polynomial $Q(x)$ such that
    \begin{equation}
    \label{eq: qomp  sign approx lemma proof}
        \abs{\mathrm{sign}(x) - Q(x)}\leq \epsilon \text{ for all } x \in [-2,2]\setminus\left(-\frac{w}{3}, \frac{w}{3}\right)
    \end{equation}
    and $\forall x \in [-2, 2]: \abs{Q(x)} \leq 1$.
    This requires degree $O\left( \frac{\log(1/\epsilon)}{w}\right)$.
    Then, we approximate $f(x)$ via the polynomial $P(x) = \frac{Q(x+2w/3) + Q(x-2w/3)}{2}$.
    We now show that this approximation satisfy the claims in the lemma.
    \begin{enumerate}
        \item \emph{Parity.}
        By construction, $P(x)$ is an efficiently computable real odd polynomial of degree $O\left( \frac{\log(1/\epsilon)}{w-\gamma}\right)$.
        Indeed, $Q(x) = -Q(-x)$ implies $P(-x) = \frac{Q(-x+2w/3) + Q(-x-2w/3)}{2} = \frac{-Q(x-2w/3) - Q(x+2w/3)}{2} = -P(x).$
        \item \emph{Approximation.} We have $\abs{f(x) - P(x)} \leq \frac{1}{2}(\abs{\mathrm{sign}(x+2w/3) - Q(x+2w/3)}+\abs{\mathrm{sign}(x-2w/3) - Q(x-2w/3)})$.
        Using Eq.~(\ref{eq: qomp  sign approx lemma proof}), we see that the first term is smaller than $\epsilon$ for all $x \in [-2,2]\setminus\left(-w, -w/3\right)$ and so is the second one for $x \in [-2,2]\setminus\left(+w/3, +w\right)$.
        This implies $\abs{f(x) - P(x)} \leq \epsilon$ for all $x \in [-1, -w] \cup [-w/3, +w/3] \cup [+w, 1]$.
        
        \item \emph{Boundedness.} We have $\abs{P(x)} \leq \frac{1}{2}(\abs{Q(x+2w/3)}+\abs{Q(x-2w/3)})$. 
        Using $\forall x \in [-2, 2]: \abs{Q(x)} \leq 1$, we have that the first term is bounded by $1$ for $x \in [-2-2w/3, 2-2w/3]$ and so is the second for $x \in [-2+2w/3, 2+2w/3]$.
        It follows that $\abs{P(x)}\leq 1$ for $[-1,1] \subset [-(2+2w/3), 2+2w/3]$.
    \end{enumerate}    
\end{proof}

%%%%%%%%%%%%%%%%%%%%%%%%%%%%%%%%%%%%%%%%%%%%%%%%%%%%%%%%%%%%%%%%%%%%%%%%%%%%%%%%%%%%%%%%%%%%%
\subsection{Column space projection}
\label{appendix: column space projection}
We are finally ready to prove our result. 
We will perform QSVT on the block-encoding of $A$ and $A^\dagger$, combine the block-encodings and use matrix-vector multiplication.
To combine block-encodings, we use the following result.

\begin{lemma}[Product of block-encoded matrices~{\cite[Lemma 53, arxiv]{gilyen2019quantum}}]
\label{lemma: qomp Product of block-encoded matrices}
    If $U$ is an $(\alpha, a, \delta)$-block-encoding of an s-qubit operator $A$, and $V$ is a $(\beta, b, \epsilon)$-block-encoding of an s-qubit operator $B$, then\footnote{Here, the identity operator $\mathbb{I}_b$ should be seen as acting on the ancilla qubits of $V$, and $\mathbb{I}_a$ on those of $U$.} $(I_b \otimes U)(I_a \otimes V)$ is an $(\alpha\beta, a+b, \alpha\epsilon + \beta\delta)$-block-encoding of $AB$.
\end{lemma}

We can now state the complexity of preparing a block-encoding of $UU^\dagger$.

\begin{lemma}[Block-encoding of $UU^\dagger$]
\label{lemma:Block encoding of UUT}
Let $A \in \C^{n \times m}$ be a matrix with singular values decomposition $A=U\Sigma V^\dagger$ and singular values in $[\sigma_{\min}(A), \norm{A}]$, with a known lower bound $\gamma \leq \sigma_{\min}(A)$.
Let $U_A$ be a $(\alpha, q, \epsilon_A)$-block-encoding of $A$ implementable in time $T_A$, with $\epsilon_A \leq \frac{\gamma^2 \epsilon^2}{c\alpha\log^2(1/\epsilon)}$ for a certain constant $c$.
Then, there exists a quantum algorithm that implements a $(1, 2(q+2), \epsilon)$-block-encoding of $UU^\dagger$ in time $O\left(\frac{\alpha}{\gamma}\log(1/\epsilon)T_A\right)$.
\end{lemma}
\begin{proof}
    The plan is to implement a block-encoding \begin{align}
        \mathrm{sign}(A)\mathrm{sign}(A^\dagger) = \sum_i \mathrm{sign}^2(\sigma_i)| \vec{u}_i\rangle \langle \vec{u}_i | = UU^\dagger
    \end{align}
    via \emph{QSVT by real and odd polynomial} (Corollary \ref{corollary:QSVT by real and odd polynomial}) and \emph{Product of block-encoded matrices} (Lemma~\ref{lemma: qomp Product of block-encoded matrices}).
    In the remainder, let $A' = A/\alpha$.

    Let $\mathrm{sign}(x)$ be approximated by $P\in \R[x]$, a degree-$d$ odd polynomial such that $\forall x \in [-1,1]: \abs{P} \leq 1$ (Lemma~\ref{lemma: qomp Polynomial approximation of the sign function}).
    Using Corollary~\ref{corollary:QSVT by real and odd polynomial}, we can implement $(1, q+2, \epsilon/6)$-block-encodings $U_{P(A')}$ and $U_{P(A'^\dagger)}$ of $P^{(SV)}(A')$ and $P^{(SV)}(A'^\dagger)$ in time $O(dT_A)$, provided $\epsilon_A \leq \frac{\alpha \epsilon^2}{16d^2}$.
    The spectrum of $A'$ and $A'^\dagger$ lies in $[\sigma_{\min}(A)/\alpha, \norm{A}/\alpha] \subseteq [\gamma/\alpha, 1]$, therefore we can require $P$ to approximate $\mathrm{sign}$ in $[-1,1] \setminus (-\frac{\gamma}{\alpha}, \frac{\gamma}{\alpha})$, leading to time complexity $O(\frac{\alpha}{\gamma}\log(1/\epsilon)T_A)$ and imposing the requirement $\epsilon_A \leq \frac{\gamma^2 \epsilon^2}{c\alpha^2\log^2(1/\epsilon)}$, for some constant $c$.
    In particular, we can require precision $\epsilon/6$.

    Let $\widetilde{\Pi} = (\langle 0|^{\otimes q+2}\otimes I )$ and $\Pi = (|0\rangle^{\otimes q+2} \otimes I)$. 
    Using Lemma~\ref{lemma: qomp Product of block-encoded matrices}, we can implement a $(1,2(q+2), \epsilon/3)$-block-encoding $U_F$ of the product $P(A')P(A'^\dagger)$ and use it as our approximation of $UU^\dagger$.

    The block-encoding error is proven by the following inequalities
    \begin{align}
        & \quad~ \norm{UU^\dagger - (\langle 0|^{\otimes 2(q+2)}\otimes I) U_F (|0\rangle^{\otimes 2(q+2)} \otimes I)} \leq\\
        &\leq \norm{UU^\dagger - {\widetilde{\Pi}U_{P(A')}\Pi\widetilde{\Pi}U_{P(A'^\dagger )}\Pi}} + \norm{{\widetilde{\Pi}U_{P(A')}\Pi\widetilde{\Pi}U_{P(A'^\dagger )}\Pi} - (\langle 0|^{\otimes 2(q+2)}\otimes I) U_F (|0\rangle^{\otimes 2(q+2)} \otimes I)}
        \\
        &\leq \norm{\mathrm{sign}(A)\mathrm{sign}(A^\dagger) - P(A')P(A'^\dagger)} + \norm{P(A')P(A'^\dagger) - \widetilde{\Pi}U_{P(A')}\Pi\widetilde{\Pi}U_{P(A'^\dagger )}\Pi} + 2\epsilon_1
        \\
        &\leq 2(\epsilon/6) + \norm{P(A') - \widetilde{\Pi}U_{P(A')}\Pi} + \norm{P(A'^\dagger) - \widetilde{\Pi}U_{P(A'^\dagger)}\Pi} + \epsilon/3\\
        & \leq \epsilon/3 + 2(\epsilon/6) + \epsilon/3 \leq \epsilon.
    \end{align}
\end{proof}

We stress once again that in the procedure above, we used the polynomial approximation of the sign function. 
We expect that in practice it could be better to use the antisymmetric step function defined in the previous section, as it might help suppress errors even further.
In any case, we are ready to prepare $\ket{AA^+ \vec{x}}$ and estimate its norm.
We report the statement of Theorem~\ref{theorem: column space projection} and conclude the proof.

\begin{theorem}[Column space projection]
    Let $\epsilon > 0$ be a precision parameter. 
    Let $U_A$ be a $(\alpha, q, \epsilon_A)$-block-encoding of a matrix $A \in \C ^{n\times m}$, implementable in time $T_A$, and let a lower bound $\gamma \leq \sigma_{\min}(A)$ be known. 
    Let there be quantum access to a vector $\Vec{x} \in \C^{n}$ of known norm $\norm{\vec{x}}_2$ in time $T_x$ via a unitary $U_x$. 
    Then, there exists a constant $c \in \R^+$ such that if $\epsilon_A \leq \frac{\norm{AA^+\vec{x}}^2\gamma^2 \epsilon^2}{c\norm{\vec{x}}^2\alpha\log^2(\norm{\vec{x}}/(\norm{AA^+\vec{x}}\epsilon))}$ there are quantum algorithms that:
    \begin{enumerate}
        \item Create a quantum state $|\vec{\phi}\rangle$ such that $\norm{|\vec{\phi}\rangle - \ket{AA^+ \vec{x}}}_2 \leq \epsilon$ in expected time $\widetilde{O}\left(\frac{\norm{\vec{x}}}{\norm{AA^+ \vec{x}}}( \frac{\alpha}{\gamma} T_A + T_x)\right)$ if $\norm{AA^+\vec{x}} \neq 0$ and otherwise runs forever. 
        \item Produce an estimate $t$ such that $\abs{t - \norm{AA^+\vec{x}}_2}\leq \epsilon$ with high probability  in time $\widetilde{O}\left(\frac{1}{\epsilon}(\frac{\alpha}{\gamma} T_A + T_x)\right)$;
        \item Produce an estimate $t$ such that $\abs{t - \norm{AA^+\vec{x}}_2}\leq \epsilon\norm{AA^+ \vec{x}}_2$ with high probability \\in expected time $\widetilde{O}\left(\frac{1}{\epsilon}\frac{\norm{\vec{x}}}{\norm{AA^+ \vec{x}}}(\frac{\alpha}{\gamma} T_A + T_x)\right)$.
    \end{enumerate}
\end{theorem}
\begin{proof}
By Lemma~\ref{lemma:Block encoding of UUT}, we can create a $(1, 2(q+2), \epsilon_U)$-block-encoding of $AA^+=UU^\dagger \in \C^{n \times n}$ in time $T_U=O\left(\frac{\alpha}{\gamma}\log(1/\epsilon_U)T_A\right)$, provided $\epsilon_A \leq \frac{\gamma^2 \epsilon_U^2}{c_0 \alpha\log^2(1/\epsilon_U)}$ for some computable constant $c_0$.
Now, we can use Theorem~\ref{thm: qomp matrix mul} (points 1, 2, and 4) for the three tasks:
\begin{enumerate}
    \item to create $|\vec{\phi}\rangle$ to additive precision $\epsilon$, we need $\epsilon_U \leq \frac{\epsilon\norm{AA^+\vec{x}}}{3\norm{\vec{x}}}$ and expected time $\widetilde{O}((T_U+T_X)\frac{\norm{\vec{x}}}{\norm{AA^+\vec{x}}})$;
    \item to estimate $\norm{AA^+\vec{x}}$ to additive precision $\epsilon$, we need $\epsilon_U \leq \epsilon/c_1$ for some computable constant $c_1$ and time $O((T_U + T_x)/\epsilon)$;
    \item to estimate $\norm{AA^+\vec{x}}$ to relative precision $\epsilon$,  we need $\epsilon_U \leq \epsilon/c_2$ for some computable constant $c_2$ and expected time $O((T_U + T_x)\frac{\norm{\vec{x}}}{\norm{AA^+\vec{x}}\epsilon})$.
\end{enumerate}
The proof follows easily from here.
\end{proof}

%%%%%%%%%%%%%%%%%%%%%%%%%%%%%%%%%%%%%%%%%%%%%%%%%%%%%%%%%%%%%%%%%%%%%%%%%%%%%%%%%%%%%%%%%%%%%
%%%%%%%%%%%%%%%%%%%%%%%%%%%%%%%%%%%%%%%%%%%%%%%%%%%%%%%%%%%%%%%%%%%%%%%%%%%%%%%%%%%%%%%%%%%%%

\section{QOMP's iteration cost: Errors and running time analysis}
\label{appendix: QOMP error and running time analysis}
This appendix constitutes a proof of \emph{QOMP's Iteration cost} (Theorem~\ref{theorem: QOMP iteration cost}).
We first analyze all the sources of errors in the algorithm, and then discuss the running time.

\subsection{Errors}
At each iteration, QOMP retrieves the index of an atom such that 
\begin{align}
    j = \argmax_{k \in \overline{\Lambda}} |\bracket{\vec{d}_k}{\vec{r}}| - 2\epsilon_i
\end{align}
where $\epsilon_i$ is the error of the inner product oracle $O_i$ (Eq.~(\ref{eq: Oi})). 
Furthermore, it evaluates the stopping condition using an estimate of $\|\vec{r}\|$ to error $\epsilon_f$.
In this section, we study the approximation error sources of QOMP and analyze the required precision of each step as a function of $\epsilon_i$ and $\epsilon_f$.
We will not try to optimize for the constant terms, but to establish the asymptotic scaling of the errors, which is the relevant quantity for our running time analysis.

We consider exact access to the target vector $\ket{\vec{s}}$, its norm $\|\vec{s}\|$, and to the dictionary entries $\{|\vec{d}_j\rangle\}_{j \in [m]}$.
We summarize the other error sources in the following boxes, providing notation for all the individual error terms.

Atom selection:
\begin{equation*}
    \boxed{
\begin{aligned}
    \abs{\Re[z_{1j}] - \Re[\inner{\vec{d}_j}{\vec{s}}]} &\leq \epsilon_{1\Re} \, &\text{(Theorem~\ref{thm:innerproductestimation})}\\
    \abs{\Im[z_{1j}] - \Im[\inner{\vec{d}_j}{\vec{s}}]} &\leq \epsilon_{1\Im} \, &\text{(Theorem~\ref{thm:innerproductestimation})}\\
    \abs{|\overline{\vec{\phi}}\rangle - |\vec{\phi}\rangle} &\leq \epsilon_{1\phi}, \, &\text{(Theorem~\ref{theorem: column space projection})} \\
    \abs{\overline{\|\vec{\phi}\|} - \|\vec{\phi}\|} &\leq \epsilon_{1\norm{\phi}}, \, &\text{(Theorem~\ref{theorem: column space projection})}\\
    \abs{\Re[z_{2j}] - \Re[\inner{\vec{d}_j}{\overline{\vec{\phi}}}]} &\leq \epsilon_{2\Re} \, &\text{(Theorem~\ref{thm:innerproductestimation})}\\
    \abs{\Im[z_{2j}] - \Im[\inner{\vec{d}_j}{\overline{\vec{\phi}}}]} &\leq \epsilon_{2\Im} \, &\text{(Theorem~\ref{thm:innerproductestimation})}
\end{aligned}
    }
\end{equation*}
Exit condition:
\begin{equation*}
    \boxed{
\begin{aligned}
    \abs{|\overline{\vec{\phi}}\rangle - |\vec{\phi}\rangle} &\leq \epsilon_{2\phi}, \, &\text{(Theorem~\ref{theorem: column space projection})} \\
    \abs{\overline{\|\vec{\phi}\|} - \|\vec{\phi}\|} &\leq \epsilon_{2\norm{\phi}}, \, &\text{(Theorem~\ref{theorem: column space projection})}\\
    \abs{z_f - \norm{\norm{s}\ket{\vec{s}} - \overline{\norm{\phi}}\ket{\overline{\phi}}}} &\leq \epsilon_w, \, &\text{(Theorem~\ref{theorem: weighted euclidean distance estimation})}
\end{aligned}
    }
\end{equation*}

\subsubsection{Inner products}
We begin by analyzing the propagation of errors in the inner products at a generic iteration.
We start by recalling
\begin{align}
    z_j \simeq |\inner{\vec{d}_j}{\vec{r}}| = |\inner{\vec{d}_j}{\vec{s}} - \inner{\vec{d}_j}{\vec{\phi}}|.
\end{align}
Hence,  the error on $\inner{\vec{d}_j}{\vec{r}}$ arises from the approximations of both $\inner{\vec{d}_j}{\vec{s}}$ and $\inner{\vec{d}_j}{\vec{\phi}}$.
We assume exact access to $|\vec{d}_j\rangle$, $\ket{\vec{s}}$, and $|\vec{s}|$, while $|\vec{\phi}\rangle$ and $|\vec{\phi}|$ are available only approximately.

Since squared values appear repeatedly in the definition of $z_j$, we begin with a generic bound
\begin{align}
\label{eq: errors and squares}
    \abs{a - \overline{a}} \leq \epsilon \implies \abs{a^2 - \overline{a}^2} \leq (2|a| + \epsilon)\epsilon.
\end{align}
Using this tool, we can proceed to bound many other terms.

First, both the real and imaginary parts of $\inner{\vec{d}_j}{\vec{s}}$ and $\inner{\vec{d}_j}{\vec{\phi}}$ are bounded by $1$ in magnitude.
Hence, considering error terms smaller than one, we have $|\Re[z_{1j}] - \Re[\inner{\vec{d}_j}{\vec{s}}]| \leq \epsilon/4 \implies |\Re[z_{1j}]^2 - \Re[\inner{\vec{d}_j}{\vec{s}}]^2| \leq \epsilon,$ which holds for all the four $\epsilon_{1\Re}, \epsilon_{1\Im}, \epsilon_{2\Re}$, $\epsilon_{2\Im}$.
Similarly, since $\|\vec{\phi}\| \leq \|\vec{s}\|$, we have $|\overline{\|\vec{\phi}\|} - \|\vec{\phi}\|| \leq \frac{\epsilon}{4\|\vec{s}\|} \implies |\overline{\|\vec{\phi}\|}^2 - \|\vec{\phi}\|^2| \leq \epsilon.$
Finally, we decompose the error on $\bracket{\vec{d}_j}{\vec{\phi}}$ as $|\Re[\langle \vec{d}_j | \vec{\phi} \rangle] - \Re[z_{2j}]|
    \leq |\Re[\langle \vec{d}_j | \vec{\phi} \rangle] - \Re[\langle \vec{d}_j | \overline{\vec{\phi}} \rangle|
    + |\Re[\langle \vec{d}_j | \overline{\vec{\phi}} \rangle - \Re[z_{2j}]|.$
This yields $|\Re[\langle \vec{d}_j | \vec{\phi} \rangle] - \Re[z_{2j}]| \leq \epsilon_{1\phi} + \epsilon_{2\Re}$, with an analogous inequality for the imaginary part.

Combining the above estimates, the deviation of $z_j$ from $\inner{\vec{d}_j}{\vec{r}}$ satisfies
\begin{align}
    |z_j - \inner{\vec{d}_j}{ \vec{r}}| \leq
\end{align}
\begin{align}
    \resizebox{\linewidth}{!}{    
        $\norm{\vec{s}}^2\abs{\Re[\langle \vec{d}_j | \vec{s} \rangle]^2 - \Re[z_{1j}]^2} + 
        2\|\vec{s}\|\abs{\|\vec{\phi}\|\Re[\langle \vec{d}_j | \vec{s} \rangle]\Re[\langle \vec{d}_j | \vec{\phi} \rangle] - \overline{\|\vec{\phi}\|}\Re[z_{1j}]\Re[z_{2j}]} +
        \abs{\|\vec{\phi}\|^2\Re[\langle \vec{d}_j | \vec{\phi} \rangle]^2 - \overline{\|\vec{\phi}\|}^2\Re[z_{2j}]^2} + $
    }\\
    \resizebox{\linewidth}{!}{    
        $\norm{\vec{s}}^2\abs{\Im[\langle \vec{d}_j | \vec{s} \rangle]^2 - \Im[z_{1j}]^2} + 
        2\|\vec{s}\|\abs{\|\vec{\phi}\|\Im[\langle \vec{d}_j | \vec{s} \rangle]\Im[\langle \vec{d}_j | \vec{\phi} \rangle] - \overline{\|\vec{\phi}\|}\Im[z_{1j}]\Im[z_{2j}]} +
        \abs{\|\vec{\phi}\|^2\Im[\langle \vec{d}_j | \vec{\phi} \rangle]^2 - \overline{\|\vec{\phi}\|}^2\Im[z_{2j}]^2}$
    }\\
    \leq \norm{\vec{s}}^24\epsilon_{1\Re} + 2 \norm{\vec{s}} (\epsilon_{1\norm{\phi}} + \overline{\|\vec{\phi}\|}(\epsilon_{1\Re} + \epsilon_{1\phi}  + \epsilon_{2\Re})) + 4\norm{\vec{s}}\epsilon_{1\norm{\phi}} + \overline{\|\vec{\phi}\|}4(\epsilon_{1\phi} + \epsilon_{2\Re})\\
    + \norm{\vec{s}}^24\epsilon_{1\Im} + 2 \norm{\vec{s}} (\epsilon_{1\norm{\phi}} + \overline{\|\vec{\phi}\|}(\epsilon_{1\Im} + \epsilon_{1\phi}  + \epsilon_{2\Im})) + 4\norm{\vec{s}}\epsilon_{1\norm{\phi}} + \overline{\|\vec{\phi}\|}4(\epsilon_{1\phi} + \epsilon_{2\Im})\\
    \leq 8\norm{\vec{s}}^2(\epsilon_{1\Re} + \epsilon_{1\Im}) + 8\norm{\vec{s}}\overline{\|\vec{\phi}\|}(\epsilon_{2\Re} + \epsilon_{2\Im}) + 12 \norm{\vec{s}} \epsilon_{1 \norm{\phi}} + 16\norm{\vec{s}}\overline{\|\vec{\phi}\|} \epsilon_{1\phi}.
\end{align}

To guarantee $|z_j - \inner{\vec{d}_j}{ \vec{r}}| \leq \epsilon_i$, it suffices to choose $\epsilon_{1\Re} = \epsilon_{1\Im} \leq \frac{\epsilon_i}{48\norm{s}^2}$, $\epsilon_{2\Re} = \epsilon_{2\Im} \leq \frac{\epsilon_i}{48\norm{\vec{s}} \overline{\|\vec{\phi}\|}}$, $\epsilon_{1\norm{\phi}} \leq \frac{\epsilon_i}{72\norm{\vec{s}}}$, $\epsilon_{1\phi} \leq \frac{\epsilon_i}{96\norm{\vec{s}}\overline{\|\vec{\phi}\|}}$.
As a remark, in the first iteration, where $z_j=\norm{\vec{s}}^2\Re[z_{1j}]^2 + \norm{\vec{s}}^2\Im[z_{1j}]^2$, a weaker condition suffices: $\epsilon_{1\Re} = \epsilon_{1\Im} \leq \epsilon_i/(8\norm{\vec{s}}^2)$.

\subsubsection{Norm estimation}
To estimate the residual's norm we approximate equation (\ref{eq: norm estimation}) using \emph{Weighted Euclidean distance estimation} (Theorem \ref{theorem: weighted euclidean distance estimation}) with $\|\vec{s}\|$, $\ket{\vec{s}}$, and our approximations of $\|\vec{\phi}\|$ and $|\vec{\phi}\rangle$, computed through \emph{Column space projection} (Theorem \ref{theorem: column space projection}).
Let $z_f$ be the output of the weighted Euclidean distance estimation, such that 
\begin{align}
    \norm{\norm{\norm{\vec{s}}\ket{\vec{s}} - \overline{\|\vec{\phi}\|} |\overline{\vec{\phi}}\rangle} - z_f} \leq \epsilon_w.    
\end{align}
Then, using the reverse triangular inequality, 
\begin{align}
    \abs{\norm{\vec{r}} - \overline{\norm{\vec{r}}}} &= 
    \abs{\|\vec{s} - \vec{\phi}\| - z_f}\\
    &\leq \abs{ \| \vec{s} - \vec{\phi} \| - \norm{\norm{\vec{s}}\ket{\vec{s}} - \overline{\| \vec{\phi}\|}| \overline{\vec{\phi}}\rangle }} + \abs{\norm{\norm{\vec{s}}\ket{\vec{s}} - \overline{\| \vec{\phi}\|}| \overline{\vec{\phi}}\rangle} - z_f}\\
    &\leq \abs{(\vec{s}-\vec{\phi}) - \left(\norm{\norm{\vec{s}}\ket{\vec{s}} - \overline{\| \vec{\phi}\|}| \overline{\vec{\phi}}\rangle}\right)}+\epsilon_w\\
    &\leq \epsilon_w + \abs{\|\vec{\phi}\||\vec{\phi}\rangle - \overline{\| \vec{\phi}\|} | \vec{\phi} \rangle} + \abs{\overline{\| \vec{\phi}\|}| \vec{\phi} \rangle - \overline{\| \vec{\phi}\|}| \overline{\vec{\phi}}\rangle}\\
    &\leq \epsilon_w + \epsilon_{2\norm{\phi}} + \overline{\|\vec{\phi}\|}\epsilon_{2\phi}.
\end{align}

Hence, to guarantee an overall error $\leq \epsilon_f$, it suffices to choose $\epsilon_w \leq \frac{\epsilon_f}{3}, \epsilon_{2\norm{\phi}} \leq \frac{\epsilon_f}{3}, \epsilon_{2\phi} \leq \frac{\epsilon_f}{3\overline{\|\vec{\phi}\|}}$.

\subsection{Running time}
After the error analysis, we can study the asymptotic running time of one QOMP algorithm iteration.
In particular, we will choose $\epsilon_{1\Re} = \epsilon_{1\Im} \leq \frac{\epsilon_i}{48\norm{\vec{s}}^2}$, $\epsilon_{2\Re} = \epsilon_{2\Im} \leq \frac{\epsilon_i}{48\norm{\vec{s}} \overline{\|\vec{\phi}\|}}$, $\epsilon_{1\norm{\phi}} \leq \frac{\epsilon_i}{72\norm{\vec{s}}}$, $\epsilon_{1\phi} \leq \frac{\epsilon_i}{96\norm{\vec{s}}\overline{\|\vec{\phi}\|}}$ to compute the inner products $\inner{\vec{d}_j}{\vec{r}}$ to precision $\epsilon_i$ and errors $\epsilon_w \leq \frac{\epsilon_f}{3}, \epsilon_{2\norm{\phi}} \leq \frac{\epsilon_f}{3}, \epsilon_{2\phi} \leq \frac{\epsilon_f}{3\overline{\|\vec{\phi}\|}}$ to evaluate $\norm{\vec{r}}_2$ to precision $\epsilon_f$.

\subsubsection{Atom selection}
Since the cost of the first iteration is lower, we analyze a generic iteration after the first one.

The cost of computing the first inner product $\bracket{\vec{d}_j}{\vec{s}}$, using Theorem~\ref{thm:innerproductestimation} with $U_D$ and $U_s$, is
\begin{align}
    \Ot \left((T_s + T_D)\left(\frac{1}{\epsilon_{1\Re}} +\frac{1}{\epsilon_{1\Im}} \right)\right)
\end{align}
Since we can set $\epsilon_{1\Re} = \epsilon_{1\Im}$ and $\epsilon_{2\Re} = \epsilon_{2\Im}$, we merge these into a single term $1/\epsilon_{1\Re}$. The same simplification applies later for $\epsilon_{2\Re}$ and $\epsilon_{2\Im}$.

The next step is to implement $U_\phi$ and compute the estimate $\overline{\norm{\phi}}$. 
We do so thanks to Theorem~\ref{theorem: column space projection}, considering that we can implement a block-encoding of $D_\Lambda$ in time $T_A$ (we will further detail this cost later on, at the end of our analysis).
The unitary $U_\phi$ requires expected time $\Ot\left( \frac{\norm{\vec{s}}}{\|\vec{\phi}\|}\left(\frac{\alpha}{\gamma}T_A + T_s\right) \right)$, where $\gamma$ is a lower bound on $\sigma_{\min}(D_\Lambda)$ and $\alpha$ is the normalization factor of the block-encoding of $D_\Lambda$.
Using the same theorem, the norm estimation can be performed to precision $\epsilon_{1\|\phi\|}$ in time $\Ot \left(\frac{1}{\epsilon_{1 \norm{\phi}} }\left(\frac{\alpha}{\gamma}T_A + T_s\right)\right)$.
Merging these times with the second inner product estimation and the subtraction, we obtain the cost of implementing the oracle $O_i$ of Eq.~(\ref{eq: Oi})
\begin{align}
    \Ot \left( \frac{1}{\epsilon_{1\norm{\phi}} }\left(\frac{\alpha}{\gamma}T_A + T_s\right)+ (T_s + T_D)\frac{1}{\epsilon_{1\Re}}  + \left(\frac{\norm{\vec{s}}}{\|\vec{\phi}\|}\left(\frac{\alpha}{\gamma}T_A + T_s\right) + T_D\right)\frac{1}{\epsilon_{2\Re}}\right).
\end{align}

Using \emph{Finding the maximum with an approximate unitary} from Corollary~\ref{coro: finding approximate maximum} on the subset of indices created by $U_{\overline{\Lambda}}$, we estimate that the cost of the atom selection procedure is
\begin{align}
    \Ot \left( \frac{1}{\epsilon_{1 \norm{\phi}} }\left(\frac{\alpha}{\gamma}T_A + T_s\right) + \sqrt{m} \left( T_{\overline{\Lambda}} + (T_s + T_D)\frac{1}{\epsilon_{1\Re}} + \left(\frac{\norm{\vec{s}}}{\|\vec{\phi}\|}\left(\frac{\alpha}{\gamma}T_A + T_s\right) + T_D\right)\frac{1}{\epsilon_{2\Im}}\right)\right).
\end{align}

Substituting the errors as a function of $\epsilon_i$, we get
\begin{align}
    \Ot \left( \frac{\norm{\vec{s}}}{\epsilon_i }\left(\frac{\alpha}{\gamma}T_A + T_s\right) + \sqrt{m} \left( T_{\overline{\Lambda}} + (T_s + T_D)\frac{\norm{\vec{s}}^2}{\epsilon_i} + \left(\frac{\norm{\vec{s}}}{\|\vec{\phi}\|}\left(\frac{\alpha}{\gamma}T_A + T_s\right) + T_D\right)\frac{\norm{\vec{s}}\overline{\|\vec{\phi}\|}}{\epsilon_i}\right)\right).
\end{align}
Considering $\norm{\vec{s}} \geq 1$,  $\frac{\overline{\|\vec{\phi}\|}}{\|\vec{\phi}\|} \rightarrow 1$ for $\epsilon_i \rightarrow 0$, we obtain
\begin{align}
\label{eq: qomp  time atom selection}
    \Ot \left( \sqrt{m} T_{\overline{\Lambda}} + \sqrt{m}\frac{\norm{\vec{s}}^2}{\epsilon_i} \left( T_s + \frac{\alpha}{\gamma}T_A + T_D \right) \right).
\end{align}

\subsubsection{Exit condition}
To estimate the residual's norm, we once again build $U_{\phi}$ and compute $\overline{\|\vec{\phi}\|}$, with precision $\epsilon_{2\phi}$ and $\epsilon_{2\norm{\phi}}$, and run the \emph{Weighted Euclidean distance estimation} of Theorem~\ref{theorem: weighted euclidean distance estimation}.
This requires time
\begin{align}
    \Ot \left(\frac{1}{\epsilon_{2 \norm{\phi} } }\left(\frac{\alpha}{\gamma}T_A + T_s\right) + \left(\frac{\norm{\vec{s}}}{\|\vec{\phi}\|}\left(\frac{\alpha}{\gamma}T_A + T_s\right) + T_s\right) \frac{\norm{\vec{s}}\overline{\|\vec{\phi}\|}}{\epsilon_w}\right).
\end{align}
Treating the error terms as a function of $\epsilon_f$ and considering $\norm{\vec{s}} \geq 1$,  $\frac{\overline{\|\vec{\phi}\|}}{\|\vec{\phi}\|} \rightarrow 1$ for $\epsilon_f \rightarrow 0$, we obtain
\begin{align}
\label{eq: qomp time final error}
    \Ot \left( \frac{\norm{\vec{s}}^2}{\epsilon_f}\left(T_s + \frac{\alpha}{\gamma}T_A \right) \right).
\end{align}

\subsubsection{Conclusion}
Considering block-encoding access to $D_\Lambda$ from quantum access to $D$ and $\Lambda$ (Theorem~\ref{thm: qomp  block encoding from qa}), we have $T_A = \Ot(T_D + T_\Lambda)$.
Moreover, we have $\alpha = \|D_\Lambda\|_F = \sqrt{k}$, as the matrix $D_\Lambda$ has $k$ non-zero columns of unit $\ell_2$ norm, one per each iteration.
Using these considerations, we can combine Eq.~(\ref{eq: qomp  time atom selection}) and (\ref{eq: qomp time final error}) to conclude the proof of Theorem~\ref{theorem: QOMP iteration cost}:
\begin{align}
    \Ot \left( \sqrt{m} T_{\overline{\Lambda}} +
    \norm{\vec{s}}^2\left( \frac{\sqrt{m}}{\epsilon_i} + \frac{1}{\epsilon_f}\right)\left(T_s + \frac{\sqrt{k}}{\gamma}(T_D + T_\Lambda) \right) \right).
\end{align}
We considered a scenario where all the subroutines succeed. 
To make the iteration succeed with high probability, we can use the \emph{Powering lemma} (Lemma~\ref{lemma: powering lemma (median)}) and the \emph{Union bound} (Theorem~\ref{theorem: union bound}) at some low overhead cost.

\end{document}